\documentclass[preprint,fleqn,authoryear,12pt]{elsarticle}
\date{October 29, 2026}
\usepackage{amsmath,amssymb,mathtools,amsthm,mathrsfs,bm,geometry,physics}
\usepackage{xcolor}
\usepackage{placeins}
\usepackage{float}
\usepackage[
  colorlinks=true,
  linkcolor=blue,
  anchorcolor=blue,
  citecolor=blue,
  urlcolor=blue
]{hyperref}
\usepackage{empheq}
\usepackage{graphicx}
\makeatletter
\g@addto@macro\ps@pprintTitle{%
  \let\XCM@titlefoot\@oddfoot
  \def\@oddfoot{\XCM@titlefoot\unskip}%
  \let\@evenfoot\@oddfoot}
\makeatother

\newtheorem{proposition}{Proposition}

\newtheorem{algorithm}{Algorithm}
\newcommand{\mathbm}[1]{\bm{\mathrm{#1}}}
\newcommand{\trans}{{\scriptscriptstyle\mathsf{T}}}
\newcommand{\bulkquantity}[1]{#1_{\mathrm b}}
\newcommand{\cohesivequantity}[1]{#1_{\mathrm{coh}}}
\newcommand{\crackstrain}[1]{#1}
\newcommand{\normalmode}{\perp}
\newcommand{\tangentialmode}{\parallel}
\newcommand{\measurerestrictionheight}{1.6ex}
\newcommand{\measurerestrictionwidth}{1.0ex}
\newcommand{\measurerestrictionthickness}{0.13ex}
\newcommand{\measurerestrictionglyph}[1]{%
  \hbox{\scalebox{#1}{\hbox{%
    \vrule height \measurerestrictionheight depth 0pt width \measurerestrictionthickness
    \vrule height \measurerestrictionthickness depth 0pt width \measurerestrictionwidth}}}}
\newcommand{\measurerestriction}{\mathbin{\mathchoice
  {\measurerestrictionglyph{1}}{\measurerestrictionglyph{1}}
  {\measurerestrictionglyph{0.7}}{\measurerestrictionglyph{0.5}}}}
\newcommand{\bmx}{\mathbm{x}}
\newcommand{\bmchi}{\mathbm{\chi}}
\newcommand{\bmX}{\mathbm{X}}
\newcommand{\bmu}{\mathbm{u}}
\newcommand{\bmw}{\crackstrain{\mathbm{w}}}
\newcommand{\w}{\crackstrain{w}}

\newcommand{\bmF}{\mathbm{F}} 
\newcommand{\bmFb}{\bulkquantity{\mathbm{F}}} 
\newcommand{\Fb}{\bulkquantity{F}} 
\newcommand{\Jb}{\bulkquantity{J}} 
\newcommand{\bmFreg}{\mathbm{F}_{\mathrm{reg}}} 
\newcommand{\bmFh}{\bmF^{h}} 
\newcommand{\bmFRh}{\bmF_{\mathrm R}^{h}} 
\newcommand{\bmFLh}{\bmF_{\mathrm L}^{h}} 

\newcommand{\bmC}{\mathbm{C}} 
\newcommand{\bmE}{\mathbm{E}} 
\newcommand{\bmR}{\mathbm{R}}
\newcommand{\bmP}{\mathbm{P}} 
\newcommand{\bmS}{\mathbm{S}} 
\newcommand{\bmT}{\mathbm{T}} 
\newcommand{\prescribedquantity}[1]{\bar{#1}}
\newcommand{\bodyforce}{\mathbm{b}}
\newcommand{\boundarytraction}{\prescribedquantity{\bmT}}
\newcommand{\bmPb}{\bulkquantity{\mathbm{P}}}
\newcommand{\bulkInvGradient}[1]{\mathbm{G}_{#1}} 
\newcommand{\bmPgeo}{\mathbm{P}_{\mathrm{geo}}}
\newcommand{\bmKomega}{\mathbm{k}_{\omega}} 
\newcommand{\bmn}{\mathbm{n}}
\newcommand{\bmnh}{\bmn^{h}}
\newcommand{\bmnreg}{\bmn_{\mathrm{reg}}}
\newcommand{\bmm}{\mathbm{m}}

\newcommand{\bmN}{\mathbm{N}}
\newcommand{\bmnu}{\mathbm{\nu}} 

\newcommand{\bmI}{\mathbm{I}} 
\newcommand{\bmIref}{\mathbm{I}_0} 
\newcommand{\tangentprojector}[1]{\bmI_{#1}^{\tangentialmode}}
\newcommand{\bmItan}{\tangentprojector{\bmn}} 
\newcommand{\jump}[1]{\left[\!\left[#1\right]\!\right]}
\newcommand{\jumpu}{\jump{\bmu}}
\newcommand{\positivepart}[1]{{#1}_{+}}
\newcommand{\positivepartexpr}[1]{\positivepart{\left(#1\right)}}
\newcommand{\Gc}{G_{\mathrm c}}
\newcommand{\kperp}{k^{\normalmode}}
\newcommand{\kparallel}{k^{\tangentialmode}}
\newcommand{\Tperpc}{T_{\mathrm c}^{\normalmode}}
\newcommand{\Tparallelc}{T_{\mathrm c}^{\tangentialmode}}
\newcommand{\Deltaperpch}{\Delta_{\mathrm{ch}}^{\normalmode}}
\newcommand{\Deltaparallelch}{\Delta_{\mathrm{ch}}^{\tangentialmode}}
\newcommand{\Deltaperp}{\Delta^{\normalmode}}
\newcommand{\Deltaparallel}{\Delta^{\tangentialmode}}
\newcommand{\wperp}{\w^{\normalmode}}
\newcommand{\wparallel}{\w^{\tangentialmode}}
\newcommand{\bmwparallel}{\bmw^{\tangentialmode}}
\newcommand{\Tperp}{T^{\normalmode}}
\newcommand{\Tparallel}{T^{\tangentialmode}}
\newcommand{\bmTparallel}{\bmT^{\tangentialmode}}
\newcommand{\bmq}{\mathbm{q}}
\newcommand{\bmQ}{\mathbm{Q}} 
\newcommand{\bmqparallel}{\bmq^{\tangentialmode}}
\newcommand{\bmz}{\mathbm{z}} 
\newcommand{\bmzparallel}{\bmz^{\tangentialmode}}
\newcommand{\lambdaperp}{\lambda^{\normalmode}}
\newcommand{\psib}{\bulkquantity{\psi}}
\newcommand{\psibhat}{\bulkquantity{\widehat{\psi}}} 
\newcommand{\psic}{\cohesivequantity{\psi}}
\newcommand{\psired}{\psi_{\mathrm{red}}} 
\newcommand{\psidir}{\psi_{\mathrm{dir}}} 
\newcommand{\bulkmodulus}{\bulkquantity{\mathbb A}} 
\newcommand{\algmodulus}{\mathbb A_{\mathrm{alg}}} 
\newcommand{\localcracktangent}{\mathbm h_{\mathrm{loc}}} 
\newcommand{\localKKTmatrix}{\mathsf M_{\mathrm{loc}}} 
\newcommand{\omegapotential}{\mathcal V_{\omega}} 
\newcommand{\omegatangent}{h_{\omega}} 
\newcommand{\quadratureweight}{W} 
\newcommand{\elementJacobian}{\mathsf J} 
\newcommand{\meshsize}{h_{\mathrm{mesh}}} 
\newcommand{\energyFE}{\energyred^{\mathrm{FE}}} 
\newcommand{\reducedstrength}{f} 
\newcommand{\bmNinit}{\bmN_\star} 
\newcommand{\varphif}{\varphi} 
\newcommand{\omegameasure}{\mu_{\omega}}
\newcommand{\omegareg}{\omega_{\mathrm{reg}}}
\newcommand{\omegasurf}{\omega_{\Gamma}}
\newcommand{\damageband}{\mathcal B_{\omega}^{h}} 
\newcommand{\energy}{\mathcal{E}} 
\newcommand{\firstvariation}[1]{\mathrm D_{#1}} 
\newcommand{\variationalderivative}[2]{\delta_{#2}#1} 
\newcommand{\energyf}{\energy_{\mathrm{frac}}} 
\newcommand{\energyred}{\energy_{\mathrm{red}}} 
\newcommand{\energyczm}{\energy_{\mathrm{CZM}}} 
\newcommand{\energyenvelope}{\energy_{\mathrm{env}}} 
\newcommand{\fractureprofile}{\mathcal A} 
\newcommand{\jumpSet}{J_u} 
\newcommand{\normalizedtraction}{\widehat T} 
\newcommand{\normalizedopening}{\widehat\Delta} 
\newcommand{\normalizedeffectivetraction}{\normalizedtraction_{\mathrm{eff}}}
\newcommand{\normalizedeffectiveopening}{\normalizedopening_{\mathrm{eff}}}
\newcommand{\damagecoupling}{\rho} 
\newcommand{\bandenergy}{\energy^{(3)}} 
\newcommand{\lagrangian}{\mathcal{L}} 
\newcommand{\psilambda}{\psi_{\lambdaperp}}
\renewcommand{\appendixname}{Appendix}
\begin{document}
\begin{frontmatter}
\title{A variational phase-field theory of cohesive fracture at finite deformation}
\author[address1,address2]{Ye~Feng\corref{cor1}}
\address[address1]{School of Aeronautics, Northwestern Polytechnical University, Xi'an, 710072, China}
\address[address2]{National Key Laboratory of Strength and Structural Integrity, Xi'an, 710072, China}
\cortext[cor1]{\textit{E-mail addresses:} fengye@nwpu.edu.cn}
\abstracttitle{Abstract}
\begin{abstract}
This paper formulates a kinematically and variationally consistent phase-field theory for cohesive fracture at finite deformation. The theory introduces an independent compliance-damage field, which accommodates both diffuse distributions and singular concentrations. The phase-field is thereby freed from the dual burden of describing crack geometry and material degradation. A rank-one additive split of the deformation gradient makes separation explicit through independent crack-strain and reference-normal fields. All fields are governed by a single principle of minimum potential energy, which consistently couples nonlinear bulk deformation, crack separation, and crack-plane rotation.  One-dimensional energy comparisons show that minimization favors Dirac-type concentration of compliance damage and crack strain. Strong displacement discontinuities thus emerge naturally at a finite phase-field length without being prescribed a priori, while the phase-field remains continuous. The concentrated energy recovers the exact bulk--interface structure of a sharp-crack cohesive zone model. General objective hyperelastic bulk potentials and admissible mixed-mode cohesive laws can therefore be prescribed independently, with the nominal traction--separation law unaffected by the bulk law or phase-field length. Diffuse damage can coexist with concentrated cracks, retaining deformation modes for bulk-energy relaxation during nucleation. Two- and three-dimensional examples validate the theory across large stretches, finite rotations, and complex crack paths. Notably, the computed near-tip Cauchy stress closely matches the {analytical $1/r$ crack-tip asymptotic solution} of \citet{long2011finite} outside the cohesive zone.
\end{abstract}

\keywordtitle{Keywords}
\end{frontmatter}

\section{Introduction}\label{sec:introduction}
Understanding how solids break has long been a central pursuit of mechanics. At its core lies the challenge of reconciling continuous bulk deformation with discontinuous separation across newly formed interfaces. \citet{griffith1921vi} laid the energetic foundation for brittle fracture by balancing elastic energy release against the cost of creating new surfaces. Although fracture mechanics has since achieved remarkable success in assessing structural safety, the onset of material failure, the selection of crack paths, and the interaction between local separation and overall deformation continue to raise fundamental questions. In soft materials in particular, large stretches, finite rotations, and constitutive nonlinearity near a crack place these questions beyond the scope of linear elasticity~\citep{bouchbinder2014dynamics}.

The role of nonlinearity in fracture extends beyond a correction to the near-tip fields. \citet{buehler2003hyperelasticity} showed that even a hyperelastic zone confined to the crack tip can govern local energy transport and crack speed. High-speed crack oscillations~\citep{chen2017instability} and supershear tensile propagation~\citep{wang2023tensile} further illustrate the richness of finite-deformation fracture dynamics. At the onset of fracture, recent studies of soft elastic solids link crack nucleation to instabilities involving large rotations and elastic softening~\citep{riccobelli2024elastic}. Together, these findings place the coupling among geometry, bulk response, damage, and separation at the center of finite-deformation fracture. Capturing this coupling calls for a theory that respects finite-deformation kinematics and gives local separation a clear physical meaning. This paper addresses quasistatic cohesive fracture as a starting point for this broader study.

Cohesive fracture theory provides a direct physical link between bulk deformation and crack separation. \citet{barenblatt1962mathematical} introduced cohesive forces acting near the crack edge to account for the local resistance to separation and its energetic cost. \citet{simo1993analysis}, \citet{armero1996analysis}, and \citet{oliver2003strong} established how strain concentration in a continuum relates to displacement jumps and how strong discontinuities enter finite-deformation kinematics. In this tradition, crack opening and sliding have their own kinematics and constitutive laws, while bulk deformation and interfacial separation remain coupled through traction balance and energy transfer. Preserving this kinematic and constitutive structure while allowing a field description to determine where cracks nucleate and how they evolve remains an enduring challenge in fracture mechanics.

Variational fracture and phase-field methods opened a new route to this challenge. \citet{FRANCFORT19981319} recast Griffith's energetic view as a minimization problem over deformation and an unknown crack set, bringing nucleation and path selection into a single variational description. By approximating the sharp crack set with a continuous field, \citet{bourdin2000numerical} obtained a formulation amenable to finite element computation. Cracks can thus nucleate, deflect, branch, and coalesce without prescribed paths. The significance of phase-field methods extends beyond a smooth representation of crack geometry: the coupling and energetic competition between fields can determine both bulk deformation and the loss of material integrity.

Cohesive phase-field models have made substantial progress at small strain. \citet{lorentz2011convergence} established a cohesive limit with distinct material and regularization scales; prescribed separation laws can be recovered through constitutive constructions~\citep{wu2017unified,feng2021endowing}. \citet{alessi2026gamma} proved one-dimensional $\Gamma$-convergence~\citep{dalmaso1993introduction,braides2002gamma} for a class including the model of \citet{feng2021endowing}. \citet{CONTI20161033,conti2024vectorial} established rigorous connections to cohesive energies on displacement jumps in multiple dimensions, including geometrically nonlinear elasticity.

Finite deformation raises a further mechanical difficulty. As a crack opens, the surrounding material unloads according to its nonlinear bulk response, while the separating faces resist opening and sliding according to an interface law. \citet{ortiz1999finite} used the mean deformation of the two crack faces to define the current interface and its spatial normal, giving opening and sliding a precise meaning on a rotating surface. In a field description, such as a phase-field model, neither the crack faces nor their orientation is available in advance. Separation and interface geometry must emerge together, and changes in orientation must enter the energy variation and stress. Replacing a small-strain bulk potential with a hyperelastic one does not by itself secure these relations. The challenge is to make nonlinear bulk deformation, separation, and crack rotation mechanically consistent within a common variational description.

To retain field-based crack tracking while making separation explicit, variational and nonlocal descriptions have introduced structured strains, eigendeformations, and auxiliary opening fields~\citep{FREDDI20101154,schmidt2009eigenfracture,verhoosel2013phase,ren2024physically}. In particular, the infinitesimal crack strain of \citet{feng2023unified} retains both crack direction and separation information. In diffuse formulations, however, distributed crack strain ties interface kinematics to the phase-field profile. At finite deformation, the nonlinear bulk response also couples to separation throughout the band, obstructing the independent specification of bulk and interfacial laws.

Can a field theory retain the crack-tracking and computational advantages of a diffuse phase-field while describing separation with the kinematic and energetic clarity of a sharp interface? Keeping the phase-field length finite would facilitate the numerical resolution of crack geometry, while strain concentration would confine the separation coupling to a lower-dimensional interface---a point in one dimension---rather than spreading it across the regularization band.

\citet{feng2025JMPS} offered a route by proposing a small-strain phase-field model, termed the $\Omega^2$ model. An independent compliance-damage variable relieves the phase-field of the dual burden of representing both crack geometry and material degradation. One-dimensional solutions exhibit Dirac-type concentrations of compliance damage and crack strain, producing genuine displacement jumps while the phase-field remains continuous. Interestingly, independent efforts to prescribe convex strength surfaces led to strong-discontinuity formulations almost concurrently~\citep{vicentini2025variational,bourdin2025variational}.\footnote{The small-strain $\Omega^2$ model \citep{feng2025JMPS} was submitted on June 1, 2025; the manuscript of \citet{vicentini2025variational} was submitted on June 13, 2025; and the preprint of \citet{bourdin2025variational} was made public on June 27, 2025.} 

These developments suggest that the time is ripe to seek a field theory of bulk deformation and interfacial degradation, grounded in general finite-deformation kinematics and a common variational principle.
Accordingly, this paper formulates a variational phase-field theory for cohesive fracture at finite deformation. The model regularizes strong-discontinuity kinematics~\citep{armero1996analysis,oliver2003strong} with independent crack-strain and reference-normal fields, together with independent compliance damage and phase-field variable. Mechanical equilibrium, crack evolution, and crack-plane rotation follow from a single energy minimization under irreversibility and non-interpenetration constraints. Notably, normal transport via Nanson's relation generates a geometric first Piola--Kirchhoff stress upon energy variation, preserving Cauchy-stress symmetry and angular momentum balance. Direct energy comparisons within a one-dimensional single-crack family show that a narrower damage band has lower energy. Energy minimization thus favors Dirac-type concentration of crack strain and compliance damage, producing a sharp displacement jump. The concentrated energy takes exactly the bulk--interface form of a sharp-crack cohesive zone model (CZM), allowing independent specification of general objective hyperelastic bulk potentials and cohesive softening laws within the adopted mixed-mode formulation. Two- and three-dimensional examples examine large stretches, rotations, and complex crack evolution. Notch-induced nucleation reveals diffuse damage coexisting with concentrated fracture, consistent with bulk-energy relaxation in cohesive fracture models~\citep{Braides1995,rodella2026sharp,bourdin2026strength}. In particular, the computed Cauchy stress closely matches the $1/r$ asymptotic field of \citet{long2011finite} outside the cohesive zone. The connection to nonlinear fracture mechanics thus reaches the local stress field as well as the energy structure.

Sections~\ref{sec:classical-kinematics}--\ref{sec:one-dimensional-analysis} develop the kinematics, model, and analysis; Sections~\ref{sec:numerical-implementation}--\ref{sec:conclusions} present algorithm, implementation, examples, and conclusions. Scalars are italic, vectors and second-order tensors bold upright, and fourth-order tensors double-struck. The symbols $\cdot$, $:$, and $\otimes$ denote the inner product, double contraction, and tensor product, respectively; the superscript $\mathsf T$ denotes transposition. The gradient $\nabla\equiv\nabla_{\bmX}$, divergence $\mathrm{Div}$, and Laplacian $\nabla^2$ act in the reference configuration.

\section{Classical strong-discontinuity kinematics}\label{sec:classical-kinematics}
We briefly review the classical strong-discontinuity kinematics developed by \citet{armero1996analysis} and \citet{oliver2003strong}. The present model adopts a distinct but closely related kinematic description, as discussed in Section~\ref{sec:fields}.

Consider the deformation and fracture of a solid $\mathcal B_0\subset\mathbb R^d$ ($d=1,2,3$) over the time interval $[0,T]$. The $(d-1)$-dimensional crack set in the reference configuration is denoted by $\Gamma_0(t)\subset\mathcal B_0$ and satisfies the irreversibility condition
\begin{equation}
	\Gamma_0(t+\Delta t)\supseteq\Gamma_0(t),
\end{equation} for all $t$ and $\Delta t$ such that $0\leq t<t+\Delta t\leq T$.
The oriented unit normal to the crack surface is $\bmN_0(\bmX)$, $\bmX\in\Gamma_0$. We use a total Lagrangian description. Unless otherwise stated, we omit the reference-configuration subscript and write $\bmN$, $\mathcal B$, and $\Gamma$ for $\bmN_0$, $\mathcal B_0$, and $\Gamma_0$, respectively.

Away from the crack surface, material points move according to
\begin{equation}
  \bmx  =\bmchi(\bmX,t) =  \bmX + \bmu(\bmX,t),\quad \bmX\in\mathcal{B}\setminus\Gamma,
\end{equation}
where $ \bmx $ is the spatial position of the material point $ \bmX $ at time $ t $, and $ \bmu $ is the displacement field.
For a piecewise-smooth deformation and a smooth crack surface, the limiting displacements on the two sides define the jump\footnote{Distributional and measure-theoretic definitions are given in \ref{app:distributional-notation}.}
\begin{equation}
\jumpu\!(\bmX,t) \coloneqq  \lim\limits_{\epsilon\to0^+}\left[\bmu(\bmX+\epsilon\bmN,t)-\bmu(\bmX-\epsilon\bmN,t)\right],\quad\bmX\in\Gamma.
\end{equation}

For this piecewise-smooth deformation, the distributional gradient of $\bmchi$ admits the additive decomposition
\begin{equation}
	\bmF\coloneqq\nabla\bmchi
	=
	\bmFreg+\jumpu\otimes\bmN\,\delta_{\Gamma}.
	\label{eq:strong-discontinuity-additive}
\end{equation}
Here $\bmFreg$ is the regular part of the distributional gradient. The surface Dirac distribution $\delta_\Gamma$ is supported on the crack surface. Equation~\eqref{eq:strong-discontinuity-additive} is understood in the sense of distributions; its definition and precise measure representation are given in \ref{app:distributional-notation}.

\citet{oliver2003strong} described strain localization in a band $\mathcal B_h$ of thickness $h$ containing the crack surface $\Gamma$. The surface Dirac distribution is regularized as $\delta_\Gamma^h$:
\begin{equation}\label{eq:classical-finite-width}
 \delta_\Gamma^h\coloneqq\frac1h\mathbf1_{\mathcal B_h},
 \qquad
 \bmFh\coloneqq\bmFreg+\jumpu\otimes\bmN\,\delta_\Gamma^h,
\end{equation}
where $\mathbf1_{\mathcal B_h}$ is the indicator function of the band, and the displacement jump and normal are extended into the band along its normals. As $h\to0^+$, $\delta_\Gamma^h$ converges to $\delta_\Gamma$ in the sense of distributions, linking finite-width localization to sharp-crack kinematics.

For orientation-preserving states with $\det\bmFh>0$ and $\det\bmFreg>0$, the rank-one inverse identity gives identical transported crack normals:
\begin{equation}\label{eq:classical-normal-transport}
 \bmnh
 =\frac{(\bmFh)^{-\trans}\bmN}{\lVert(\bmFh)^{-\trans}\bmN\rVert}
 =\frac{\bmFreg^{-\trans}\bmN}{\lVert\bmFreg^{-\trans}\bmN\rVert}
 =\bmnreg.
\end{equation}
The multiplicative representations and the derivation of normal transport are given in \ref{app:classical-multiplicative-kinematics}. Equation~\eqref{eq:classical-normal-transport} underlies the definition of the current crack normal in Section~\ref{sec:fields}.

\section{The finite-deformation
\texorpdfstring{$\Omega^2$}{Omega2} model}\label{sec:model}
This section develops the kinematics and variational formulation of the finite-deformation $\Omega^2$ model and derives the mechanical equilibrium and damage equations. The energy is defined for ordinary function fields, and its variation determines the crack strain $\bmw$, reference crack normal $\bmN$, and damage fields $\omega,\phi$. Localization and concentration are examined subsequently in Section~\ref{sec:one-dimensional-analysis}.

\subsection{Kinematic and damage fields}\label{sec:fields}
Consider a solid with reference configuration $\mathcal{B}$ and motion $\bmchi(\bmX,t)$. Its deformation gradient is decomposed additively as
\begin{equation}\label{eq:F}
\bmF\coloneqq\nabla\bmchi
	= \bmFb + \bmw\otimes\bmN,
\end{equation}
where $\bmF$ and $\bmFb$ are two-point tensors, and $\bmFb$ is the bulk deformation gradient. The spatial vector field $\bmw$ represents crack opening and sliding per unit reference length and thus has the interpretation of a strain. The unit material vector field $\bmN$ denotes the normal to a candidate crack plane in the reference configuration.\footnote{The fields $\bmw$ and $\bmN$ are defined throughout the body and determined variationally. Their rank-one contribution recovers the singular kinematics of Eq.~\eqref{eq:strong-discontinuity-additive} in the concentration limit; see \ref{app:model-concentration} and Section~\ref{sec:one-dimensional-analysis}.}
Following the normal-transport relation~\eqref{eq:classical-normal-transport}, we define the corresponding unit normal in the current configuration by
\begin{equation}\label{eq:current-crack-normal}
	\bmn=
	\frac{\bmF^{-\trans}\bmN}
	{\lVert\bmF^{-\trans}\bmN\rVert}=
	\frac{\bmFb^{-\trans}\bmN}
	{\lVert\bmFb^{-\trans}\bmN\rVert},
\end{equation}
where $\bmn$ is a spatial vector.

The crack strain $\bmw$ and reference crack normal $\bmN$ are independent local variables. Together with the total deformation gradient, they determine the bulk gradient $\bmFb=\bmF-\bmw\otimes\bmN$.
For a given $ \bmF $, the admissible fields $\bmw $ and $ \bmN $ satisfy two physical constraints:
\begin{equation}\label{eq:unp}
	\Jb\coloneqq\det\bmFb>0,\qquad \bmw\cdot\bmn\geq0,
\end{equation}
which enforce local orientation preservation of the bulk deformation and non-interpenetration across the crack interface, respectively.

Under the constraints~\eqref{eq:unp}, the matrix determinant lemma applied to Eq.~\eqref{eq:F} gives
\begin{equation}
	J\coloneqq\det\bmF
	=
	\Jb
	\left(1+\bmN\cdot\bmFb^{-1}\bmw\right)
	=
	\Jb
	\left[
		1+
		\lVert\bmFb^{-\trans}\bmN\rVert
		(\bmw\cdot\bmn)
	\right]
	\geq \Jb >0.
\end{equation}
Thus the total deformation also preserves orientation locally.

The model uses two irreversible damage fields:
\begin{equation}\label{eq:omega_split}
	\omega(\bmX,t)\in[0,+\infty),
	\qquad
	\phi(\bmX,t)\in[0,1],
	\qquad
	\dot{\omega}\geq0,
	\qquad
	\dot{\phi}\geq0,
\end{equation}
where a superposed dot denotes a time derivative. As in classical phase-field fracture, $\phi$ is continuous and regularizes the crack surface through a gradient term. Compliance damage $\omega$ is a nonnegative internal variable describing the growth of local compliance; it is treated as an ordinary function field in the main formulation. Its measure representation accommodates both a diffuse volume part and a concentrated interfacial part. The function spaces and detailed representation are given in \ref{app:model-concentration}.

\subsection{Energy densities and variational formulation}\label{sec:energy}

The total recoverable strain-energy density is
\begin{equation}\label{eq:split}
	\psi(\bmF,\bmw,\bmN,\omega)
	\coloneqq
	\psib(\bmFb)
	+
	\psic\bigl(\bmw,\bmn(\bmF,\bmN),\omega\bigr),
\end{equation}
where $\psib$ and $\psic$ are the recoverable bulk and cohesive energy densities, respectively. The current normal $\bmn(\bmF,\bmN)$ is defined by Eq.~\eqref{eq:current-crack-normal}.

The energy split~\eqref{eq:split} permits independent bulk and interfacial constitutive descriptions; its small-strain form is discussed by \citet{feng2023unified}. The independence concerns constitutive inputs; the effective bulk response can still depend on cohesive strength through minimization over crack strain and compliance damage. The resulting relaxation mechanism is illustrated in the small-strain setting by \citet[Appendix~C]{feng_Li2026}.

The bulk energy $\psib$ in \eqref{eq:split} may be any objective hyperelastic potential, satisfying
\begin{equation}\label{eq:bulk-objectivity}
 \psib(\bmR\bmFb)=\psib(\bmFb)
 \qquad\text{for all}\ \bmR\in\mathrm{SO}(d),
\end{equation}
where $d$ is the spatial dimension and $\mathrm{SO}(d)$ is the corresponding group of proper rotations.  \ref{app:bulk-constitutive} lists energies and bulk stresses for Neo--Hookean, generalized Neo--Hookean, Mooney--Rivlin, and Saint Venant--Kirchhoff materials.

To describe mixed-mode opening and sliding, we introduce the orthogonal projector onto the current crack plane:\footnote{For any spatial vector $\bmz$, its tangential component is
	$
		\bmzparallel\coloneqq\bmItan\bmz
		=\bmz-(\bmz\cdot\bmn)\bmn
 $, the orthogonal projection onto the current crack plane. It satisfies $\bmzparallel\cdot\bmn=0$.}
\begin{equation}
\bmItan\coloneqq\bmI-\bmn\otimes\bmn,
\end{equation}
where $\bmI$ is the spatial identity tensor. The normal and tangential components of $\bmw$ are
\begin{equation}\label{eq:wn}
 \wperp\coloneqq\bmw\cdot\bmn\geq0,
 \qquad
 \bmwparallel\coloneqq\bmItan\bmw.
\end{equation}
The normal opening strain $\wperp$ is nonnegative by the non-interpenetration constraint~\eqref{eq:unp}. The vector $\bmwparallel$ describes relative tangential sliding across the crack; its magnitude is denoted by
$
	\wparallel\coloneqq\lVert\bmwparallel\rVert.
$
Superscripts $\normalmode$ and $\tangentialmode$ denote components normal and tangential to the current crack plane, respectively.

Using the opening--sliding decomposition~\eqref{eq:wn}, we adopt the quadratic cohesive potential
\begin{equation}\label{eq:psic}
	\psic(\bmw,\bmn,\omega)=
	\frac{1}{2}\bmw\cdot\bmKomega(\bmn)\bmw	=
	\frac{\ell\kperp}{2\omega}
	\left[(\wperp)^2+\frac{(\wparallel)^2}{\beta^2}\right]
	,\quad\omega>0,
\end{equation}
where $\ell>0$ is the phase-field length in the reference configuration, $\kperp>0$ is the normal cohesive stiffness parameter, and $\beta>0$ is the pure-mode tensile-to-shear strength ratio. The second-order cohesive stiffness tensor $\bmKomega$ is defined in the current configuration by
\begin{equation}
	\bmKomega(\bmn)
	=
	\frac{\ell\kperp}{\omega}
	\left[
	\bmn\otimes\bmn
	+
	\frac{1}{\beta^2}
	\bmItan
	\right].
\end{equation}
The cohesive stiffness decreases as $\omega$ increases.

At $\omega=0$, Eq.~\eqref{eq:psic} is extended lower semicontinuously as
\begin{equation}\label{eq:psic-omega-zero}
	\psic(\bmw,\bmn,0)
	\coloneqq
	\begin{cases}
		0, & \bmw=\mathbm{0},\\
		+\infty, & \bmw\neq\mathbm{0}
	\end{cases}.
\end{equation}
Crack opening and tangential sliding are therefore absent before cohesive damage is activated, consistent with the initially rigid extrinsic cohesive model of \citet{camacho1996computational}.

The normal cohesive stiffness parameter $\kperp$ in Eq.~\eqref{eq:psic} is specified in terms of the critical pure-normal nominal traction and characteristic separation:
\begin{equation}
	\kperp
	\coloneqq
	\frac{\Tperpc}{\Deltaperpch},
	\qquad
	\Deltaperpch
	\coloneqq
	\frac{2\Gc}{\Tperpc},
\end{equation}
where $\Tperpc>0$ is the critical pure-normal nominal traction, $\Gc>0$ is the fracture-energy parameter, and $\Deltaperpch$ is the characteristic normal separation.
The corresponding tangential parameters follow from the pure-mode tensile-to-shear strength ratio $\beta$:
\begin{equation}
	\beta
	\coloneqq
	\frac{\Tperpc}{\Tparallelc},\qquad
	\Tparallelc
	=
	\frac{\Tperpc}{\beta},
	\qquad
	\Deltaparallelch
	\coloneqq
\frac{2\Gc}{\Tparallelc}
	=
	\beta\Deltaperpch,
	\qquad
	\kparallel
	\coloneqq
	\frac{\Tparallelc}{\Deltaparallelch}
		=
		\frac{\kperp}{\beta^2},
\end{equation}
where $\Tparallelc$ is the critical pure-tangential nominal traction, while $\Deltaparallelch$ and $\kparallel$ are the characteristic tangential separation and tangential cohesive stiffness.

The irreversible Griffith fracture-energy density in the finite-deformation $ \Omega^2 $ model is defined as
\begin{equation}\label{eq:fracture-energy-density}
	\varphif(\phi,\nabla\phi,\omega)
	=
	\frac{\Gc}{2\ell}
	\left[
		\Omega(\phi)^2+\ell^2\lvert\nabla\phi\rvert^2
	\right]
	+
	\frac{\Gc}{\ell}(1-\phi)^2\omega.
\end{equation}
The function $\Omega$ satisfies the following monotonicity, convexity, and normalization conditions on $0\leq\phi<1$:\footnote{Derivatives are taken on smooth branches, with one-sided limits at endpoints. For piecewise-smooth functions, the third condition means global convexity of $\Omega^2$.}
\begin{equation}\label{eq:Omega-conditions}
 \Omega(0)=0,\qquad
  \Omega'(\phi)>0,\qquad
  \left[\Omega(\phi)^2\right]''\geq0,
\qquad
 2\int_0^1\Omega(\phi)\,\dd \phi=1.
\end{equation}
The first two conditions ensure that the local Griffith term vanishes in the intact state and increases monotonically with $\phi$. The third ensures convexity in $\phi$. Normalization sets the work of complete separation of a concentrated cohesive interface to $\Gc$ per unit reference crack area. Following \citet{feng2021endowing}, we determine $\Omega$ from the target traction--separation law, grounding the characteristic function in a measurable constitutive relation; see Section~\ref{sec:cohesive-law}.

Integrating the recoverable strain energy and irreversible fracture energy over the reference configuration and adding the potential of the external loads gives
\begin{equation}\label{eq:E}
	\energy[\bmchi,\bmw,\bmN,\omega,\phi]
	=
	\int_{\mathcal{B}}
		\varphif(\phi,\nabla\phi,\omega)\, \dd V
		+
	\int_{\mathcal{B}}
		\psi(\bmF,\bmw,\bmN,\omega)\,\dd V
	+
	\mathcal{P}_{\mathrm{ext}}[\bmchi].
\end{equation}
For prescribed, deformation-independent body forces and nominal boundary tractions, the external potential $\mathcal{P}_{\mathrm{ext}}[\bmchi]$ is
\begin{equation}
	\mathcal{P}_{\mathrm{ext}}[\bmchi]
	=
	-\int_{\mathcal{B}}
\bodyforce\cdot\bmchi\,\dd V
	-\int_{\partial \mathcal{B}_\mathrm{N}}
	\boundarytraction\cdot\bmchi\,\dd A,
\end{equation}
where $\bodyforce$ is the prescribed spatial body force per unit reference volume, and $\boundarytraction$ is the prescribed spatial nominal traction per unit reference area on the Neumann boundary $\partial\mathcal{B}_\mathrm{N}$.

Quasistatic deformation and fracture are governed by the principle of minimum potential energy subject to irreversibility:
\begin{equation}\label{eq:min}
\boxed{\{\bmchi,\bmw,\bmN,\omega,\phi\}_t
	\in
	\operatorname*{arg\,min}_{\mathcal{A}_t}\,
	\energy[\bmchi,\bmw,\bmN,\omega,\phi].}
\end{equation}
The admissible set $\mathcal A_t$ incorporates the prescribed essential boundary conditions, the kinematic constraints~\eqref{eq:unp}, and $\lVert\bmN\rVert=1$. After time discretization, damage irreversibility requires $0\leq\omega_{\mathrm{prev}}\leq\omega$ and $0\leq\phi_{\mathrm{prev}}\leq\phi\leq1$, where the subscript $\mathrm{prev}$ denotes the previous accepted loading step.

\subsection{Mechanical equilibrium and crack direction under kinematic constraints}\label{sec:mechanics}
With the damage fields fixed, constrained variations in $\bmchi$, $\bmw$, and $\bmN$ give global equilibrium, local traction balance, and the crack-direction condition. The non-interpenetration constraint~\eqref{eq:unp} couples deformation and crack direction through the current normal $\bmn(\bmF,\bmN)$; its contribution must be included consistently in all three equations.

We now derive the first-order necessary conditions for the constrained minimization~\eqref{eq:min} at fixed $\omega$ and $\phi$. We introduce the Lagrange multiplier $\lambdaperp\leq0$ for non-interpenetration and define\footnote{Variations are taken within the orientation-preserving domain $\Jb>0$.}
\begin{equation}\label{eq:contact-lagrangian}
	\lagrangian
	[\bmchi,\bmw,\bmN,\omega,\phi;\lambdaperp]
	\coloneqq \energy +
	\int_{\mathcal{B}}
		\lambdaperp\wperp(\bmF,\bmw,\bmN)\,\dd V,
\end{equation}
where $\energy$ is the total potential~\eqref{eq:E} and $\wperp=\bmw\cdot\bmn$ is the normal opening strain defined in Eq.~\eqref{eq:wn}. The multiplier $\lambdaperp$ represents the nominal normal contact stress, negative in compression, as discussed in Section~\ref{sec:local-traction}.

The Karush--Kuhn--Tucker (KKT) necessary conditions are derived below on a smooth branch with $\omega>0$ \citep{bertsekas1999nonlinear}. Relations at $\omega=0$ are discussed in Section~\ref{sec:strength-surface}.

\subsubsection{Piola stress and global equilibrium}\label{sec:Piola-stress}
The KKT stationarity condition for $\bmchi$ is $\delta_{\bmchi}\lagrangian=0$, where $\delta\bmchi$ satisfies homogeneous essential boundary conditions. Under the smoothness assumptions above, integration by parts gives the equilibrium equations in the reference configuration:
\begin{subequations}\label{eq:mechanical-equilibrium}
	\begin{alignat}{2}
		&\mathrm{Div}\,\bmP + \bodyforce = \mathbm{0}
		&\qquad& \text{in }\,\mathcal{B},
		\label{eq:momentum-balance}\\[0.6em]
		&\bmP\bmnu = \boundarytraction
		&\qquad& \text{on }\partial\mathcal{B}_{\mathrm N},
		\label{eq:traction-boundary}
	\end{alignat}
\end{subequations}
where $\bmnu$ is the outward unit normal to the reference boundary and $\bmP$ is the total first Piola--Kirchhoff stress:
\begin{equation}\label{eq:first-Piola-stress}
	\boxed{\bmP
		\coloneqq \frac{\partial}{\partial\bmF}\left(\psi + \lambdaperp\wperp\right) =
		\bmPb + \bmPgeo.}
\end{equation}
The derivative with respect to $\bmF$ in Eq.~\eqref{eq:first-Piola-stress} holds $\bmw$, $\bmN$, $\omega$, and $\lambdaperp$ fixed. It retains the dependence of $\bmn(\bmF,\bmN)$ on $\bmF$. The term $\bmPb$ is the bulk first Piola--Kirchhoff stress. The geometric Piola stress $\bmPgeo$ arises from changes in the direction field $\bmn(\bmF,\bmN)$. Their expressions are derived in \ref{app:geometric-Piola-stress}:
\begin{subequations}\label{eq:stresses-breakdown}
	\begin{align}
		&\bmPb
		\coloneqq
		\frac{\partial\psib}{\partial\bmFb}, \\[0.6em]
		&\bmPgeo
		\label{eq:geometric-Piola-stress}
		\coloneqq \frac{\partial}{\partial\bmF}\left(\psic + \lambdaperp\wperp\right) 
		= \bmq \cdot \frac{\partial\bmn}{\partial\bmF} 
		= -\bmn \otimes \left(\bmF^{-1}\bmqparallel\right).
	\end{align}
\end{subequations}
The generalized spatial force $\bmq$ and its tangential projection $\bmqparallel$ are defined by
\begin{equation}\label{eq:q-vectors}
	\bmq
	\coloneqq \frac{\partial}{\partial\bmn}\left(\psic + \lambdaperp\wperp\right)
	= \left[ \frac{\ell(\kperp-\kparallel)}{\omega}\wperp + \lambdaperp \right]\bmw,
	\qquad		
	\bmqparallel
	\coloneqq \bmItan\bmq.
\end{equation}
In Eq.~\eqref{eq:q-vectors}, we first take $\partial_{\bmn}$ by treating $\bmn$ as an independent vector in the cohesive and contact terms. We then substitute the transport relation. Since $\lVert\bmn\rVert=1$, only the tangential projection $\bmqparallel$ contributes to the work associated with changes in the normal. Fixing $\bmN$ does not fix the current normal $\bmn$; the geometric stress can therefore remain nonzero during motion.

\subsubsection{Local traction balance and crack opening/closure}\label{sec:local-traction}
At fixed $\bmF$ and $\bmN$, the normal $\bmn=\bmn(\bmF,\bmN)$ is independent of $\bmw$, so that $\partial_{\bmw}\bmn=\mathbm0$. We hold the other independent fields and the contact multiplier fixed. KKT stationarity in $\bmw$, namely $\delta_{\bmw}\lagrangian=0$, then gives the local traction balance
\begin{equation}\label{eq:w-equilibrium}
\boxed{\bmKomega\bmw +\lambdaperp \bmn = \bmPb\bmN.}
\end{equation}
Equation~\eqref{eq:w-equilibrium} balances nominal tractions across the crack plane. The left-hand side combines recoverable cohesive traction and nominal normal contact stress; the right-hand side is the nominal traction exerted by the bulk material on that plane.

The KKT feasibility and complementarity conditions for non-interpenetration are
\begin{equation}\label{eq:KKT_lambda}
	\boxed{\wperp\geq0,\qquad\lambdaperp\leq0,\qquad \lambdaperp\wperp = 0.}
\end{equation}
This complementarity relation allows a nonzero contact multiplier only when the normal opening vanishes.

Since $\bmqparallel\cdot\bmn=0$ and $\bmF^{-\trans}\bmN\parallel\bmn$ imply $\bmPgeo\bmN=\mathbm0$, the total nominal traction on a plane with reference normal $\bmN$ satisfies
\begin{equation}\label{eq:traction-equilibrium}
 \bmT\coloneqq\bmP\bmN
 =\bmPb\bmN
 =\bmKomega\bmw+\lambdaperp\bmn.
\end{equation}
Equation~\eqref{eq:traction-equilibrium} ensures agreement between bulk and cohesive tractions. The geometric stress does not alter the nominal traction on this plane, but it contributes to global equilibrium and angular momentum balance.

To distinguish tensile opening from compressive contact and tangential sliding, we resolve the nominal traction relative to the current crack plane:
\begin{equation}\label{eq:traction-components}
 \Tperp\coloneqq\bmn\cdot\bmT,
 \qquad\bmTparallel\coloneqq\bmItan\bmT,
 \qquad\Tparallel\coloneqq\lVert\bmTparallel\rVert.
\end{equation}
Here $\Tperp$ is the normal nominal traction, $\bmTparallel$ is the tangential nominal traction vector, and $\Tparallel$ is its magnitude. For $\omega>0$, traction balance~\eqref{eq:traction-equilibrium} and contact complementarity~\eqref{eq:KKT_lambda} give the opening and closure conditions:
	\begin{itemize}
\item Open crack: $\wperp>0\iff\Tperp>0\implies\lambdaperp=0$;
\item Closed crack: $\wperp=0\iff\Tperp\leq0\implies\lambdaperp=\Tperp\leq0$.
	\end{itemize}
Thus $\lambdaperp$ is precisely the nominal normal contact stress on the closed branch.\footnote{Nanson's formula gives the current-to-reference area ratio $J\lVert\bmF^{-\trans}\bmN\rVert$. The normal contact stress per unit current area on the closed branch is therefore $\lambdaperp/[J\lVert\bmF^{-\trans}\bmN\rVert]$, again negative in compression.}

To establish angular momentum balance, we examine the symmetry of the total Cauchy stress. Objectivity of the constrained potential density $\psilambda\coloneqq\psi+\lambdaperp\wperp$, together with local equilibrium in $\bmw$, gives
\begin{equation}\label{eq:cauchy-symmetry}
 \mathbm{\sigma}\coloneqq J^{-1}\bmP\bmF^{\trans}
 =\mathbm{\sigma}^{\trans}.
\end{equation}
Local angular momentum balance therefore holds. The geometric term $\bmPgeo$ is essential to this symmetry; omitting it generally destroys the symmetry of the total Cauchy stress. A proof is given in \ref{app:angular-momentum}.

\subsubsection{Variational selection of crack direction}\label{sec:crack-direction}
With the other fields fixed, the reference crack normal $\bmN$ is selected by constrained energy minimization:\footnote{The material vector $\bmQ=\partial_{\bmN}(\psi+\lambdaperp\wperp)$ is conjugate to the reference direction and differs from the spatial vector $\bmq$ in Eq.~\eqref{eq:q-vectors}. Its explicit expression is given in Eq.~\eqref{eq:direction-conjugate-force}. At a smooth local minimum, $\bmN$ must satisfy
	$
	(\bmIref-\bmN\otimes\bmN)
	\bmQ
	=
	\mathbm{0}.
	$
Here $\bmIref$ is the reference identity tensor. At a stationary direction, the configurational force conjugate to $\bmN$ is normal to the unit sphere.}
\begin{equation}\label{eq:N-direct-minimization}
	\bmN
	\in
	\operatorname*{arg\,min}_{\bmN\in
	\mathcal{A}_{\bmN}(\bmF,\bmw)}
	\psi(\bmF,\bmw,\bmN,\omega),
\end{equation}
where the admissible set of candidate directions is
\begin{equation}
	\mathcal{A}_{\bmN}(\bmF,\bmw)
	\coloneqq
	\left\{
		\bmN\in\mathbb{S}^{d-1}
	\,\middle|\,
	\wperp(\bmF,\bmw,\bmN)\geq0,\ 
	\det(\bmF-\bmw\otimes\bmN)>0
	\right\},
\end{equation}
and $\mathbb{S}^{d-1}\coloneqq
\{\mathbm{a}\in\mathbb{R}^{d}\mid\lVert\mathbm{a}\rVert=1\}$ is the unit sphere. The use of minimum potential energy to determine crack direction in phase-field modeling can be traced to \citet{feng2023unified}.

From a multiscale viewpoint, crack strain can first be equilibrated locally at each candidate direction, leaving an energy in $\bmF$, $\bmN$, and $\omega$. Minimizing this energy over $\bmN$ preserves the same minimum-energy principle while simplifying both theoretical analysis and subsequent numerical implementation. For fixed $\bmF$ satisfying $\det\bmF>0$ and fixed $\omega$, we therefore define
\begin{equation}\label{eq:w-reduced-energy}
\psired(\bmF,\bmN,\omega)
\coloneqq
\min_{
\bmw\in\mathcal{A}_{\bmw}(\bmF,\bmN)}
\left\{
\psib(\bmF-\bmw\otimes\bmN)
+
\psic(\bmw,\bmn,\omega)
\right\},
\end{equation}
where the admissible set for a candidate $ \bmw $ is
\[
\mathcal{A}_{\bmw}(\bmF,\bmN)
\coloneqq
\left\{
\bmw\in\mathbb{R}^{d}
\,\middle|\,
\wperp(\bmF,\bmw,\bmN)\geq0,\ 
\det(\bmF-\bmw\otimes\bmN)>0
\right\}.
\]

The joint minimization then becomes\footnote{If the inner minimum is attained and $\psired$ is finite and lower semicontinuous in $\bmN$, compactness of the unit sphere guarantees an optimal direction.}
\begin{equation}\label{eq:N-reduced-minimization}
 \bmN\in\operatorname*{arg\,min}_{\bmN\in\mathbb S^{d-1}}
 \psired(\bmF,\bmN,\omega).
\end{equation}
Two aspects of direction selection require clarification. First, $\bmN$ represents the reference normal of the currently active fracture microplane; it need not coincide with the geometrical normal of an already formed macroscopic crack. The active direction can therefore be updated by energy minimization after damage initiation, while irreversibility is imposed on $\omega$ and $\phi$. Second, in the intact state $\omega=\phi=0$, $\bmw=\mathbm0$ makes the energy independent of $\bmN$. The initiation direction is then determined from the first-order one-sided limit of Eq.~\eqref{eq:N-reduced-minimization} as $\omega\to0^+$, as discussed in Section~\ref{sec:strength-surface}.

Unlike the small-strain linear-elastic case \citep{feng2023unified,feng20253d}, an optimal normal generally has no closed-form expression at finite deformation and must be found numerically. Section~\ref{sec:direction-search} gives the search procedure. Constitutive structure can nevertheless simplify the search substantially: the following proposition reduces a class of three-dimensional problems to a single angle in a principal plane.
\begin{proposition}[Principal-plane reduction of three-dimensional crack directions]\label{prop:crack-direction-plane-reduction}
Let $d=3$, and fix a deformation gradient $\bmF$ satisfying $\det\bmF>0$ and positive compliance damage $\omega>0$. Define $\psired$ by Eq.~\eqref{eq:w-reduced-energy}. If the bulk energy has the form
\begin{equation}\label{eq:bulk-plane-reduction-assumption}
 \psib(\bmFb)=\psibhat(I_1,\Jb),
 \qquad I_1\coloneqq\bmFb:\bmFb,
 \qquad \frac{\partial\psibhat}{\partial I_1}\geq0,
\end{equation}
then at least one globally optimal direction lies in the plane spanned by the directions of the largest and smallest right principal stretches. Specifically, let $\lambda_1\geq\lambda_2\geq\lambda_3>0$ be the principal stretches of $\bmF$ and let $\bmN_i$ be the corresponding right principal directions, namely the unit eigenvectors of $\bmF^\trans\bmF$. The search in Eq.~\eqref{eq:N-reduced-minimization} can then be restricted to
\begin{equation}\label{eq:N-principal-plane}
 \bmN(\theta)=\cos\theta\,\bmN_1+\sin\theta\,\bmN_3,
 \qquad 0\leq\theta\leq\frac\pi2.
\end{equation}
\end{proposition}
The proof is given in \ref{app:crack-direction-reduction}. This result reduces local direction selection to a single-angle search in three-dimensional calculations. Neo--Hookean elasticity satisfies condition~\eqref{eq:bulk-plane-reduction-assumption}. 

For general isotropic bulk laws, symmetry still restricts the search to two angular parameters in three dimensions. The following proposition establishes this reduction and explains the angular interval in Eq.~\eqref{eq:N-principal-plane}.

\begin{proposition}[Symmetry reduction of the crack-direction search]\label{prop:crack-direction-symmetry-reduction}
Let $d\in\{2,3\}$, with $\det\bmF>0$ and $\omega>0$ fixed. Assume an objective isotropic bulk energy, the cohesive energy~\eqref{eq:psic}, and the reduced energy $\psired$ defined by Eq.~\eqref{eq:w-reduced-energy}. Then the direction search in Eq.~\eqref{eq:N-reduced-minimization} can be restricted without changing the minimum energy:
\begin{equation}
	 \min_{\bmN\in\mathbb S^{d-1}}\psired(\bmF,\bmN,\omega)
	=\min_{\bmN\in\mathcal S_+(\bmF)}\psired(\bmF,\bmN,\omega),
\end{equation}
where $\mathcal S_+(\bmF)$ is the set of unit directions with nonnegative components along the right principal directions,
\begin{equation}\label{eq:direction-positive-orthant}
 \mathcal S_+(\bmF)\coloneqq
 \left\{\sum_{i=1}^{d}a_i\bmN_i\,\middle|\,
 a_i\geq0,\ \sum_{i=1}^{d}a_i^2=1\right\},
\end{equation}
and $\bmN_i$ are the right principal directions of $\bmF$. In two dimensions, $\mathcal S_+(\bmF)$ corresponds to
$\bmN(\theta)=\cos\theta\,\bmN_1+\sin\theta\,\bmN_2$, $0\leq\theta\leq\pi/2$. In three dimensions, it is a spherical octant in right-principal coordinates, which can also be parameterized by the two-dimensional simplex $x_i=a_i^2\geq0$, $\sum_i x_i=1$.
\end{proposition}
\begin{proof}
Choose the right and corresponding left principal directions as material and spatial bases, respectively, so that $\bmF$ is diagonal with positive entries. Let $\mathbm D=\operatorname{diag}(\varepsilon_1,\ldots,\varepsilon_d)$, where $\varepsilon_i\in\{-1,1\}$. Under the simultaneous transformations $\bmN\mapsto\mathbm D\bmN$ and $\bmw\mapsto\mathbm D\bmw$, one has $\bmn\mapsto\mathbm D\bmn$ and $\bmFb\mapsto\mathbm D\bmFb\mathbm D$. The transformation preserves the bulk principal stretches, $\Jb$, normal opening, and tangential sliding magnitude. It therefore preserves both energies and admissibility, and establishes a bijection between admissible crack strains $\bmw$ for the two candidate normals $\bmN$. Minimization over $\bmw$ preserves the reduced energy, so each component along a right principal direction may independently be chosen nonnegative.
\end{proof}

\subsection{Damage evolution and initiation}\label{sec:damage}
\subsubsection{Damage evolution}\label{sec:damage-evolution}

Having established local traction balance and direction selection, we next determine how the damage fields evolve under irreversibility. Since crack strain is equilibrated locally, the driving forces for $\omega$ and $\phi$ can be obtained from the reduced total potential
\begin{equation}\label{eq:reduced-total-energy}
 \energyred[\bmchi,\bmN,\omega,\phi]
 \coloneqq
 \int_{\mathcal B}
 \left[\psired(\bmF,\bmN,\omega)
       +\varphif(\phi,\nabla\phi,\omega)\right]\,\dd V
 +\mathcal P_{\mathrm{ext}}[\bmchi].
\end{equation}
At fixed $\bmchi$ and $\bmN$, the damage fields minimize $\energyred$ subject to irreversibility. Assume that the required functional derivatives exist and continuous-time evolution is sufficiently smooth. For $\phi<1$, the first-order necessary conditions have the rate form
\begin{subequations}\label{eq:damage-KKT}
 \begin{align}
 &\variationalderivative{\energyred}{\omega}
 =\frac{\Gc}{\ell}(1-\phi)^2
  +\frac{\partial\psired}{\partial\omega}\geq0,
 \qquad\dot\omega\geq0,
 \qquad\variationalderivative{\energyred}{\omega}\dot\omega=0,
 \label{eq:omega-KKT}\\[0.7em]
 &\variationalderivative{\energyred}{\phi}
 =\frac{\Gc}{\ell}
  \left[\Omega(\phi)\Omega'(\phi)
        -\ell^2\nabla^2\phi-2(1-\phi)\omega\right]\geq0,
 \qquad\dot\phi\geq0,
 \qquad\variationalderivative{\energyred}{\phi}\dot\phi=0.
 \label{eq:phi-KKT}
 \end{align}
\end{subequations}
For the damage field $\alpha\in\{\omega,\phi\}$, the first variation along a test perturbation $v$, with the other arguments fixed, is defined by
\[
 \firstvariation{\alpha}\energyred[v]\coloneqq
 \lim_{s\to0^+}\frac{\energyred[\alpha+sv]-\energyred[\alpha]}{s}.
\]
Under the regularity assumptions above, a smooth admissible $v$ that vanishes near the boundary satisfies $\firstvariation{\alpha}\energyred[v]=\int_{\mathcal B}(\variationalderivative{\energyred}{\alpha})v\,\dd V$, where $\variationalderivative{\energyred}{\alpha}$ is the interior variational-derivative density used in Eq.~\eqref{eq:damage-KKT}. Boundary contributions and the weak form for general test perturbations are given in \ref{app:damage-inequality}.

Phase-field irreversibility also restricts boundary variations. On a boundary portion with no prescribed phase-field value and $\phi<1$, the same variational conditions give the boundary complementarity relation \citep{pham2011approximate}:
\begin{equation}\label{eq:phi-boundary-KKT}
 \nabla\phi\cdot\bmnu\geq0,
 \qquad\dot\phi\geq0,
 \qquad(\nabla\phi\cdot\bmnu)\dot\phi=0,
\end{equation}
where $\bmnu$ is the outward unit normal in the reference configuration. The normal derivative therefore vanishes where the boundary phase-field is increasing, but may be positive where it is stationary.

\subsubsection{Strength surface and damage initiation}\label{sec:strength-surface}\label{sec:damage-initiation}
The onset of damage is most readily characterized by the traction on a candidate plane, which also determines how the driving force depends on crack direction. Recovering the bulk deformation gradient from the minimizer $\bmw(\bmF,\bmN,\omega)$ in Eq.~\eqref{eq:w-reduced-energy} gives
\begin{equation}\label{eq:reduced-traction}
 \bmT(\bmF,\bmN,\omega)=\bmPb(\bmFb)\bmN,
 \qquad\bmFb=\bmF-\bmw(\bmF,\bmN,\omega)\otimes\bmN.
\end{equation}
This traction agrees with $\bmP\bmN$ in Eq.~\eqref{eq:traction-equilibrium}; its normal and tangential components are defined as in Eq.~\eqref{eq:traction-components}.

We define the strength function in terms of the normal and tangential tractions as
\begin{equation}\label{eq:strength-function}
 \reducedstrength(\bmF,\bmN,\omega,\phi)\coloneqq
 \positivepartexpr{\frac{\Tperp}{\Tperpc}}^2
 +\left(\frac{\Tparallel}{\Tparallelc}\right)^2
 -(1-\phi)^2,
\end{equation}
where $\positivepart{(s)}\coloneqq\max\{s,0\}$ denotes the positive part of a scalar. The derivative of the reduced total potential with respect to compliance damage then reads
\begin{equation}\label{eq:reduced-strength-relation}
 \variationalderivative{\energyred}{\omega}
 =-\frac{\Gc}{\ell}\reducedstrength(\bmF,\bmN,\omega,\phi).
\end{equation}
The compliance-damage KKT conditions~\eqref{eq:omega-KKT} can therefore be written as
\begin{equation}\label{eq:omega-strength-surface}
 \reducedstrength\leq0,
 \qquad\dot\omega\geq0,
 \qquad\reducedstrength\dot\omega=0.
\end{equation}
Normal tension and tangential sliding drive damage, whereas normal compression is carried by the contact reaction. Equation~\eqref{eq:reduced-strength-relation} applies for $\omega\geq0$; its derivation is given in \ref{app:damage-inequality}.

Equation~\eqref{eq:omega-strength-surface} identifies $\reducedstrength=0$ as the quadratic mixed-mode tensile--shear strength surface: compliance damage $ \omega $ cannot grow for $\reducedstrength<0$ and may grow only at $\reducedstrength=0$. This quadratic form is consistent with the mixed-mode cohesive model of \citet{ortiz1999finite} and the unified tensile fracture criterion of \citet{zhang2005unified}. As $\phi$ increases, the tensile--shear branch shrinks by the factor $1-\phi$; the critical tractions in pure-normal and pure-tangential loading are $(1-\phi)\Tperpc$ and $(1-\phi)\Tparallelc$, respectively.

In the intact state $\omega=\phi=0$, one has $\bmw=\mathbm0$, and the traction in Eq.~\eqref{eq:reduced-traction} reduces to $\bmT_0=\bmPb(\bmF)\bmN$. Although $\psired(\bmF,\bmN,0)=\psib(\bmF)$ is direction-independent, energy minimization for $\omega>0$ still determines the initiation direction as $\omega\to0^+$. The first-order expansion in \ref{app:damage-initiation} gives the limiting set of energy-minimizing directions:
\begin{equation}\label{eq:intact-direction-selection}
 \limsup_{\omega\to0^+}
 \left(
 \operatorname*{arg\,min}_{\bmN\in\mathbb S^{d-1}}
 \psired(\bmF,\bmN,\omega)
 \right)
 \subseteq
 \operatorname*{arg\,max}_{\bmN\in\mathbb S^{d-1}}
 \reducedstrength(\bmF,\bmN,0,0).
\end{equation}
The set upper limit collects all accumulation points of energy-minimizing directions at positive compliance damage $\omega>0$ as $\omega\to0^+$. Directions selected by this one-sided limit therefore maximize the driving force, giving the first-order initiation condition
\begin{equation}\label{eq:intact-critical-plane}
 \max_{\bmN\in\mathbb S^{d-1}}\reducedstrength(\bmF,\bmN,0,0)=0.
\end{equation}
Before initiation, every candidate plane satisfies $\reducedstrength\leq0$. The direction-search algorithm is described in Section~\ref{sec:numerical-implementation}.

\subsection{Relation to energy-degradation phase-field models}\label{sec:degradation-relation}
The local condensation also clarifies the relation to small-strain cohesive phase-field models, which commonly degrade elastic energy with a scalar function \citep{lorentz2011convergence,lorentz2012modelling,wu2017unified,feng2021endowing}. To recover this form and examine its range of validity, we specialize the bulk and crack-strain energies to quadratic forms. For one-dimensional small-strain tension, let $\varepsilon\geq0$ be the total strain and $E>0$ Young's modulus. At fixed $\omega>0$, the reduced energy~\eqref{eq:w-reduced-energy} becomes
\begin{equation}\label{eq:quadratic-crack-strain-reduction}
 \min_{\w\geq0}
 \left\{\frac E2(\varepsilon-\w)^2
       +\frac{\ell\kperp}{2\omega}\w^2\right\}
 =\frac E2g_1(\omega)\varepsilon^2,
 \qquad
 g_1(\omega)=\frac1{1+c_1\omega},
 \qquad c_1=\frac E{\ell\kperp}.
\end{equation}
The minimizer is $\w=[1-g_1(\omega)]\varepsilon$; at $\omega=0$, continuous extension gives $\w=0$ and $g_1=1$. Normal and tangential degradation functions in directional small-strain models follow from the same local condensation \citep{feng2023unified,feng20253d}. Thus the reciprocal degradation function in Eq.~\eqref{eq:quadratic-crack-strain-reduction} is a result of eliminating crack strain from quadratic bulk and cohesive energies. The multiscale interpretation of damage through internal strain variables builds on the work of \citet{FREDDI20101154} and \citet{feng2023unified}. If $\omega=\omega(\phi)$ is further prescribed, this mapping determines the dependence of the degradation function on $\phi$, which need not be rational. For a nonlinear bulk law, Eq.~\eqref{eq:w-reduced-energy} still defines the reduced energy. In general, however, it cannot be written as the original bulk energy multiplied by a scalar function of damage alone.

In cohesive phase-field models that prescribe $\omega=\omega(\phi)$ \citep{lorentz2011convergence,lorentz2012modelling,wu2017unified,feng2021endowing}, convexity of the composite degradation function $g_1(\omega(\phi))$ in $\phi$ imposes the upper bound $ \ell\leq C \ell_{\mathrm{ch}}  $ on $\ell$. Here $ C $ is a model-dependent constant and $ \ell_{\mathrm{ch}} \propto E\Gc/(\Tperpc)^2$ is the characteristic fracture-process-zone length for a propagating crack (the Irwin length). The restriction on $\ell$ imposed by convexity limits applications to high-strength materials and large structures.

Retaining independent compliance damage changes the relevant convexity condition: the curvature of $g_1$ with respect to $\omega$ is
\begin{equation}\label{eq:quadratic-degradation-convexity}
 \frac{\dd^2g_1}{\dd\omega^2}
 =\frac{2c_1^2}{(1+c_1\omega)^3}>0,
 \qquad\omega\geq0,
\end{equation}
The curvature is positive for every $\ell>0$, removing this upper bound on $\ell$ and enabling more efficient analysis of large structures.
\citet{huber2024length} first introduced a construction using an independent variable to remove this bound.

The present nonlinear model prescribes no fixed scalar degradation function. Cohesive behavior instead follows from the energy split between bulk deformation and crack separation, independent compliance damage, and its evolution. The upper bound on $ \ell $ is therefore absent, making the model suitable for large structures and high-strength materials. Section~\ref{sec:cohesive-law} further shows that the traction--separation relation contains no $\ell$ and is independent of the bulk energy. The cohesive law is therefore intrinsically insensitive to $ \ell $.

\section{One-dimensional solutions, strong discontinuities, and cohesive laws}\label{sec:one-dimensional-analysis}
\subsection{\texorpdfstring{{One-dimensional analytical solutions}}{One-dimensional analytical solutions}}\label{sec:one-dimensional-solutions}
{To determine whether the fields develop a displacement jump at finite $\ell$ and recover the prescribed cohesive law, we examine a one-dimensional tensile problem. Consider a homogeneous bar $\mathcal B=(0,L)$ of reference length $L$ and unit cross-sectional area, with axial material coordinate $X$.} There is no body force. Uniaxial tension is imposed by the end displacements
\begin{equation}\label{eq:1d-loading}
 u(0,t)=0,\qquad u(L,t)=\bar u(t),
\end{equation}
where $u$ is the axial displacement and $\bar u(t)$ is the monotonically increasing boundary displacement. The phase-field boundary conditions are those in Eq.~\eqref{eq:phi-boundary-KKT}.

Let $\psib(\Fb)$ be the one-dimensional bulk strain-energy density, $\Fb>0$ the bulk stretch, and $\w\geq0$ the axial crack strain. The reference state is stress-free, and the energy reference is chosen so that $\psib(1)=\psib'(1)=0$. Bulk stretches are restricted to a tensile interval on which $\psib''(\Fb)>0$, making $\psib$ strictly convex there. The equilibrium equation~\eqref{eq:momentum-balance}, kinematic split, and traction relation~\eqref{eq:traction-equilibrium} reduce to
\begin{subequations}\label{eq:1d-mechanics}
 \begin{align}
 &P_{,X}=0,
 \label{eq:1d-balance}\\[0.7em]
 &P=\psib'(\Fb),\qquad \Fb=F-\w,
\qquad \w=\frac{\omega P}{\ell\kperp}.
 \label{eq:1d-kinematics}
 \end{align}
\end{subequations}
Here $P$ is the axial nominal stress, uniform along the bar by Eq.~\eqref{eq:1d-balance}. Strict convexity makes $P=\psib'(\Fb)$ strictly increasing on the selected interval; its inverse is denoted by $\Fb(P)$.

On the opening branch $P\geq0$, the compliance-damage KKT condition~\eqref{eq:omega-strength-surface} for $\omega$ reduces to
\begin{equation}\label{eq:1d-omega-KKT}
 \reducedstrength=(P/\Tperpc)^2-(1-\phi)^2\leq0,
 \qquad\dot\omega\geq0,\qquad\reducedstrength\dot\omega=0.
\end{equation}

The phase-field KKT condition~\eqref{eq:phi-KKT} becomes
\begin{equation}\label{eq:1d-phi-KKT}
 \variationalderivative{\energyred}{\phi}
 =\frac{\Gc}{\ell}\left[\Omega(\phi)\Omega'(\phi)-\ell^2\phi_{,XX}
       -2(1-\phi)\omega\right]\geq0,\qquad
 \dot\phi\geq0,\qquad
 \variationalderivative{\energyred}{\phi}\dot\phi=0.
\end{equation}

On the intact branch, $\omega=\phi=\w=0$, and the bar carries load through a uniform bulk stretch. When the compliance-damage initiation condition $P=\Tperpc$ in Eq.~\eqref{eq:1d-omega-KKT} is reached, assume that a single region of increasing $\omega$ develops:
\begin{equation}\label{eq:1d-damage-band}
 \damageband\coloneqq\{X\in\mathcal B:\dot\omega(X,t)>0\}
 =(X_-,X_+),\qquad h\coloneqq X_+-X_-,
\end{equation}
where $0<X_-<X_+<L$. Also assume that $\omega=0$ holds outside the band throughout loading.

Within $\damageband$, the compliance-damage KKT condition~\eqref{eq:1d-omega-KKT} for $\omega$ gives $P=(1-\phi)\Tperpc$. Since $P$ is uniform, the phase-field is also constant there; denote it by $\phi^*$. Thus
\begin{equation}\label{eq:1d-band-strength}
 \phi(X,t)=\phi^*(t),\qquad P=(1-\phi^*)\Tperpc,
 \qquad X\in\damageband.
\end{equation}

We examine the damage-loading branch with $0<\phi^*<1$ and $\dot\phi^*>0$. Take $L\gg\ell$ and place the damage band $\damageband$ far from the ends, so that $\phi$ and $\phi_{,X}$ vanish there in the long-bar approximation. Denote the phase-field loading region by $\mathcal B_\phi\coloneqq\{X\in\mathcal B:\dot\phi(X,t)>0\}$. Since $ \omega $ acts as a driving force for the phase-field (Eq.~\eqref{eq:1d-phi-KKT}), $ \mathcal B_\phi\supseteq \damageband $. Moreover, $\omega=0$ outside $ \damageband $, giving
\begin{equation}\label{eq:1d-phase-tail-ode}
 -\ell^2\phi_{,XX}+\Omega(\phi)\Omega'(\phi)=0,
 \qquad X\in\mathcal B_\phi\setminus\overline{\damageband}.
\end{equation}
This autonomous equation for $ \phi $ is subject to the boundary condition $\phi=\phi^*$ on $ \partial\damageband $.

Following \citet{feng2025JMPS}, we obtain the tail profile from a first integral of Eq.~\eqref{eq:1d-phase-tail-ode} (\ref{app:1d-phase-profile}):
\begin{equation}\label{eq:1d-phase-tail}
 \mathrm{dist}(X,\damageband)
 =\ell\int_{\phi(X)}^{\phi^*}\frac{\dd s}{\Omega(s)},
 \qquad X\in\mathcal B_\phi\setminus\overline{\damageband},
\end{equation}
where $\operatorname{dist}(X,\damageband)\coloneqq\inf_{S\in\damageband}|X-S|$ is the reference distance from $X$ to the damaged region, equal to $X_--X$ and $X-X_+$ in the left and right tails, respectively. Intact material outside the phase-field loading region has $\phi=0$. For a fixed damage band, the width enclosed by any given phase-field level is nondecreasing with $\phi^*$; see \ref{app:1d-phase-profile}.

Inside the damaged region, $\phi_{,XX}=0$ and $\dot\phi>0$. The phase-field KKT condition~\eqref{eq:1d-phi-KKT} and constitutive relation~\eqref{eq:1d-kinematics} give the distributions of compliance damage $\omega$ and crack strain $\w$:
\begin{equation}\label{eq:1d-band-internal-fields}
 \omega(X)=\omega^*\mathbf1_{\damageband}(X),
 \qquad \w(X)=\w^*\mathbf1_{\damageband}(X),
 \qquad X\in\mathcal B,
\end{equation}
where $\mathbf1_{\damageband}$ is the indicator of the damaged region, while $\omega^*$ and $\w^*$ are the uniform compliance damage and crack strain within it. Using $P=(1-\phi^*)\Tperpc$ and $\kperp=(\Tperpc)^2/(2\Gc)$, both can be expressed in terms of $\phi^*$ as
\[
 \omega^*=\frac{\Omega(\phi^*)\Omega'(\phi^*)}{2(1-\phi^*)},
 \qquad
 \w^*=\frac{\Gc}{\ell\Tperpc}\Omega(\phi^*)\Omega'(\phi^*).
\]
Integrating the kinematic relation~\eqref{eq:1d-kinematics} recovers the displacement and enforces the end conditions:
\begin{equation}\label{eq:1d-displacement-response}
u(X)=[\Fb(P)-1]X+\w^*\int_0^X\mathbf1_{\damageband}(S)\,\dd S,
 \qquad X\in[0,L].
\end{equation}
The end displacement is $\bar u=L[\Fb(P)-1]+h\w^*$. 

{These relations determine $u$, $\phi$, $\omega$, and $\w$ for a prescribed damage band.}

\subsection{Energy comparison and emergence of strong discontinuities}\label{sec:energy-comparison}
{The remaining question is which band width is energetically preferred. \citet{feng2025JMPS} established the necessity of concentration from the incompatibility between the plateau and tail gradients at the band boundary; see also \ref{app:1d-band-concentration}. We reach the same conclusion by comparing the minimum energies of bands with different widths at the same boundary displacement.}

{Consider a band $\damageband=(X^*-h/2,X^*+h/2)$ with fixed center $X^*\in(0,L)$ and $h>0$, retaining the long-bar approximation and $\omega=0$ outside the band.}

{To compare different band widths at the same total separation, define the crack opening by}
\begin{equation}\label{eq:1d-total-opening}
 \Delta\coloneqq\int_0^L\w\,\dd X\geq0.
\end{equation}
{For fixed $h$, $\Delta$, and $0\leq\phi^*<1$, retain the phase-field plateau $\phi=\phi^*$ inside the band and the tails~\eqref{eq:1d-phase-tail} outside.}

For this phase-field profile, minimize over the remaining fields to define the conditional minimum energy
\begin{equation}\label{eq:1d-band-conditional-minimum}
 \bandenergy_h(\Delta,\phi^*;\bar u)
 \coloneqq\min_{(u,\w,\omega)\in\mathcal A_h(\bar u,\Delta)}
 \energy[\chi,\w,\omega,\phi].
\end{equation}
{Here $\chi=X+u$, and the superscript $(3)$ denotes the variables $\Delta,\phi^*,\bar u$, with $h$ treated as a parameter. The admissible set $\mathcal A_h(\bar u,\Delta)$ enforces the end-displacement and total-opening constraints, $\w,\omega\geq0$, and $\w=\omega=0$ outside the band.} 

The kinematic relation~\eqref{eq:1d-kinematics} shows that $\bar u$ and $\Delta$ fix the total bulk extension. Strict convexity of $\psib$ makes a uniform bulk stretch minimize the bulk energy, giving
\begin{equation}\label{eq:1d-comparison-bulk-stretch}
 \Fb=1+\frac{\bar u-\Delta}{L}>0.
\end{equation}
The corresponding bulk energy is
\begin{equation}
 \int_0^L\psib(\Fb)\,\dd X
 =L\psib\!\left(1+\frac{\bar u-\Delta}{L}\right).
\end{equation}

{Having minimized the bulk energy, we next eliminate compliance damage from the cohesive and fracture contributions. The two $\omega$-dependent terms in Eqs.~\eqref{eq:psic} and~\eqref{eq:fracture-energy-density} satisfy}
\begin{equation}\label{eq:1d-compliance-condensation}
 \int_{\damageband}\left[\frac{\ell\kperp}{2\omega}\w^2
 +\frac{\Gc}{\ell}(1-\phi^*)^2\omega\right]\,\dd X
 \geq\Tperpc(1-\phi^*)\Delta,
\end{equation}
where $\kperp=(\Tperpc)^2/(2\Gc)$; equality holds for $\omega=\ell\Tperpc\w/[2\Gc(1-\phi^*)]$. At $\w=0$, take $\omega=0$ with the energy extension~\eqref{eq:psic-omega-zero}. This lower bound depends only on the total opening, not on the distribution of $\w$ within the band.

{Substituting the tail solution~\eqref{eq:1d-phase-tail} into the phase-field energy gives}
\begin{equation}\label{eq:1d-phase-tail-energy}
 \frac{\Gc}{2\ell}\int_{\mathcal B\setminus\overline{\damageband}}
 \left[\Omega(\phi)^2+\ell^2\phi_{,X}^2\right]\,\dd X
 =2\Gc\int_0^{\phi^*}\Omega(s)\,\dd s.
\end{equation}
{Adding the plateau contribution $h\Gc\Omega(\phi^*)^2/(2\ell)$ to the bulk, cohesive, and tail energies then yields}
\begin{equation}\label{eq:1d-positive-width-energy}
 \bandenergy_h(\Delta,\phi^*;\bar u)
 =L\psib\!\left(1+\frac{\bar u-\Delta}{L}\right)
 +\Tperpc(1-\phi^*)\Delta
 +2\Gc\int_0^{\phi^*}\Omega(s)\,\dd s
 +\frac{h\Gc}{2\ell}\Omega(\phi^*)^2.
\end{equation}

Equation~\eqref{eq:1d-positive-width-energy} shows that, for any fixed admissible $\Delta,\bar u$ and $\phi^*\in(0,1)$, the energy increases strictly with damage-band width $h$:
\begin{equation}\label{eq:1d-width-energy-derivative}
 \frac{\partial\bandenergy_h}{\partial h}
 =\frac{\Gc}{2\ell}\Omega(\phi^*)^2>0, \qquad \phi^*>0,
 \qquad h>0.
\end{equation}
The conditional minimum energy therefore approaches its infimum as $h\to0^+$, namely Eq.~\eqref{eq:1d-positive-width-energy} without its final term.\footnote{Candidate states are compared from the intact state at the same end displacement.} For any fixed $\Delta>0$, the average crack strain within the band diverges as $h\to0^+$. The crack strain develops a Dirac-type concentration at $X^*$, producing a strong displacement discontinuity with jump $\jump{u}(X^*)=\Delta$. Its measure representation is given in \ref{app:model-concentration}.

{Fig.~\ref{fig:theory-concentration} illustrates concentration at fixed $\ell$. In Fig.~\ref{fig:theory-concentration}(a--c), the phase-field remains continuous while crack strain and compliance damage concentrate, and the separation displacement approaches a jump of height $\Delta$. Fig.~\ref{fig:theory-concentration}(d) compares the conditional minimum energies at the same $\Delta$, $\phi^*$, and $\bar u$.}

\begin{figure}[!htbp]
 \centering
 \includegraphics[width=\linewidth]{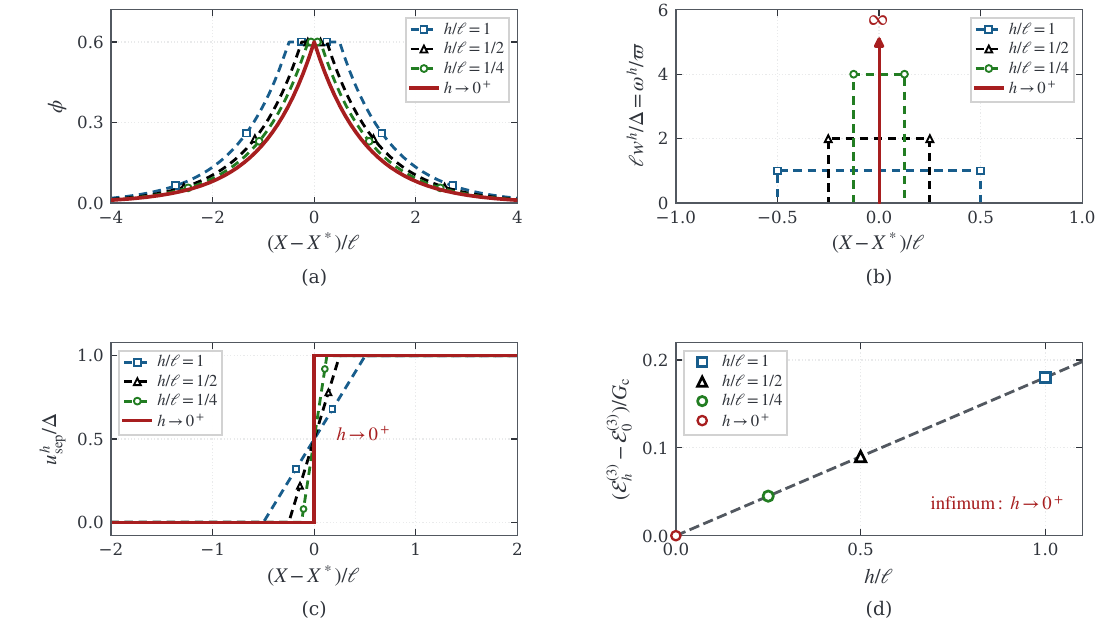}
 \caption{Concentration at fixed $\ell$, $\Delta$, $\phi^*=0.6$, and $\bar u$, with $\Omega(\phi)=\phi$: (a) phase-field; (b) normalized crack strain and compliance damage; (c) separation displacement; (d) conditional excess energy. Red curves and symbols denote the $h\to0^+$ limits.}
 \label{fig:theory-concentration}
\end{figure}

\FloatBarrier
\subsection{\texorpdfstring{{CZM energy in the concentrated limit}}{CZM energy in the concentrated limit}}\label{sec:concentrated-energy}
{The concentration selected by the energy comparison defines an interface $\Gamma=\{X^*\}$. In the strong-discontinuity limit, crack strain and compliance damage satisfy}
\begin{equation}\label{eq:1d-band-concentration-sequence}
 {
 \w^h\,\dd X\overset{*}{\rightharpoonup}\Delta\delta_{X^*},
 \qquad
 \omega^h\,\dd X\overset{*}{\rightharpoonup}\ell\varpi\delta_{X^*},}
\end{equation}
{where $\Delta=\jump{u}(X^*)$ and $\varpi\geq0$ is a dimensionless interfacial damage variable. To retain interfacial damage evolution, we take $u$, $\varpi$, and $\phi^*=\phi(X^*)$ as independent variables. The current displacement jump set satisfies $\jumpSet\subseteq\Gamma$.}\footnote{The displacement jump may vanish when a crack closes, but the damage already formed must remain. Thus $\Gamma$ does not disappear when the current jump set $\jumpSet$ does.}

{To obtain the energy in these variables, write the one-dimensional potential~\eqref{eq:E} as}
\begin{equation}\label{eq:1d-original-energy}
 \energy=\int_0^L\left[\psib(\Fb)
 +\frac{\ell\kperp}{2\omega}\w^2
 +\frac{\Gc}{2\ell}\bigl(\Omega(\phi)^2+\ell^2\phi_{,X}^2+2(1-\phi)^2\omega\bigr)\right]\,\dd X,
\end{equation}
{where $\Fb=1+u_{,X}-\w$. As $h\to0^+$, the bulk and phase-field plateau energies inside the band vanish, while the cohesive, coupling, and tail energies remain; see \ref{app:model-concentration} for evaluation through finite-width fields. The resulting energy is}
\begin{equation}\label{eq:1d-concentrated-energy}
 \boxed{\energyczm[u,\varpi,\phi^*]
 =\int_{\mathcal B\setminus\Gamma}\psib(1+u_{,X})\,\dd X
 +\int_\Gamma\!\left[\frac{\kperp}{2\varpi}\jump{u}^{2}
 +\Gc(1-\phi^*)^2\varpi+\Gc\fractureprofile(\phi^*)\right]\dd\mathcal H^0,}
\end{equation}
{where $u_{,X}$ is the ordinary derivative on either side of the crack and $\mathcal H^0$ is counting measure. The normalized phase-field tail energy is}
\begin{equation}\label{eq:1d-phase-tail-potential}
 \fractureprofile(\phi^*)\coloneqq2\int_0^{\phi^*}\Omega(s)\,\dd s.
\end{equation}
{Normalization~\eqref{eq:Omega-conditions} ensures $\fractureprofile(1)=1$. The interface integral retains both cohesive storage and fracture dissipation. With $\dot\varpi\geq0$ and $\dot\phi^*\geq0$, Eq.~\eqref{eq:1d-concentrated-energy} describes irreversible interfacial damage and has the energy structure of the model of \citet{feng_Li2026}.}\footnote{{At $\varpi=0$, the cohesive storage is zero for $\Delta=0$ and $+\infty$ otherwise, as in Eq.~\eqref{eq:psic-omega-zero}.}}

Further condensation of $\varpi$ and $\phi^*$ by minimizing the concentrated interfacial energy gives the potential of the monotonic loading envelope from the intact state:
\begin{equation}\label{eq:1d-cohesive-separation-potential}
 g(\Delta)=\inf_{\substack{\varpi\geq0\\0\leq\phi^*\leq1}}
 \left\{\frac{\kperp\Delta^2}{2\varpi}
 +\Gc(1-\phi^*)^2\varpi+\Gc\fractureprofile(\phi^*)\right\}
 =\min_{0\leq\phi^*\leq1}
 \left\{\Tperpc(1-\phi^*)\Delta+\Gc\fractureprofile(\phi^*)\right\}.
\end{equation}
Minimization over $\varpi$ first gives $\varpi=\Tperpc\Delta/[2\Gc(1-\phi^*)]$ for $0\leq\phi^*<1$; substitution yields the second equality.
The corresponding loading-envelope energy depends only on displacement:
\begin{equation}\label{eq:1d-CZM-functional}
 \boxed{\energyenvelope[u]
 =\int_{\mathcal B\setminus\jumpSet}\psib(1+u_{,X})\,\dd X
 +\int_{\jumpSet}g\!\left(\jump{u}\right)\,\dd\mathcal H^0.}
\end{equation}
This functional combines the bulk volume energy with a jump-dependent fracture energy, giving the standard cohesive fracture energy structure \citep{CONTI20161033,conti2024vectorial}.

The model also admits diffuse damage, providing deformation modes for the bulk-energy relaxation required by free cohesive-crack functionals \citep{Braides1995,rodella2026sharp}. Section~\ref{sec:example4} illustrates this diffuse response.

\subsection{\texorpdfstring{{Cohesive laws}}{Cohesive laws}}\label{sec:cohesive-law}
{Phase-field equilibrium at the concentrated interface now determines its nominal traction--separation law.}

The strength relation~\eqref{eq:1d-band-strength} and the concentrated-interface constitutive law give
\begin{equation}\label{eq:1d-concentrated-traction-opening}
 \Tperp=P=(1-\phi^*)\Tperpc,
 \qquad\Delta=\frac{\varpi P}{\kperp}.
\end{equation}

As $h\to0^+$, both tails remain continuous at the crack point $X^*$, where the phase-field equals $\phi^*$, but the gradient jumps. The tail gradient~\eqref{eq:1d-tail-signed-gradient} gives
\begin{equation}\label{eq:1d-crack-phase-gradient-jump}
 \jump{\phi_{,X}}=-\frac{2}{\ell}\Omega(\phi^*),
\end{equation}
where $\jump{\phi_{,X}}\coloneqq\phi_{,X}(X^{*+})-\phi_{,X}(X^{*-})$ is the right-minus-left jump in the phase-field gradient. On the loading branch $0<\phi^*<1$, $\dot\phi^*>0$, integrate the phase-field equation~\eqref{eq:1d-phi-KKT} over a neighborhood of the crack point and let its width vanish. The integral of the local term $\Omega\Omega'$ tends to zero. {The gradient jump and the concentrated compliance damage in Eq.~\eqref{eq:1d-band-concentration-sequence} give the crack-point balance}
\begin{equation}\label{eq:1d-singular-phase-balance}
 -\ell^2\jump{\phi_{,X}}-2\ell(1-\phi^*)\varpi=0,
 \qquad\text{or equivalently,}\qquad\Omega(\phi^*)=(1-\phi^*)\varpi.
\end{equation}
Substitution of Eq.~\eqref{eq:1d-singular-phase-balance} into the cohesive constitutive law~\eqref{eq:1d-concentrated-traction-opening} gives $\Delta=\Deltaperpch\Omega(\phi^*)$.

Scaling by the critical pure-normal nominal traction $\Tperpc$ and characteristic separation $\Deltaperpch=2\Gc/\Tperpc$, write $\normalizedtraction\coloneqq\Tperp/\Tperpc$ and $\normalizedopening\coloneqq\Delta/\Deltaperpch$. The dimensionless parametric cohesive law is
\begin{equation}\label{eq:1d-parametric-cohesive-law}
 \boxed{\normalizedtraction=1-\phi^*,
 \qquad\normalizedopening=\Omega(\phi^*),
 \qquad0\leq\phi^*<1.}
\end{equation}
Complete separation is understood as $\phi^*\to1^-$, after which the traction vanishes. The finite-deformation bulk law affects the bar's overall load--displacement response. It does not alter the nominal traction--separation relation per unit reference area in Eq.~\eqref{eq:1d-parametric-cohesive-law}.

{The same law follows from stationarity of the concentrated energy~\eqref{eq:1d-cohesive-separation-potential}: for $0<\phi^*<1$, $-\Tperpc\Delta+\Gc\fractureprofile'(\phi^*)=0$. Since $\fractureprofile'=2\Omega$, this gives the opening relation, while $g'(\Delta)=\Tperpc(1-\phi^*)$ gives the traction. Crack-point equilibrium and energy minimization thus recover the same cohesive law.}

Since $\Omega$ is strictly increasing, Eq.~\eqref{eq:1d-parametric-cohesive-law} uniquely determines the softening shape, $\normalizedtraction(\normalizedopening)=1-\Omega^{-1}(\normalizedopening)$. Conversely, inverting a strictly decreasing target cohesive law uniquely determines $\Omega(\phi^*)=\normalizedtraction^{-1}(1-\phi^*)$. The resulting $\Omega$ must also satisfy conditions~\eqref{eq:Omega-conditions}. Fiber-bridging traction--separation curves can be measured directly or inferred from experiments on soft collagenous tissues. These measurements link postpeak softening to tissue toughness \citep{tang2025tough}. This provides a physical basis for constructing $\Omega$ from measured softening curves.

The area under the cohesive law is the energy expended in complete separation per unit reference crack area. Let $\Delta_{\mathrm f}\coloneqq\Deltaperpch\Omega(1^-)$ be the opening at complete separation, possibly infinite.\footnote{For $\Omega(1^-)=\infty$, the energy integral is understood as a limit. Monotonicity and integrability of $\Omega$ ensure $(1-s)\Omega(s)\to0$, so the endpoint term in integration by parts still vanishes.} Using the parametric form~\eqref{eq:1d-parametric-cohesive-law} and integrating by parts gives
\begin{equation}\label{eq:1d-cohesive-work-normalization}
 \int_0^{\Delta_{\mathrm f}}\Tperp(\Delta)\,\dd\Delta
 =2\Gc\int_0^1(1-s)\Omega'(s)\,\dd s
 =2\Gc\int_0^1\Omega(s)\,\dd s
 =\Gc.
\end{equation}
Thus normalization~\eqref{eq:Omega-conditions} makes the complete-separation energy equal to $\Gc$.

Equation~\eqref{eq:1d-parametric-cohesive-law} describes the monotonic softening envelope from the intact state. During unloading and reloading before the envelope is reached again, the damage variables $\varpi$ and $\phi^*$ remain constant. The recoverable cohesive term in the CZM energy~\eqref{eq:1d-concentrated-energy} gives $\Tperp=\kperp\Delta/\varpi$, so unloading and reloading follow a straight line through the origin. Damage may evolve again only upon returning to the softening envelope.

{Under mixed-mode loading, the cohesive response must account for both normal opening and tangential sliding. We therefore extend the concentration analysis to two and three dimensions.} Following \citet{feng2021endowing}, we analyze concentration along the reference normal to a locally planar crack, with locally uniform stress and a common current normal $\bmn$ throughout the band. Nominal traction is measured per unit reference area.\footnote{{Nominal strength can be understood through an array of springs or molecular bonds bridging a material interface, each with a prescribed objective critical force. For a given material patch, the initial number per reference area remains fixed even when the current area changes. The nominal strength therefore measures the original force capacity per reference area, while damage describes its degradation. This interpretation is consistent with defining cohesive energy per reference area and obtaining nominal traction by differentiation with respect to separation.}} Its components and the displacement jump are resolved relative to the current crack plane. The normal opening and tangential sliding magnitude are
\begin{equation}\label{eq:mixed-mode-separation-components}
 \Deltaperp\coloneqq\jumpu\cdot\bmn\geq0,
 \qquad\Deltaparallel\coloneqq\lVert\bmItan\jumpu\rVert.
\end{equation}
Normalize these components using the critical pure-mode nominal tractions and characteristic separations defined in Section~\ref{sec:energy}:
\begin{equation}\label{eq:mixed-mode-component-normalization}
 \normalizedtraction^{\normalmode}\coloneqq\frac{\Tperp}{\Tperpc},
 \qquad
 \normalizedtraction^{\tangentialmode}\coloneqq\frac{\Tparallel}{\Tparallelc},
 \qquad
 \normalizedopening^{\normalmode}\coloneqq\frac{\Deltaperp}{\Deltaperpch},
 \qquad
 \normalizedopening^{\tangentialmode}\coloneqq\frac{\Deltaparallel}{\Deltaparallelch}.
\end{equation}
{The quadratic cohesive energy and strength surface suggest combining these normalized components into the effective traction and separation}
\begin{equation}\label{eq:mixed-mode-effective-variables}
 \normalizedeffectivetraction\coloneqq
 \sqrt{\positivepartexpr{\normalizedtraction^{\normalmode}}^2
       +(\normalizedtraction^{\tangentialmode})^2},
 \qquad
 \normalizedeffectiveopening\coloneqq
 \sqrt{(\normalizedopening^{\normalmode})^2
       +(\normalizedopening^{\tangentialmode})^2}.
\end{equation}
At the concentrated interface, traction balance~\eqref{eq:traction-equilibrium} gives $\Deltaperp=\varpi\positivepartexpr{\Tperp}/\kperp$ and $\Deltaparallel=\varpi\Tparallel/\kparallel$. Together with $\kperp\Deltaperpch=\Tperpc$ and $\kparallel\Deltaparallelch=\Tparallelc$, these imply $\normalizedeffectiveopening=\varpi\normalizedeffectivetraction$. Normal compression is carried by contact and does not enter the effective traction; the relations also apply to closed sliding.

On the damage-loading branch, the strength surface~\eqref{eq:omega-strength-surface} gives $\normalizedeffectivetraction=1-\phi^*$. Integrating phase-field equilibrium along the reference normal yields $\Omega(\phi^*)=(1-\phi^*)\varpi$, the same concentration relation as Eq.~\eqref{eq:1d-singular-phase-balance}. The mixed-mode cohesive law is therefore
\begin{equation}\label{eq:mixed-mode-parametric-cohesive-law}
 \boxed{\normalizedeffectivetraction=1-\phi^*,
 \qquad\normalizedeffectiveopening=\Omega(\phi^*),
 \qquad0\leq\phi^*<1.}
\end{equation}
Mixed-mode opening and sliding share the dimensionless pure-normal softening envelope: $\beta$ sets the relative scales of the two modes and $\Omega$ sets the normalized softening shape. This parameterization follows \citet{feng2021endowing}.

The same softening envelope can also be constructed by fixing the local phase-field term at $\phi^2$ and varying the coupling between phase-field and compliance damage \citep{feng2025JMPS}. Characteristic functions for common cohesive laws and these alternative constructions are given in \ref{app:cohesive-laws}.

\section{Numerical implementation}\label{sec:numerical-implementation}
We describe a finite element implementation in the reference configuration for $d=2,3$. Displacement and phase-field are global unknowns. Compliance damage $\omega$ and the reference crack normal $\bmN$ are updated at Gauss points, while crack strain $\bmw$ is condensed locally. Displacement, $\omega$, and $\bmN$ are all solved using the reduced energy~\eqref{eq:w-reduced-energy}: every change in $\bmF$, $\omega$, or $\bmN$ requires recomputing $\bmw$ and its derivatives. The staggered solution is described below for displacement-controlled loading.

\subsection{Spatial discretization and reduced potential}\label{sec:FE-discretization}
Consider conforming first-order isoparametric elements: bilinear quadrilaterals in two dimensions and trilinear hexahedra in three. Continuous shape functions $\mathcal N_a$ interpolate displacement $\bmu=\bmchi-\bmX$ and phase-field:
\begin{equation}\label{eq:FE-interpolation}
 \bmu(\bmX)\approx\sum_a\mathcal N_a(\bmX)\bmu_a,
 \qquad
 \phi(\bmX)\approx\sum_a\mathcal N_a(\bmX)\phi_a,
\end{equation}
where $a$ indexes nodes, and $\bmu_a$ and $\phi_a$ are nodal values. In a Cartesian basis shared by the reference and current configurations, the same interpolation gives the deformation gradient, $\bmF=\bmIref+\nabla\bmu$.

We store the current and previous accepted values of $\omega$ and $\bmN$ at each Gauss point. The crack strain $\bmw$ is obtained by local condensation at the current state. These three fields require no continuous interpolation and introduce no additional global degrees of freedom into the displacement and phase-field subproblems.
In this discretization, Dirac concentration is approximated by $\omega$ and $\bmw$ within a narrow band of elements. Resolution of its width $h$ is limited by the local reference mesh size $\meshsize$. For example, the one-dimensional construction in \ref{app:model-concentration} has $\omega^h=\ell\varpi/h$ and $\w^h=\Delta/h$ inside the band. Their amplitudes grow as the band narrows, while their integrals remain $\ell\varpi$ and $\Delta$. Together with Section~\ref{sec:energy-comparison}, this suggests that mesh refinement toward $\meshsize\to0$ on a localized cracking branch produces concentration of $\omega$ and $\bmw$ and a finite displacement jump. The fixed material length $\ell$ continues to govern the phase-field transition.

Applying two-point Gauss quadrature in each coordinate direction of every reference element $e$ gives the discrete reduced potential~\eqref{eq:reduced-total-energy}:
\begin{equation}\label{eq:FE-reduced-potential}
 \energyFE
 =\sum_e\sum_q
 \left[
 \psired(\bmF_{eq},\bmN_{eq},\omega_{eq})
 +\varphif(\phi_{eq},(\nabla\phi)_{eq},\omega_{eq})
 \right]\quadratureweight_q\det\elementJacobian_{eq}
 +\mathcal P_{\mathrm{ext}}^{\mathrm{FE}},
\end{equation}
where $q$ indexes quadrature points and $\quadratureweight_q>0$ is the Gauss weight on the parent element. The matrix $\elementJacobian_{eq}$ is the Jacobian of the mapping from the parent element to the reference element, with $\det\elementJacobian_{eq}>0$. Prescribed displacement values are eliminated from the global unknowns; other boundary conditions enter through the potential and admissible set.

At fixed Gauss-point values of $\omega$ and $\bmN$, displacement variation of this discrete potential gives mechanical equilibrium, $\firstvariation{\bmu}\energyFE=0$. For every test perturbation $\delta\bmu$ satisfying homogeneous displacement boundary conditions,
\begin{equation}\label{eq:FE-mechanical-residual}
 \firstvariation{\bmu}\energyFE[\delta\bmu]
 =\sum_e\sum_q\left[\bmP_{eq}:(\nabla\delta\bmu)_{eq}\right]\quadratureweight_q\det\elementJacobian_{eq}
 +\firstvariation{\bmu}\mathcal P_{\mathrm{ext}}^{\mathrm{FE}}[\delta\bmu],
 \qquad\text{with}\quad
 \bmP=\frac{\partial\psired}{\partial\bmF}.
\end{equation}
The local constraints depend on $\bmF$, so the envelope derivative in Eq.~\eqref{eq:FE-mechanical-residual} includes the contact contribution. The returned total stress is precisely $\bmPb+\bmPgeo$ in Eq.~\eqref{eq:first-Piola-stress}.

At fixed $\omega$, the phase-field subproblem involves only fracture energy. Its first variation and tangent on a smooth branch are
\begin{subequations}\label{eq:FE-phase-variations}
\begin{align}
 &\firstvariation{\phi}\energyFE[\delta\phi]
 =\frac{\Gc}{\ell}\sum_e\sum_q
 \left\{
 \left[\Omega\Omega'-2(1-\phi)\omega\right]\delta\phi
 +\ell^2\nabla\phi\cdot\nabla\delta\phi
 \right\}_{eq}\quadratureweight_q\det\elementJacobian_{eq},\label{eq:FE-phase-residual}\\
 &\firstvariation{\phi}^2\energyFE[\delta\phi,\Delta\phi]
 =\frac{\Gc}{\ell}\sum_e\sum_q
 \left\{
 \left[(\Omega')^2+\Omega\Omega''+2\omega\right]\delta\phi\Delta\phi
 +\ell^2\nabla\delta\phi\cdot\nabla\Delta\phi
 \right\}_{eq}\quadratureweight_q\det\elementJacobian_{eq}.\label{eq:FE-phase-tangent}
\end{align}
\end{subequations}
Here $\Omega$ and its derivatives are evaluated at $\phi_{eq}$, and $\Delta\phi$ is the phase-field correction. The phase-field minimizes energy over its admissible set, with constrained degrees of freedom handled by an active-set method. Over the domain of $\Omega$, conditions~\eqref{eq:Omega-conditions} ensure convexity in the phase-field at fixed $\omega$. If the discrete admissible set is also convex, the subproblem is a convex minimization problem. A linear $\Omega$ gives a quadratic energy.

\subsection{Local condensation of crack strain and consistent tangent}\label{sec:local-condensation}
{Evaluating the reduced potential and its stress at each Gauss point requires local equilibration of the crack strain. Given $\bmF$, $\bmN$, and $\omega>0$, the local solve enforces traction balance~\eqref{eq:w-equilibrium} and contact complementarity~\eqref{eq:KKT_lambda}. We therefore define the traction residual by}
\begin{equation}\label{eq:local-return-residual}
 \mathbm r_{\bmw}
 \coloneqq\bmKomega\bmw-\bmPb(\bmF-\bmw\otimes\bmN)\bmN
 +\lambdaperp\bmn.
\end{equation}
On the open branch, set $\lambdaperp=0$ and solve $\mathbm r_{\bmw}=\mathbm0$. If the computed normal opening is negative, solve again on the closed branch $\bmw\cdot\bmn=0$ and recover a nonpositive contact multiplier. Both branches must satisfy $\Jb>0$. Energy descent and backtracking keep trial states within the orientation-preserving domain.\footnote{The local iteration uses the scaled unknown $\bmw/\omega$ to avoid ill-conditioning from $1/\omega$ in $\bmKomega$ as $\omega$ approaches zero. Its zero-damage limit follows from Eq.~\eqref{eq:intact-w-asymptotics}. The actual crack strain $\bmw$ is recovered from the scaled variable and set to $\mathbm0$ at $\omega=0$.} Each material-point solve involves only $d$ components of $\bmw$ and, under closure, one contact multiplier, independently of the global mesh size.

To construct the local iteration and condensed tangent, define the bulk material tangent $\bulkmodulus\coloneqq\partial\bmPb/\partial\bmFb$. The local tangent $\localcracktangent=\partial_{\bmw}\mathbm r_{\bmw}$ with respect to crack strain $\bmw$ acts on an arbitrary spatial vector $\bmz$ as
\begin{equation}\label{eq:local-return-tangent}
 \localcracktangent\bmz
 =\bmKomega\bmz+
 \left[\bulkmodulus:(\bmz\otimes\bmN)\right]\bmN.
\end{equation}
This tensor combines cohesive stiffness with the bulk response to local cracking. We linearize about a stable local minimum, with $\localcracktangent$ positive definite on the open branch or on $\bmn^\perp$ under closure.

At fixed $\bmN$, implicit differentiation of the local equations gives the changes in $\bmw$ and the contact multiplier with $\bmF$ and $\omega$. On the closed branch, these increments satisfy
\begin{equation}\label{eq:local-return-linearization}
 \underbrace{\begin{bmatrix}
 \localcracktangent&\bmn\\
 \bmn^\trans&0
 \end{bmatrix}}_{\localKKTmatrix}
 \begin{bmatrix}\delta\bmw\\\delta\lambdaperp\end{bmatrix}
 =-\begin{bmatrix}
 (\partial_{\bmF}\mathbm r_{\bmw}):\delta\bmF
 +(\partial_\omega\mathbm r_{\bmw})\delta\omega\\
 \bmw\cdot\delta\bmn
 \end{bmatrix}.
\end{equation}
Partial derivatives on the right hold $\bmw$ and $\lambdaperp$ fixed; $\delta\bmn$ is evaluated from Eq.~\eqref{eq:normal-full-variation} using $\delta\bmN=\mathbm0$. The open branch retains only the first row and sets $\delta\lambdaperp=0$. Displacement linearization uses $\delta\omega=0$, while linearization in $\omega$ uses $\delta\bmF=\mathbm0$. After a change in the contact active set, both the local solution and its linearization are recomputed on the new branch.

Displacement and $\omega$ are updated separately, requiring second partial derivatives of the reduced energy in $\bmF$ and $\omega$, respectively. Comma subscripts below denote partial derivatives with the other independent variables fixed, evaluated at the local solution. Writing $\alpha\in\{\bmF,\omega\}$, the open branch satisfies
\begin{equation}\label{eq:open-condensed-Hessian}
 (\psired)_{,\alpha\alpha}
 =\psi_{,\alpha\alpha}
 -\psi_{,\alpha\bmw}\localcracktangent^{-1}\psi_{,\bmw\alpha}.
\end{equation}
On the closed branch, use the constrained energy $\psilambda=\psi+\lambdaperp\wperp$ to obtain
\begin{equation}\label{eq:closed-condensed-Hessian}
 (\psired)_{,\alpha\alpha}
 =(\psilambda)_{,\alpha\alpha}-
 \begin{bmatrix}(\psilambda)_{,\alpha\bmw}&(\wperp)_{,\alpha}\end{bmatrix}
 \localKKTmatrix^{-1}
 \begin{bmatrix}(\psilambda)_{,\bmw\alpha}\\(\wperp)_{,\alpha}\end{bmatrix}.
\end{equation}
Products in Eqs.~\eqref{eq:open-condensed-Hessian} and~\eqref{eq:closed-condensed-Hessian} denote composition of linear maps. The latter also includes the deformation dependence of the closure constraint. Let $\algmodulus\coloneqq\partial^2\psired/\partial\bmF^2$ denote the consistent material tangent. The second variation for the displacement subproblem is
\begin{equation}\label{eq:algorithmic-material-tangent}
\firstvariation{\bmu}^2\energyFE[\delta\bmu,\Delta\bmu]
 =\sum_e\sum_q
 \left[(\nabla\delta\bmu)_{eq}:(\algmodulus)_{eq}:(\nabla\Delta\bmu)_{eq}\right]\quadratureweight_q\det\elementJacobian_{eq},
\end{equation}
where $\Delta\bmu$ is the displacement correction, while $\omega$ and $\bmN$ are fixed in this subproblem. The prescribed-load potential has no second variation. This condensation follows the structure of consistent linearization in return-mapping algorithms \citep{simo2006computational}. Local variables are absent from the global mechanical unknowns, but their response to trial deformation is retained in the tangent. Thus $\algmodulus$ accounts for the local equilibration of $\bmw$, the deformation dependence of $\bmn$, and the geometric Piola stress $\bmPgeo$. Neither a tangent at fixed $\bmw$ nor $\bulkmodulus$ alone is sufficient.

At $\omega=0$, energy and stress are $\psib(\bmF)$ and $\bmPb(\bmF)$, and the mechanical tangent reduces to the bulk tangent. The first derivative in $\omega$ follows the one-sided expansion in \ref{app:damage-initiation}. The required curvature is computed from the one-sided branch limit of the scaled local equations, keeping near-zero $\omega$ updates consistent with initiation.

\subsection{Gauss-point updates of compliance damage and reference crack normal}\label{sec:local-damage-direction}
\subsubsection{Local energy minimization for compliance damage}\label{sec:omega-update}
At fixed $\bmF$, $\bmN$, and $\phi$, the $\omega$ subproblems are independent at individual Gauss points. Omitting terms independent of $\omega$ gives the local energy
\begin{equation}\label{eq:local-omega-potential}
 \omegapotential(\omega)
 \coloneqq\psired(\bmF,\bmN,\omega)
 +\frac{\Gc}{\ell}(1-\phi)^2\omega,
 \qquad \omega\geq\omega_{\mathrm{prev}}.
\end{equation}
The local update therefore follows the damage conditions above, with derivative and curvature
\begin{equation}\label{eq:local-omega-derivatives}
 \frac{\dd\omegapotential}{\dd\omega}
 =\variationalderivative{\energyred}{\omega}
 =-\frac{\Gc}{\ell}\reducedstrength(\bmF,\bmN,\omega,\phi),
 \qquad
 \omegatangent\coloneqq\frac{\dd^2\omegapotential}{\dd\omega^2}
 =\frac{\partial^2\psired}{\partial\omega^2}.
\end{equation}
The scalar $\omegatangent$ is the tangent for $\omega$ after eliminating $\bmw$; it follows from Eq.~\eqref{eq:open-condensed-Hessian} or~\eqref{eq:closed-condensed-Hessian} with $\alpha=\omega$. Every evaluation of a trial $\omega$ requires a new local solution for $\bmw$, rather than reusing $\bmw$ from the previous damage state.

To handle both near-zero and large $\omega$, use logarithmic coordinates $s=\log(1+\omega)$. On the selected smooth local branch, the energy gradient and curvature are
\begin{equation}\label{eq:log-compliance-derivatives}
 g_s\coloneqq\frac{\dd\omegapotential}{\dd s}
 =(1+\omega)\frac{\dd\omegapotential}{\dd\omega},
 \qquad
 h_s\coloneqq\frac{\dd^2\omegapotential}{\dd s^2}
 =(1+\omega)^2\omegatangent+g_s.
\end{equation}
These define the projected descent step
\begin{equation}\label{eq:projected-compliance-step}
 s_{\mathrm{trial}}
 =\operatorname{proj}_{\mathcal I_s}\left(s-\frac{g_s}{m_s}\right),
 \qquad
 m_s=\max\left\{|h_s|,\frac{|g_s|}{\Delta s_{\max}},\varepsilon_m\frac{\Gc}{\ell}\right\},
\end{equation}
where $\mathcal I_s$ is the image in logarithmic coordinates of the admissible interval for $\omega$, and $\operatorname{proj}$ is its orthogonal projection. With only a history lower bound, $\mathcal I_s=[\log(1+\omega_{\mathrm{prev}}),\infty)$. Projection then takes the maximum with that bound. A computational cap $\omega_{\max}$ limits ill-conditioning in highly damaged regions, giving upper endpoint $\log(1+\omega_{\max})$. The parameter $\Delta s_{\max}>0$ limits the logarithmic step, while the small dimensionless $\varepsilon_m>0$ keeps $m_s$ positive; $m_s$ has units of energy density. Backtracking along the projected direction continues until local energy decreases and orientation preservation and contact are satisfied. The modified curvature provides a descent direction for the original energy. Termination is based on the KKT residual over the specified admissible interval. It requires $\dd\omegapotential/\dd\omega\geq0$ at the history lower bound and $\dd\omegapotential/\dd\omega=0$ at an unconstrained interior point. If a finite upper bound is active, the derivative is nonpositive.

\subsubsection{Energy search over multiple candidate directions}\label{sec:direction-search}
At fixed $\bmF$ and $\omega>0$, the reference crack normal $\bmN$ is determined by Eq.~\eqref{eq:N-reduced-minimization}. Each candidate $\bmN$ requires a fresh condensation of $\bmw$, so energies are compared only after the crack strain has equilibrated locally. Near zero $\omega$, the common bulk energy can obscure small energy differences among candidate normals $\bmN$. We therefore use the equivalent scaled objective
\begin{equation}\label{eq:direction-search-potential}
 \psidir(\bmF,\bmN,\omega)\coloneqq
 \begin{cases}
 \displaystyle\frac{\psired(\bmF,\bmN,\omega)-\psib(\bmF)}{\omega},
 &\omega>0,\\[0.8em]
 \partial_\omega^+\psired(\bmF,\bmN,0),&\omega=0
 \end{cases}.
\end{equation}
For $\omega>0$, the subtracted energy and scaling factor are independent of $\bmN$ and leave the optimal direction unchanged. The zero-damage objective follows Eq.~\eqref{eq:intact-directional-energy-expansion} and recovers selection by the maximum driving force. Near zero, an equivalent energy-increment expression avoids directly subtracting nearly equal energies.

For objective isotropic bulk laws, we express directions in the right principal-stretch basis. Proposition~\ref{prop:crack-direction-symmetry-reduction} then restricts the search to the nonnegative part of the unit sphere. One angle suffices in two dimensions. If Proposition~\ref{prop:crack-direction-plane-reduction} applies in three dimensions, use the single-angle principal-plane search~\eqref{eq:N-principal-plane}. The general three-dimensional case requires two angles, for example
\begin{equation}\label{eq:3d-direction-search-parameters}
 \bmN(\theta,\vartheta)
 =\cos\theta\,\bmN_1
 +\sin\theta\cos\vartheta\,\bmN_2
 +\sin\theta\sin\vartheta\,\bmN_3,
 \qquad 0\leq\theta,\vartheta\leq\frac\pi2.
\end{equation}
If the constitutive law lacks these symmetries, search the full unit sphere and identify $\bmN$ and $-\bmN$ with the same crack plane. After symmetry reduction, restore equal-energy mirror directions and choose the branch whose plane is closest to that of the previous direction. This avoids spurious direction jumps caused solely by eigenvector signs or equal-energy multiplicity.

Optimization of $\bmN$ combines systematic sampling with local refinement of multiple candidates. For a single-angle problem in two dimensions or a three-dimensional principal plane, first sample the admissible interval. Retain the previous $\bmN$ and the endpoints, and identify neighborhoods of candidate minima. If directional derivatives at the ends of a candidate interval change from negative to positive, refine the stationary point by bisection. Otherwise, apply the interval-minimization algorithm of \citet{brent1973algorithms} and increase sampling over the full interval. Compare refined candidates, sampled points, and endpoints to select the lowest-energy direction.

For three-dimensional problems without principal-plane reduction, replace intervals with spherical patches. Refine around each sampled minimum under the directional constraint, then compare interior and boundary candidates. When symmetry reduction applies, Eq.~\eqref{eq:3d-direction-search-parameters} defines the parameter domain. Otherwise, patches cover the full unit sphere, with duplicate planes removed using $\bmN\sim-\bmN$.

{Local refinement of candidate directions uses derivatives of the condensed energy along the unit sphere. On a smooth branch with $\omega>0$ and $\bmN\cdot\delta\bmN=0$, the constrained envelope gives}
\begin{equation}\label{eq:reduced-direction-derivative}
\firstvariation{\bmN}\psired[\delta\bmN]=\bmQ\cdot\delta\bmN,
 \qquad
 \firstvariation{\bmN}\psidir[\delta\bmN]=\frac{\bmQ\cdot\delta\bmN}{\omega},
\end{equation}
where $\bmQ$ is defined by Eq.~\eqref{eq:direction-conjugate-force} and accounts for both the local equilibration of $\bmw$ and changes in the contact constraint. At exactly zero damage, differentiate the one-sided objective in Eq.~\eqref{eq:direction-search-potential} directly; substituting $\bmw=\mathbm0$ into $\bmQ$ cannot determine the initiation direction. Candidates at contact transitions are retained by energy comparison rather than stationarity alone.

For Gauss points with $\omega=0$, the energy itself is independent of $\bmN$, and initiation direction follows the one-sided objective~\eqref{eq:direction-search-potential}. To avoid searching every Gauss point, let $k_{\min}\coloneqq\min(\kperp,\kparallel)$ and use the traction $\bmT_0=\bmPb(\bmF)\bmN$ to obtain the conservative strength-function bound
\begin{equation}\label{eq:intact-direction-screening}
 \max_{\bmN\in\mathbb S^{d-1}}\reducedstrength(\bmF,\bmN,0,\phi)
 \leq\frac{\lVert\bmPb(\bmF)\rVert_{\mathrm{op}}^2}{2\Gc k_{\min}}-(1-\phi)^2,
\end{equation}
where $\lVert\cdot\rVert_{\mathrm{op}}$ denotes the largest singular value of the matrix. A negative right-hand side guarantees that no candidate plane has reached the strength surface. No initiation search is then needed. Only points with a nonnegative bound undergo the one-sided search and a renewed check of $\omega$ evolution. Screening is repeated after each deformation or phase-field update. For positive $\omega$, direction updates use the finite-compliance objective.

\subsection{Staggered solution}\label{sec:staggered-solution}
A staggered scheme couples local updates with global finite element equations to trace the equilibrium branch. For each prescribed boundary-displacement step, initialize from the previous accepted state and hold damage-history lower bounds fixed throughout the step. Each staggered iteration solves displacement, $\omega$, and phase-field in turn, then updates $\bmN$, as summarized in Algorithm~\ref{alg:staggered-condensation}.

The displacement subproblem uses Eq.~\eqref{eq:FE-mechanical-residual} as its residual and Eq.~\eqref{eq:algorithmic-material-tangent} as its consistent tangent in Newton's method. Every trial displacement triggers new Gauss-point condensation of $\bmw$ and checks of orientation preservation and contact; an inadmissible trial is shortened.

\begin{algorithm}[Staggered solution with local crack-strain condensation]\label{alg:staggered-condensation}
Given the boundary displacement and previous accepted state, repeat:
\begin{enumerate}
\item Hold $\omega,\bmN,\phi$ fixed and solve for the global displacement. For every trial $\bmF$, condense $\bmw$ locally and return $\psired,\bmP,\algmodulus$.
\item Hold the updated $\bmF$ and $\bmN,\phi$ fixed and solve $\omega$ at each Gauss point. Every trial $\omega$ requires renewed condensation of $\bmw$ and returns the derivative and curvature in Eq.~\eqref{eq:local-omega-derivatives}.
\item Hold $\omega$ fixed and solve the global phase-field over its admissible set. The first three steps may be repeated at fixed $\bmN$.
\item At fixed current $\bmF,\omega$, search for $\bmN$ as in Section~\ref{sec:direction-search}. For $\omega>0$, condense $\bmw$ for each candidate $\bmN$. For points with $\omega=0$, screen using Eq.~\eqref{eq:intact-direction-screening} and apply the one-sided objective at potential initiation points.
\item Reevaluate mechanical equilibrium, damage conditions, and the directional energy gap for $\bmN$ using the retained fields. Accept the loading step and commit the history only when all tolerances are met; otherwise continue the staggered iterations.
\end{enumerate}
\end{algorithm}

Convergence is assessed using residuals for mechanical equilibrium, local traction balance, and damage complementarity \citep{farrell2017linear}. Directional convergence for $\bmN$ is assessed by the total energy reduction available among the searched candidates relative to the current normal:
\begin{equation}\label{eq:numerical-direction-energy-gap}
 \Delta\energy_{\bmN}
 \coloneqq\sum_e\sum_q
 \left[
 \psired(\bmF_{eq},\bmN_{eq},\omega_{eq})
 -\psired(\bmF_{eq},\bmN_{eq}^{\mathrm{cand}},\omega_{eq})
 \right]_{+}\quadratureweight_q\det\elementJacobian_{eq},
\end{equation}
where $\bmN^{\mathrm{cand}}$ is the lowest-energy candidate returned by the current search. If this gap is below the direction tolerance, the existing $\bmN$ may be retained at damaged points; points with zero or near-zero $\omega$ are still updated as in Section~\ref{sec:direction-search}. Field residuals are recomputed after every $\bmN$ update, and the current state is accepted only when all tolerances are met. The examples specify the meshes and loading controls; mesh and phase-field-length sensitivity are examined in Section~\ref{sec:example1}.

{The basic staggered scheme can be complemented by methods that accelerate convergence or account more directly for interfield coupling. Anderson acceleration, as used by \citet{storvik2021accelerated}, can improve convergence of staggered iterations.} \citet{lampron2021efficient} developed monolithic modified-Newton solvers. \citet{farrell2017linear} combined alternate minimization with coupled Newton iterations for phase-field fracture. For coupled fields, the coupling-strength analysis of \citet{zhao2025staggered} can guide update ordering and field grouping. \citet{jiang2026schur} incorporated interfield coupling into subfield updates through Schur-complement-enhanced iterations, improving computational efficiency in phase-field fracture.

For the adaptive examples below, we follow the multilevel refinement strategy of \citet{feng2026multilevel}, using $\phi$ and $\omega$ to identify the phase-field transition and damage-concentration regions, respectively. Finite-mesh corrections are specified in the relevant examples.

{Beyond this refinement strategy, adaptive remeshing can improve local resolution and element quality. \citet{cheng2025global} developed global--local adaptive meshing with nodal density controlled by phase-field and displacement errors.} \citet{cheng2026selfcorrecting} developed three-dimensional self-correcting remeshing to balance local resolution and element quality.
\section{Numerical examples}\label{sec:numerical-examples}
\subsection{Uniaxial tension: independent bulk and cohesive constitutive laws}\label{sec:example1}
We first use uniaxial tension to test whether the finite-deformation $\Omega^2$ model reproduces independently prescribed bulk and cohesive laws. This example also illustrates how damage concentration and displacement discontinuities coexist with a continuous phase-field. Neo--Hookean (NH), Mooney--Rivlin (MR), and Saint Venant--Kirchhoff (SVK) bulk laws are each combined with linear, exponential, and quadratic PPR softening, giving a $3\times3$ test matrix. Bulk and cohesive laws are listed in \ref{app:bulk-constitutive} and~\ref{app:cohesive-laws}, respectively; the PPR law uses $\alpha=3$.

A rectangular specimen of length $L=100\,\mathrm{mm}$ and height $H=2\,\mathrm{mm}$ is analyzed in two dimensions with reference thickness $1\,\mathrm{mm}$. The axial displacement is fixed at $u_x=0$ on the left, and a monotonically increasing displacement $u_x=\bar u$ is applied on the right. Transverse motion is free at both ends, except that $u_y=0$ is imposed at the lower-left corner to remove rigid translation. The upper and lower edges are traction-free, permitting lateral contraction. The specimen is initially intact. The phase-field is fixed at $\phi=0$ at both ends, with the natural condition~\eqref{eq:phi-boundary-KKT} on the upper and lower edges.

All nine calculations use the same mesh of 504 bilinear quadrilateral elements. The central axial mesh size is $0.1\,\mathrm{mm}$ and the transverse size $0.5\,\mathrm{mm}$; elements become coarser in the axial direction toward the ends. This constitutive comparison uses a fixed, locally refined mesh and displacement increments of $4\,\mathrm{mm}$.

All calculations use shear modulus $\mu=0.0315\,\mathrm{N/mm^2}$, fracture energy $\Gc=10\,\mathrm{N/mm}$, normal cohesive strength $\Tperpc=0.05\,\mathrm{N/mm^2}$, phase-field length $\ell=5\,\mathrm{mm}$, and tensile-to-shear strength ratio $\beta=0.5$. NH and MR use incompressible plane-stress reduction, with MR coefficients $c_{10}=0.8$ and $c_{01}=0.2$. SVK uses the two-dimensional energy~\eqref{eq:bulk-StVK-energy-stress} with Poisson's ratio $\nu=0.4$ and Lam\'e constant $\lambda=2\mu\nu/(1-\nu)$. The relatively large fracture energy makes the softening opening large compared with the bulk elastic extension. The response can thus be tracked under ordinary displacement control without arc-length methods.

Figure~\ref{fig:example1-stress-displacement} plots axial nominal stress $P$ against displacement $\bar u$ for all nine combinations. The stress $P$ is the right-end reaction divided by the reference cross-sectional area. Rows correspond to bulk laws and columns to cohesive laws. Black solid curves show analytical bulk loading responses; the uniaxial relations are given in Eqs.~\eqref{eq:bulk-NH-uniaxial}, \eqref{eq:bulk-MR-uniaxial}, and~\eqref{eq:bulk-StVK-uniaxial}. Before initiation, $\Fb=1+\bar u/L$ and the bulk potential governs the response. The three bulk laws exhibit different nonlinear elastic behavior, and all three loading curves for each law coincide with its analytical curve. Post-initiation softening follows the prescribed cohesive law.

\begin{figure}[!htbp]
 \centering
 \includegraphics[width=\linewidth]{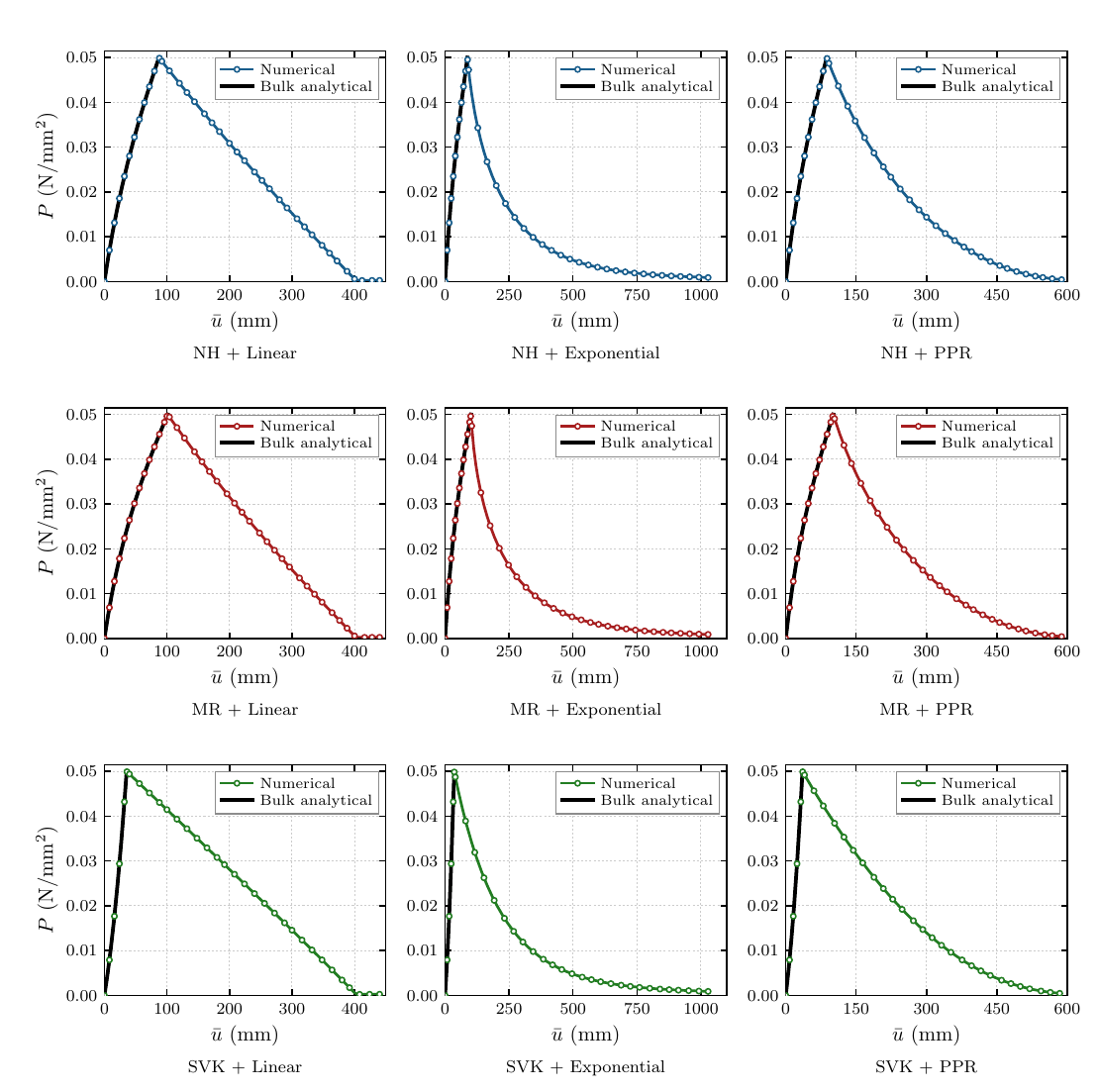}
\caption{Nominal stress--displacement responses for nine constitutive combinations. NH, MR, and SVK denote Neo--Hookean, Mooney--Rivlin, and Saint Venant--Kirchhoff materials. Black curves: analytical bulk response.}
 \label{fig:example1-stress-displacement}
\end{figure}

Figure~\ref{fig:example1-fields} shows the NH/linear-cohesive case during softening at a nominal traction of approximately $\Tperpc/2$, corresponding to $\bar u=236\,\mathrm{mm}$ and $P/\Tperpc=0.502$. The phase-field $\phi$ decays smoothly on either side of the crack, retaining a transition width controlled by $\ell$. Compliance damage $\omega$ and the axial displacement gradient concentrate within a narrow region of width approximately $0.1\,\mathrm{mm}$. Although the peak $\phi$ is only about $0.5$, well below the complete-fracture value $\phi=1$, the displacement already shows a pronounced jump.

Concentration begins near initiation. In the first post-initiation state $\bar u=88\,\mathrm{mm}$, the peak $\phi$ is only $2.65\times10^{-3}$, yet $\omega$ has already localized in the narrow central element band. Consistent with Section~\ref{sec:energy-comparison}, this supports the emergence of a strong discontinuity at crack nucleation rather than only at complete phase-field damage. Continuous finite elements approximate the jump through a steep displacement variation across a narrow band. Conventional stiffness-degradation phase-field models spread crack deformation over a finite-width damaged zone. Here $\omega$ and $\bmw$ concentrate, allowing a diffuse phase-field to coexist with a displacement discontinuity.

\begin{figure}[!htbp]
 \centering
 \includegraphics[width=\linewidth]{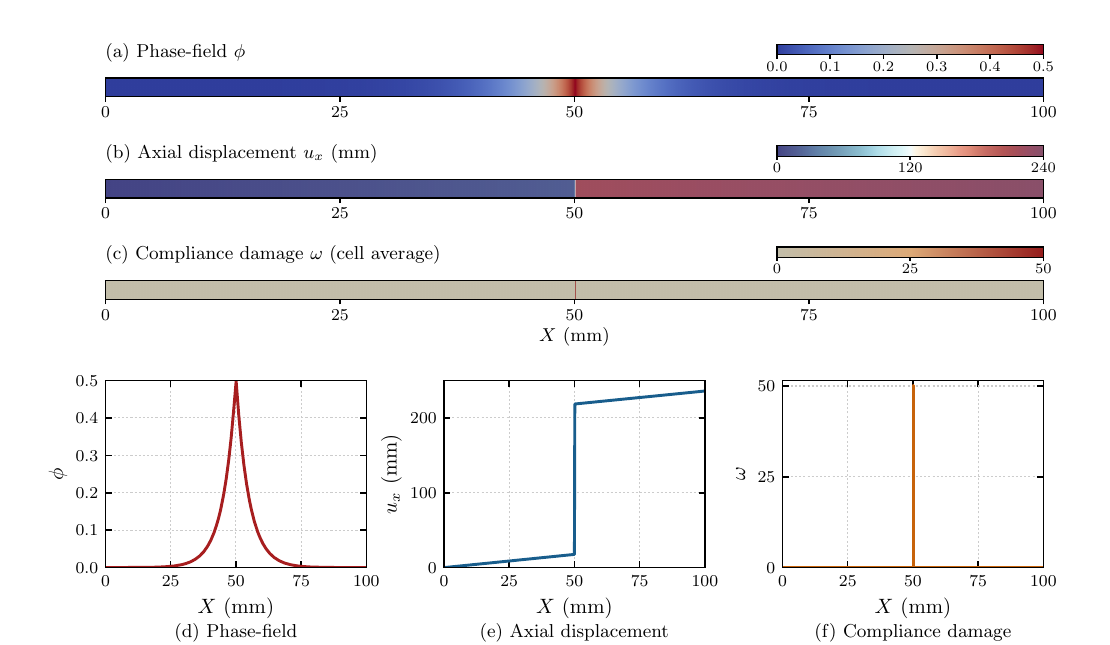}
\caption{Fields for the NH/linear-cohesive case at $\bar u=236\,\mathrm{mm}$: (a--c) phase-field $\phi$, axial displacement $u_x$, and element-averaged compliance damage $\omega$ in the reference configuration; (d--f) centerline profiles.}
 \label{fig:example1-fields}
\end{figure}

{To isolate cohesive separation from the overall response, we subtract the bulk elastic extension and obtain the crack opening $\Delta=\bar u-L[\Fb(P)-1]$.} Here $\Fb(P)$ is the inverse uniaxial bulk relation, and $\Tperp=P$ in pure opening. We use the normalized quantities $\normalizedtraction=\Tperp/\Tperpc$ and $\normalizedopening=\Delta/\Deltaperpch$ from Eq.~\eqref{eq:1d-parametric-cohesive-law}, with $\Deltaperpch=2\Gc/\Tperpc=400\,\mathrm{mm}$. Figure~\ref{fig:example1-cohesive-laws} groups the traction--opening curves by target softening law. Results for all three bulk laws agree closely with their target cohesive laws. Across all softening states compared at equal opening, the maximum absolute traction error normalized by $\Tperpc$ is no greater than $1.36\%$.

\begin{figure}[!htbp]
 \centering
 \includegraphics[width=\linewidth]{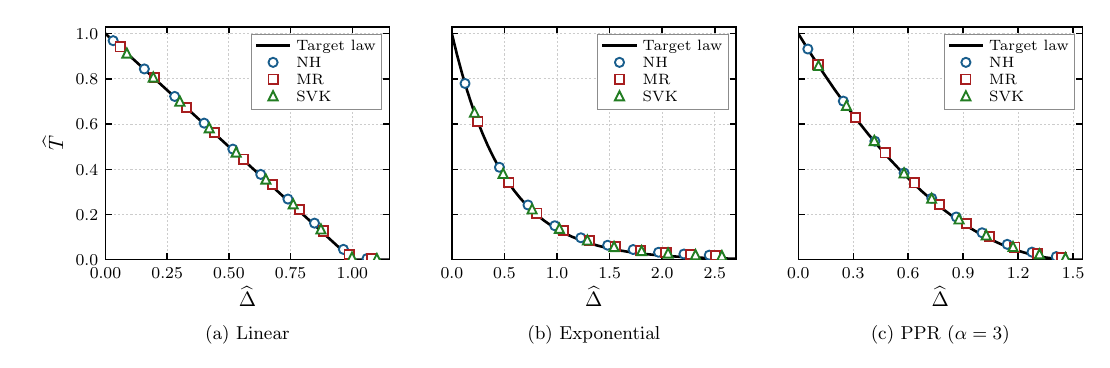}
\caption{Dimensionless traction--opening relations for nine constitutive combinations: (a) linear, (b) exponential, and (c) PPR softening. Symbols: numerical results; solid curves: target cohesive laws.}
 \label{fig:example1-cohesive-laws}
\end{figure}

We examine mesh and phase-field length effects for NH elasticity with linear cohesive softening. In Fig.~\ref{fig:example1-mesh-length-sensitivity}(a,b), fix $\ell=5\,\mathrm{mm}$ and use central axial sizes $\meshsize=0.2,0.1,0.05\,\mathrm{mm}$, or $\meshsize/\ell=1/25,1/50,1/100$, giving 158, 504, and 1792 elements. Peak nominal stresses differ by less than $0.005\%$, and all ascending branches coincide with the analytical bulk response. The softening response converges with refinement; the medium-mesh curve differs from the finest by approximately $0.60\%$.\footnote{The relative $L^2$ difference is $\lVert P-P_{\mathrm{ref}}\rVert_{L^2}/\lVert P_{\mathrm{ref}}\rVert_{L^2}$, integrated over loading displacement $\bar u$. Mesh comparisons use the finest mesh as reference over $0\leq\bar u\leq436\,\mathrm{mm}$. Length-scale comparisons below use $\ell=5\,\mathrm{mm}$ as reference over $0\leq\bar u\leq404\,\mathrm{mm}$.} All traction--opening curves approach the target linear law as the mesh is refined.

Figure~\ref{fig:example1-mesh-length-sensitivity}(c,d) fixes $\meshsize/\ell=1/50$ and compares $\ell=1,2,5\,\mathrm{mm}$, with central axial sizes $0.02,0.04,0.1\,\mathrm{mm}$. Over the common loading interval $0\leq\bar u\leq404\,\mathrm{mm}$, the nominal stress--displacement curves and cohesive laws nearly coincide; peak stresses differ by less than $0.002\%$. Relative to $\ell=5\,\mathrm{mm}$, the $\ell=1,2\,\mathrm{mm}$ curves have relative $L^2$ differences of approximately $0.020\%$ and $0.016\%$, respectively. This agrees with Section~\ref{sec:cohesive-law}: $\ell$ controls the diffuse phase-field transition, while the target cohesive law is independent of $\ell$.

\begin{figure}[!htbp]
 \centering
 \includegraphics[width=0.730\linewidth]{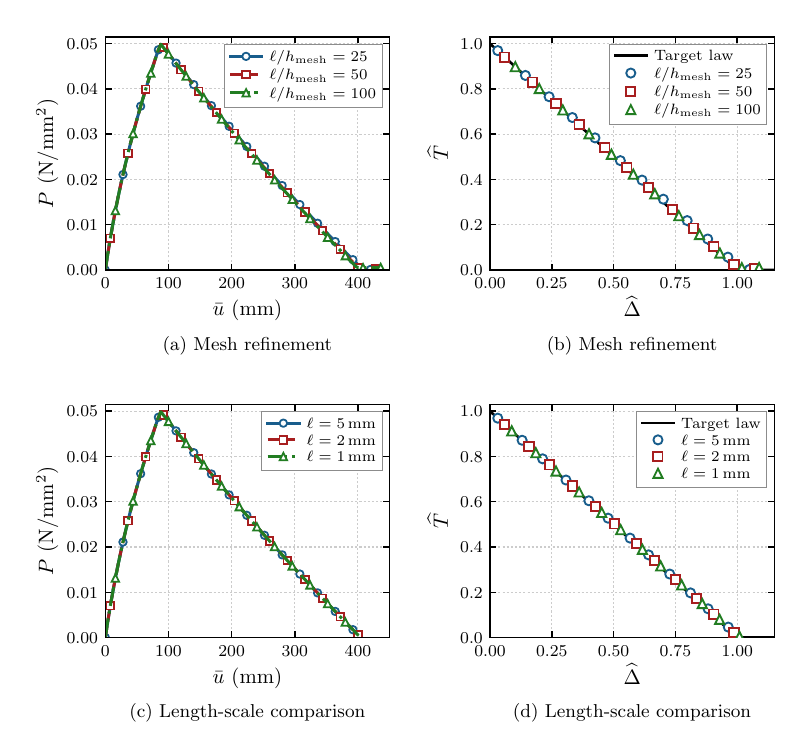}
\caption{Sensitivity of the NH/linear-cohesive case to (a,b) mesh size and (c,d) phase-field length. (a,c) Nominal stress--displacement; (b,d) traction--opening. Black curves: target cohesive law.}
 \label{fig:example1-mesh-length-sensitivity}
\end{figure}

{To test the energetic connection in Section~\ref{sec:energy-comparison}, we also compare bulk elastic energy and equivalent cohesive-interface potential with analytical CZM solutions. Figure~\ref{fig:example1-energy} reports this comparison for NH elasticity with linear, exponential, and PPR laws.} All three cases use $\meshsize/\ell=1/100$ in the concentration region. Numerical energies are $\energy_{\mathrm b}=\int_{\mathcal B}\psib\,\dd V$ and $\energy_{\Gamma}=\int_{\mathcal B}(\psic+\varphif)\,\dd V$; the latter includes cohesive stored energy and fracture energy. The analytical curves follow the single-crack loading branch~\eqref{eq:1d-CZM-functional}, with interface potential $A_0g(\Delta)$ and reference area $A_0=2\,\mathrm{mm^2}$. All three cohesive laws show close agreement.

\begin{figure}[!htbp]
 \centering
 \includegraphics[width=\linewidth]{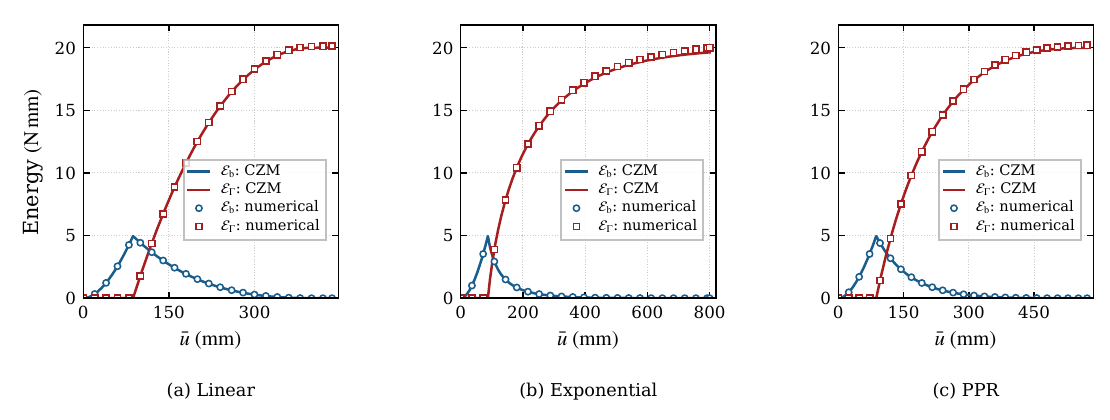}
\caption{Energy evolution for NH elasticity with (a) linear, (b) exponential, and (c) PPR cohesive laws. Solid curves: analytical CZM solutions; symbols: numerical results.}
 \label{fig:example1-energy}
\end{figure}

\FloatBarrier
\subsection{Peeling: large rotations and mixed-mode cohesive response}\label{sec:example2}
We use the two-dimensional peeling geometry reported by \citet{tian2022mixed} and also considered by \citet{tian2025large}. Constitutive laws and parameters are selected for the present model. Bending and large rotation of the upper arm generate both normal tension and tangential shear at the interface, allowing mixed-mode cohesion to be tested during crack growth. We examine the overall reaction and field distributions, then sample tractions and separations on a fixed reference section. These local quantities are compared directly with the dimensionless law of Section~\ref{sec:cohesive-law}.

Figure~\ref{fig:example2-geometry-mesh} shows the geometry, boundary conditions, and mesh. The reference length is $L=80\,\mathrm{mm}$ and height $H=16\,\mathrm{mm}$, with unit reference thickness. A starter crack of length $a_0=20\,\mathrm{mm}$ lies on the midline $Y=H/2$, with a prescribed propagation interface ahead of it. The bottom edge is fixed, $u_x=u_y=0$. An upward displacement $u_y=\bar u$ is applied over the upper-left segment of width $1\,\mathrm{mm}$, with horizontal motion free; other boundaries are traction-free. Initially $\omega=\phi=0$, and the phase-field satisfies the natural condition~\eqref{eq:phi-boundary-KKT}. The prescribed interface has reference normal $\bmN=(0,1)^\trans$, with local directional relaxation allowed near the initial tip.\footnote{Related discussions of diffuse damage and energy relaxation near a crack tip are given in \citet{Braides1995,rodella2026sharp}.}

\begin{figure}[!htbp]
 \centering
 \includegraphics[width=0.640\linewidth]{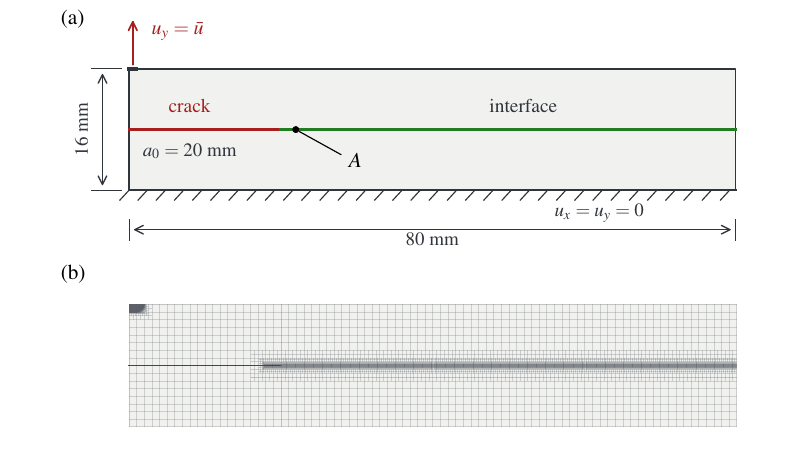}
\caption{Peeling specimen: (a) reference geometry and boundary conditions; (b) mesh. Red and green lines mark the initial crack and propagation interface; $A$ is the fixed sampling section.}
 \label{fig:example2-geometry-mesh}
\end{figure}

The bulk follows incompressible plane-stress NH elasticity~\eqref{eq:bulk-incompressible-NH} and the cohesive law is linear, $\Omega(\phi)=\phi$. Parameters are $\mu=0.03\,\mathrm{N/mm^2}$, $\Gc=0.01\,\mathrm{N/mm}$, $\Tperpc=0.01\,\mathrm{N/mm^2}$, $\ell=0.2\,\mathrm{mm}$, and $\beta=1$, with normal and tangential characteristic separations both $2\,\mathrm{mm}$. The mesh has 14618 bilinear quadrilaterals, with size $\meshsize=0.0625\,\mathrm{mm}$ near the interface, graded to $1\,\mathrm{mm}$ away from it, and additional refinement at the loading segment. This fixed, locally refined mesh uses an initial load increment $0.5\,\mathrm{mm}$, reduced when needed.\footnote{The computational compliance-damage cap is $\omega=10^4$. At $\bar u=15.5\,\mathrm{mm}$, one distorted, already damaged element near the initial tip was deactivated and equilibrium recomputed at the same displacement; the element remains in the output and field plots.}

Figure~\ref{fig:example2-fields} shows deformations and fields at $\bar u=25$, $44.75$, and $60\,\mathrm{mm}$. The upper arm bends and rotates as the opened region advances rightward. The highly damaged, nearly separated band stretches with the opening, while the reference phase-field retains a narrow continuous transition. Compliance damage concentrates near the interface and advances with the crack. Concentrated compliance damage and crack strain thus accommodate the widening opening without comparable widening of the reference band. All panels use the same spatial scale and actual deformation, with all output elements retained. The common compliance-damage color limit is $\omega_{\max}=10^4$.

\begin{figure}[!htbp]
 \centering
 \includegraphics[width=\linewidth]{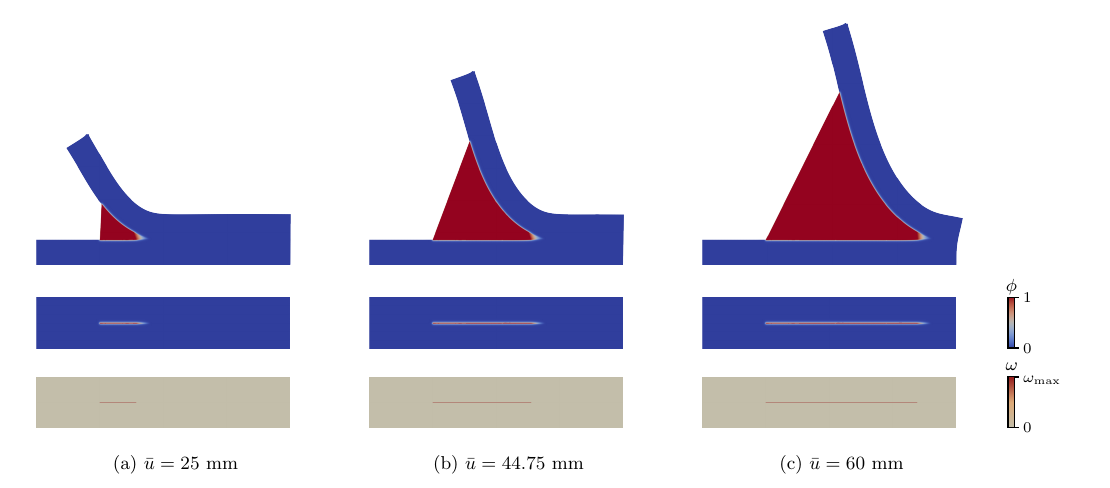}
\caption{Peeling fields. Rows: current-configuration phase-field $\phi$, reference-configuration phase-field $\phi$, and element-averaged compliance damage $\omega$. Columns correspond to the indicated loading displacements.}
 \label{fig:example2-fields}
\end{figure}

Figure~\ref{fig:example2-response} gives the overall and local responses. In Fig.~\ref{fig:example2-response}(a), the vertical reaction $F_y$ is plotted against displacement $\bar u$. It reaches $F_y\approx1.68\times10^{-2}\,\mathrm{N}$ at $\bar u\approx14\,\mathrm{mm}$, then decreases and remains at a comparable level during sustained propagation. Late in propagation, it drops substantially and subsequently rises. In Fig.~\ref{fig:example2-response}(b), crack length $a$ is the reference horizontal coordinate of the rightmost point on the interface at which $\phi=0.95$, or $a=a_0$ before such a point exists. It tracks the advance of the nearly separated region from the initial tip toward the right end.\footnote{Path tracking ends as the crack passes through the specimen. The separated upper arm retains rigid-body degrees of freedom unconstrained by the displacement boundary conditions; its subsequent motion is not followed.}

We test local mixed-mode cohesion at the fixed reference section $A$ ($X_A\simeq22\,\mathrm{mm}$) in Fig.~\ref{fig:example2-geometry-mesh}(a). Traction is averaged over the damaged band; separation is obtained by integrating crack strain along the reference normal. Figure~\ref{fig:example2-response}(c) shows traction components normalized using Eq.~\eqref{eq:mixed-mode-component-normalization}. Both normal tension and tangential shear are substantial during softening, confirming mixed-mode separation. The effective quantities~\eqref{eq:mixed-mode-effective-variables} allow direct comparison of the local damage-loading branch with the target envelope.\footnote{Following \citet{feng2026multilevel}, we divide both integrated separations by $1+\meshsize/(2\ell)=1.15625$. This quasi-one-dimensional finite-band-width correction follows from the mesh band width and $\ell$. For linear $\Omega=\phi$, the plateau term in Eq.~\eqref{eq:1d-positive-width-energy} increases the stationary separation by this factor.}

Figure~\ref{fig:example2-response}(d) shows the measured mixed-mode law at section $A$. The dimensionless effective traction--separation curve closely follows the target line $\normalizedeffectivetraction=1-\normalizedeffectiveopening$ in Eq.~\eqref{eq:mixed-mode-parametric-cohesive-law}. Traction components change as the current normal rotates, but the effective response retains the prescribed envelope. The target cohesive law is thus recovered under finite deformation and tensile--shear coupling.

\begin{figure}[!htbp]
 \centering
 \includegraphics[width=0.730\linewidth]{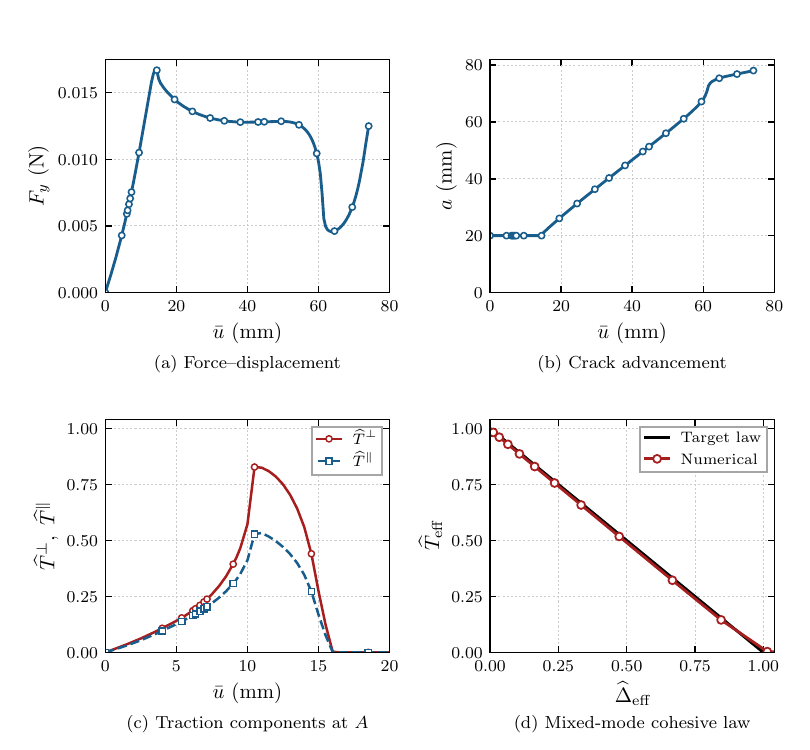}
\caption{Peeling response: (a) vertical reaction--displacement; (b) crack length; (c) normal and tangential tractions at section $A$; (d) mixed-mode cohesive law at that section. Black line: target envelope.}
 \label{fig:example2-response}
\end{figure}

\FloatBarrier
\subsection{Three-dimensional single-edge-notched specimen: competing tensile and shear strengths and out-of-plane deformation}\label{sec:example3}
We examine a three-dimensional single-edge-notched thin plate under horizontal shear. The aim is to determine how the relative tensile and shear strengths affect the crack path and its coupling with out-of-plane deformation. Both calculations use the same geometry, boundary conditions, compressible three-dimensional NH law, fracture energy, and normal cohesive strength. Strength ratios $\beta=\Tperpc/\Tparallelc$ of 10 and 0.1 make the specimen relatively weaker in shear and tension, respectively.

Figure~\ref{fig:example3-geometry-mesh}(a) shows the reference geometry and boundary conditions. The reference axes $X,Y,Z$ align with the loading, height, and thickness directions, respectively. The in-plane dimensions are $L=H=1\,\mathrm{mm}$ and the thickness is $B=0.1\,\mathrm{mm}$. A through-thickness edge notch centered on the left edge has length $a_0=0.5\,\mathrm{mm}$ and width $0.002\,\mathrm{mm}$. The bottom face $Y=0$ is fixed, $u_x=u_y=u_z=0$; the top face $Y=H$ is subjected to horizontal displacement $u_x=\bar u$ with $u_y=u_z=0$ constrained. Other surfaces are traction-free. The phase-field satisfies the natural condition~\eqref{eq:phi-boundary-KKT} on all external surfaces. No propagation path is prescribed ahead of the notch.

\begin{figure}[!htbp]
 \centering
 \includegraphics[width=\linewidth]{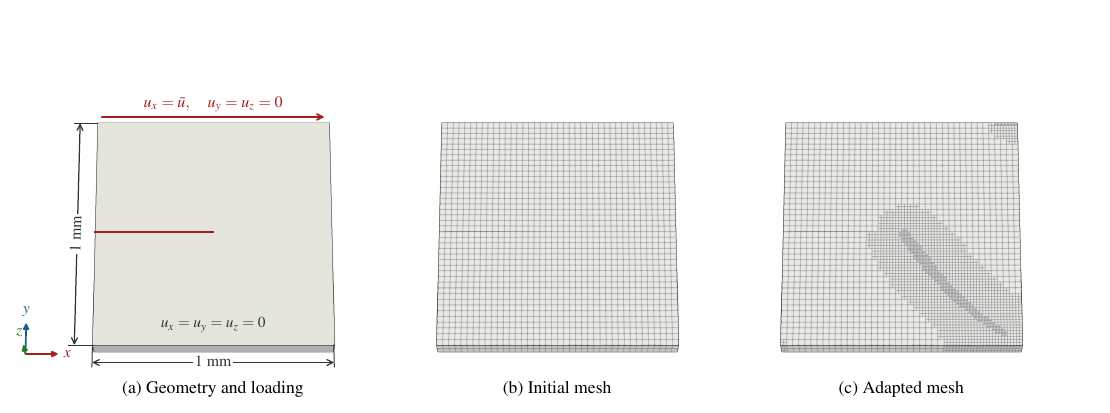}
\caption{Reference configuration of the three-dimensional single-edge-notched specimen: (a) geometry and shear loading; (b) initial mesh; (c) adapted mesh at $\beta=0.1$ and $\bar u=0.284\,\mathrm{mm}$. Thickness: $B=0.1\,\mathrm{mm}$.}
 \label{fig:example3-geometry-mesh}
\end{figure}

The bulk uses the NH potential~\eqref{eq:bulk-NH-energy} with $\mu=2\,\mathrm{N/mm^2}$ and $\nu=0.4$, giving first Lam\'e constant $\lambda=8\,\mathrm{N/mm^2}$. Linear cohesive softening, $\Omega(\phi)=\phi$, uses $\Gc=0.025\,\mathrm{N/mm}$, $\Tperpc=1\,\mathrm{N/mm^2}$, and $\ell=0.04\,\mathrm{mm}$. Ratios $\beta=10$ and 0.1 give tangential strengths $\Tparallelc$ of 0.1 and $10\,\mathrm{N/mm^2}$. Tangential characteristic separation and stiffness vary consistently as $\Deltaparallelch=\beta\Deltaperpch$ and $\kparallel=\kperp/\beta^2$; the corresponding $\Deltaparallelch$ values are 0.5 and $0.005\,\mathrm{mm}$.

Displacement and phase-field use trilinear hexahedra with $2\times2\times2$ full integration. Internal variables are stored at Gauss points. The initial mesh has 3200 elements and two layers through the thickness. Two-level adaptive refinement \citep{feng2026multilevel} targets element diameters $\sqrt{3}\ell/2$ and $\sqrt{3}\ell/3$ in the phase-field transition and compliance-damage concentration regions, respectively. Final meshes contain 9269 and 14470 elements for $\beta=10$ and 0.1; Fig.~\ref{fig:example3-geometry-mesh}(c) shows the latter. NH elasticity satisfies Proposition~\ref{prop:crack-direction-plane-reduction}. We therefore search for Gauss-point directions in the plane of the largest and smallest right principal stretches and update them during loading. All three displacement components remain active. To track instability in the $\beta=0.1$ case, we constrain the phase-field increment norm. When this constraint is active, the boundary displacement is held fixed. Elements whose vertices all satisfy $\phi>0.99$ are deactivated after a converged step, retaining their energy at deactivation. Both field sequences show actual deformation, with blue to red indicating $\phi$ from 0 to 1.

When shear strength is relatively low ($\beta=10$), the crack extends mainly rightward along the original notch direction (Fig.~\ref{fig:example3-fields-shear}). Relative sliding between the upper and lower parts increases, while the overall deformation remains nearly in-plane. Figure~\ref{fig:example3-fields-shear}(a) corresponds to the reaction peak [Fig.~\ref{fig:example3-response}(a)]; Fig.~\ref{fig:example3-fields-shear}(b,c) shows postpeak propagation.

\begin{figure}[!htbp]
 \centering
 \includegraphics[width=\linewidth]{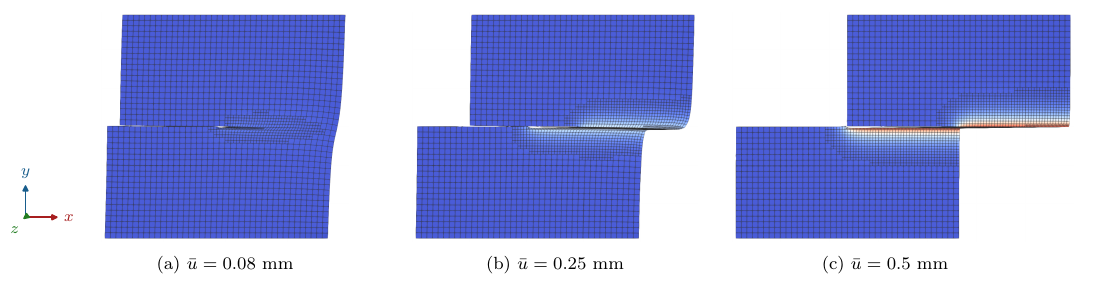}
\caption{Case with relatively low shear strength ($\beta=10$): current configuration, phase-field $\phi$, and adaptive mesh. Actual deformation is shown; regions with $\phi>0.95$ are hidden only for visualization.}
 \label{fig:example3-fields-shear}
\end{figure}

When tensile strength is relatively low ($\beta=0.1$), the crack turns onto an inclined opening path accompanied by pronounced out-of-plane bending (Fig.~\ref{fig:example3-fields-opening}). Figure~\ref{fig:example3-fields-opening}(a) corresponds to the reaction peak at $\bar u\approx0.196\,\mathrm{mm}$, when deformation remains nearly in-plane. Bending then grows rapidly after the peak, reaching its largest amplitude at $\bar u\approx0.284\,\mathrm{mm}$ [Fig.~\ref{fig:example3-fields-opening}(b)]. As the crack advances toward the lower edge, bending is progressively released. Figure~\ref{fig:example3-fields-opening}(c) shows the final state at $\bar u\approx0.288\,\mathrm{mm}$.

\begin{figure}[!htbp]
 \centering
 \includegraphics[width=\linewidth]{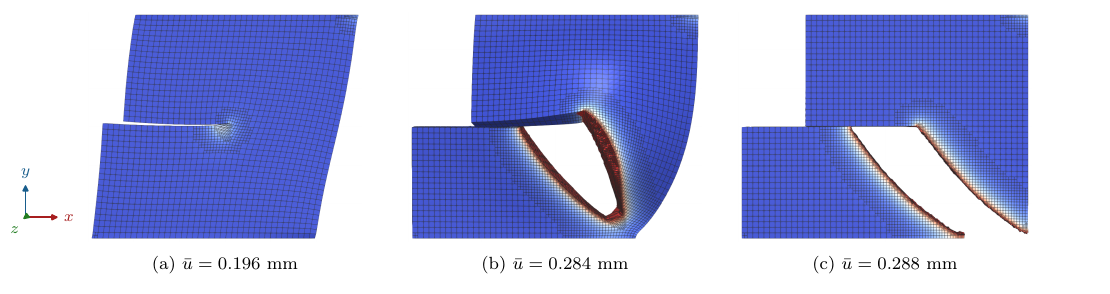}
\caption{Case with relatively low tensile strength ($\beta=0.1$): current configuration, phase-field $\phi$, and adaptive mesh. Actual deformation is shown; regions with $\phi>0.95$ are hidden.}
 \label{fig:example3-fields-opening}
\end{figure}

The two failure modes have different postpeak responses. In Fig.~\ref{fig:example3-response}(a), the $\beta=10$ reaction decreases gradually with in-plane sliding and crack growth, whereas $\beta=0.1$ produces a steeper drop as inclined opening and out-of-plane bending develop.

Figure~\ref{fig:example3-response}(b) quantifies out-of-plane deformation by $\max_{Z=B/2}|u_z|/B$, maximized over mesh nodes on the reference midplane. The $\beta=10$ value remains below 0.01, consistent with nearly in-plane sliding. For $\beta=0.1$, it reaches approximately 2.06 at $\bar u\approx0.284\,\mathrm{mm}$; Fig.~\ref{fig:example3-response}(c) views this state along $x$. Further growth reduces the final amplitude to about 0.21 and reaction to approximately $0.58\%$ of its peak. Several late equilibrium states share one boundary displacement, producing nearly vertical drops in the response curves. Thus, under the same horizontal shear, changing the relative tensile and shear strengths shifts failure from nearly in-plane sliding to inclined opening coupled with out-of-plane deformation. Out-of-plane buckling of cracked thin plates under shear was also studied by \citet{brighenti2011buckling}.

\begin{figure}[!htbp]
 \centering
 \includegraphics[width=\linewidth]{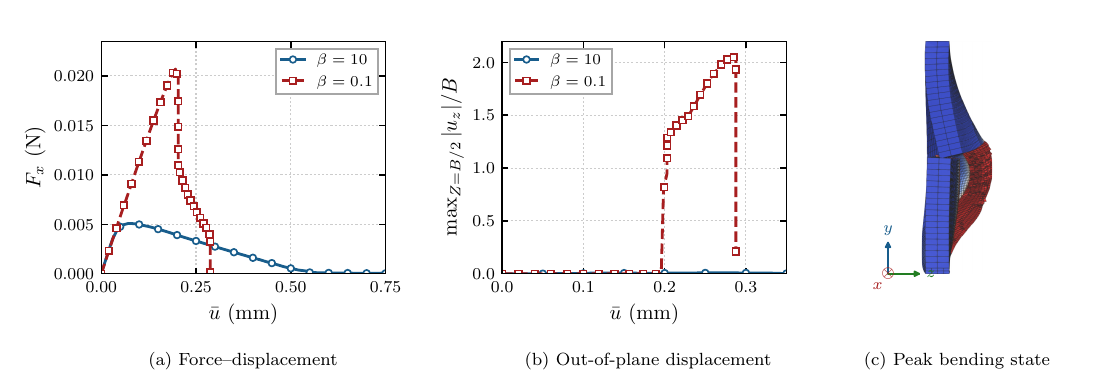}
\caption{Three-dimensional shear response: (a) horizontal reaction--displacement; (b) normalized out-of-plane displacement; (c) $yz$-plane front view at maximum bending ($\beta=0.1$, $\bar u\approx0.284\,\mathrm{mm}$).}
 \label{fig:example3-response}
\end{figure}

\FloatBarrier
\subsection{Crack growth and crack-tip fields in a hyperelastic sheet}\label{sec:example4}
A hyperelastic sheet with a single rounded edge notch is loaded in symmetric tension. This example compares the field ahead of a propagating crack with the asymptotic solution of \citet{long2011finite} and examines diffuse damage during notch-induced nucleation. The latter concerns bulk-energy relaxation and its compatibility with interfacial strength in free cohesive-crack models \citep{Braides1995,rodella2026sharp,bourdin2026strength}.

The reference sheet measures $40\times40\,\mathrm{mm}$ with unit reference thickness. A rounded notch of length $a_0=2\,\mathrm{mm}$ and tip radius $0.0125\,\mathrm{mm}$ lies on the left midline $Y=20\,\mathrm{mm}$. The upper and lower boundaries receive $u_y=\bar u$ and $u_y=-\bar u$, respectively, with lateral motion free. At the lower-left corner, $u_x=0$ removes rigid translation; all remaining boundaries are traction-free. Figure~\ref{fig:example4-geometry} shows the geometry and fixed mesh of 22949 bilinear quadrilaterals. The vertical element size near the notch and along the refined midline band is $\meshsize=0.005\,\mathrm{mm}$; five straight segments approximate the rounded tip. No crack-growth interface is prescribed, and no elements are deleted.

The bulk uses incompressible plane-stress NH elasticity with $\mu=1\,\mathrm{MPa}$. Linear cohesive softening uses $\Gc=5\,\mathrm{N/mm}$, $\Tperpc=\sqrt{50}\,\mathrm{MPa}$, and $\beta=0.1$, with phase-field length $\ell=0.1\,\mathrm{mm}$. A relatively high cohesive strength reduces the characteristic cohesive length, separating the cohesive zone more clearly from the outer $1/r$ asymptotic region:
\begin{equation}\label{eq:example4-characteristic-length}
 \ell_{\mathrm{ch}}=\frac{2\mu\Gc}{(\Tperpc)^2}
 =0.2\,\mathrm{mm},\qquad \frac{\ell}{\ell_{\mathrm{ch}}}=0.5.
\end{equation}
The discrete energy coefficient uses the mesh correction $\Gc^{\mathrm{disc}}=\Gc/[1+\meshsize/(2\ell)]$, where superscript $\mathrm{disc}$ denotes the discretization correction \citep{feng2026multilevel}. The computational cap on $\omega$ is $10^4$.

\begin{figure}[!htbp]
 \centering
 \includegraphics[width=\linewidth]{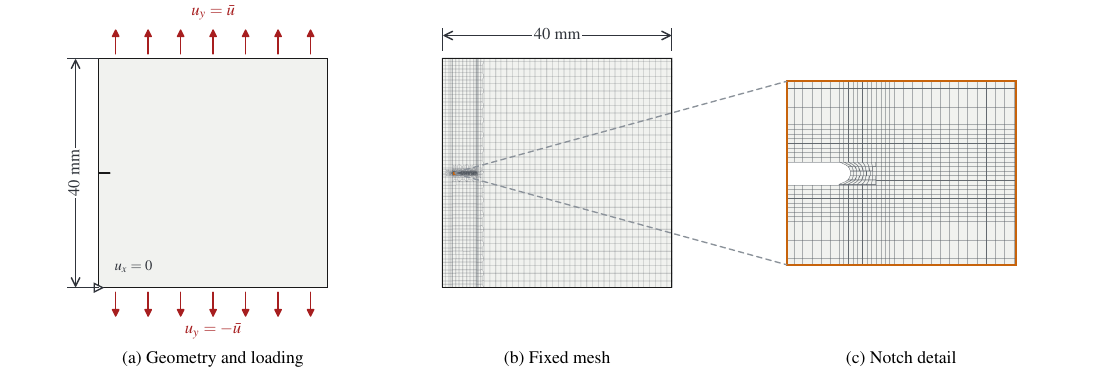}
\caption{Single-edge-notched sheet: (a) geometry and loading; (b) fixed mesh; (c) mesh near the notch.}
 \label{fig:example4-geometry}
\end{figure}

\subsubsection{Nucleation and crack growth}
The boundary displacement is increased initially. Once damage develops, we prescribe a positive fracture-energy increment to trace the equilibrium path, including displacement snap-back. This follows dissipation-control and phase-field fracture-energy-based arc-length approaches \citep{verhoosel2009dissipation,bharali2022robust}. Here $\bar u$ is the displacement of one boundary; the total relative displacement is $2\bar u$. The crack-mouth opening displacement (CMOD) is monitored but does not control loading. It is the vertical displacement difference between the initial notch faces at $X=1.5\,\mathrm{mm}$. The arc-length control quantity is the fracture energy
\begin{equation}\label{eq:example4-arc-energy}
 \energyf=\int_{\mathcal B}\varphif(\phi,\nabla\phi,\omega)\,\dd V,
\end{equation}
where $\varphif$ is the density~\eqref{eq:fracture-energy-density}, with $\Omega(\phi)=\phi$ and discrete coefficient $\Gc^{\mathrm{disc}}$. Figure~\ref{fig:example4-response} plots reaction against boundary displacement, CMOD, and $\energyf$. Circles and squares distinguish states with increasing and decreasing boundary displacement. The reaction peaks at $124.56\,\mathrm N$ for $\bar u\approx17.12\,\mathrm{mm}$, followed by pronounced displacement snap-back during growth; CMOD also decreases locally. Neither boundary displacement nor CMOD is therefore monotone along the entire path, whereas fracture-energy-increment control traces the snap-back branch.

\begin{figure}[!htbp]
 \centering
 \includegraphics[width=\linewidth]{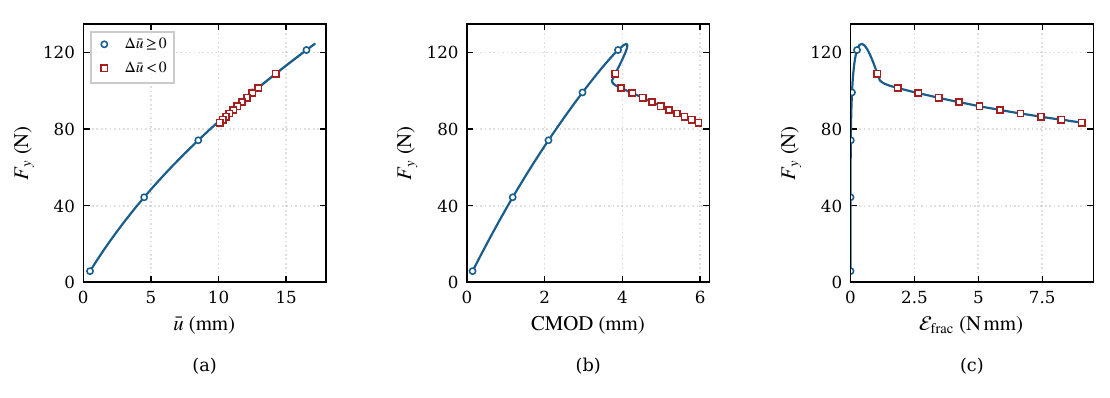}
\caption{Overall snap-back response: reaction versus (a) boundary displacement, (b) CMOD, and (c) fracture energy.}
 \label{fig:example4-response}
\end{figure}

Figure~\ref{fig:example4-early} compares reference-configuration fields near the notch at $\bar u=9.50\,\mathrm{mm}$ and $14.50\,\mathrm{mm}$. During the early stages of nucleation, $\phi$, $\omega$, and $\bmw$ are diffuse ahead of the rounded tip, extending beyond a single midline element layer. With increasing load, $\omega$ and $\bmw$ grow and their supports extend ahead of the notch. This diffuse response retains the deformation modes needed for bulk-energy relaxation outside the concentrated crack. Lower semicontinuity and relaxation of free cohesive-crack functionals require compatibility of bulk energy and interfacial strength \citep{Braides1995,rodella2026sharp}; strength-degradation phase-field models also retain the associated nonlinear diffuse deformation \citep{bourdin2026strength}. Logarithmic scales for $\omega$ and $\bmw$ in Fig.~\ref{fig:example4-early} reveal low-amplitude regions, with a common range for each field at both states.

\begin{figure}[!htbp]
 \centering
 \includegraphics[width=\linewidth]{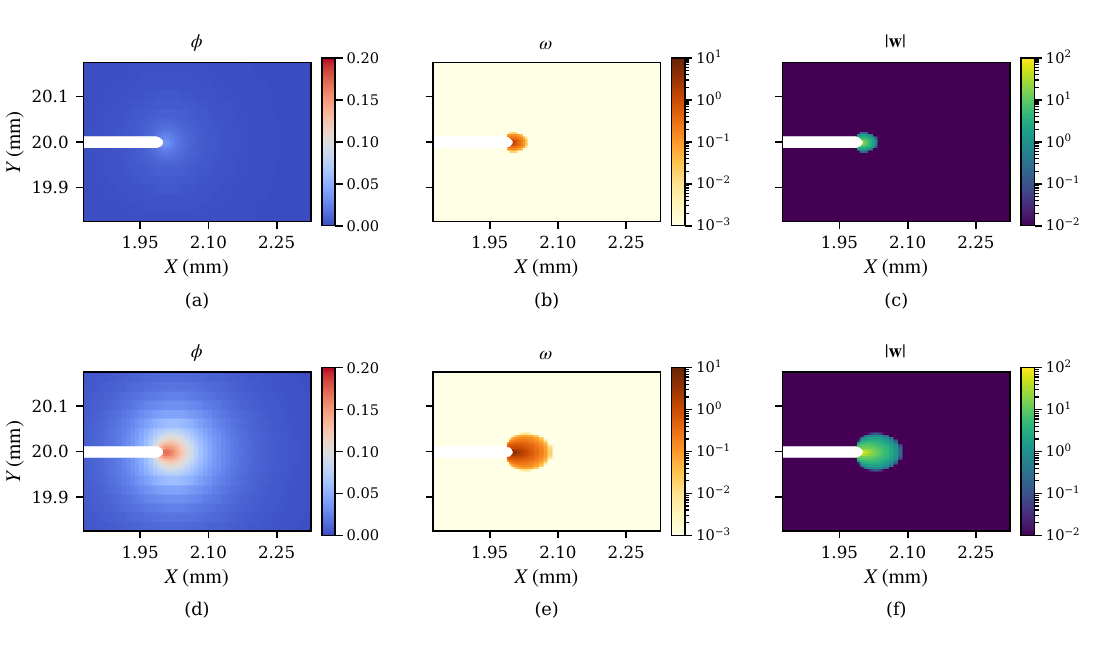}
\caption{Reference-configuration fields near the notch. Rows: $\bar u=9.50$ and $14.50\,\mathrm{mm}$; columns: $\phi$, $\omega$, and $|\bmw|$.}
 \label{fig:example4-early}
\end{figure}

Figure~\ref{fig:example4-developed} shows local fields after substantial propagation, at $\bar u=10.09\,\mathrm{mm}$ and CMOD $5.96\,\mathrm{mm}$. Phase-field $\phi$ and compliance damage $\omega$ are plotted in the reference configuration; crack strain $\bmw$ is shown in the deformed configuration. Both $\omega$ and $\bmw$ concentrate within a narrow band, approximating a singular distribution. The phase-field $\phi$ retains a continuous transition. Panel (c) shows sparse direction arrows only for sufficiently large $\bmw$ in the reference midline layer, indicating the separation direction. All elements are retained; opening is displayed through actual deformation and field values.

\begin{figure}[!htbp]
 \centering
 \includegraphics[width=\linewidth]{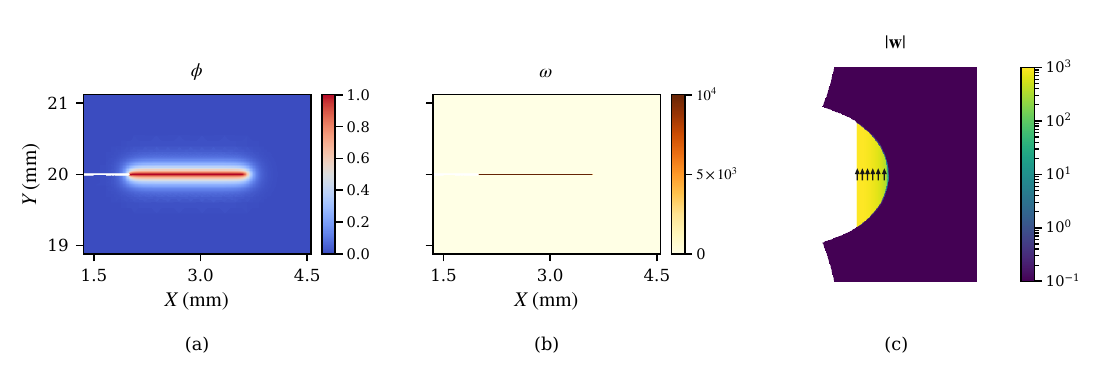}
\caption{Local fields of the propagated crack: (a) phase-field $\phi$; (b) compliance damage $\omega$; (c) crack strain $\bmw$ and its direction in the deformed configuration.}
 \label{fig:example4-developed}
\end{figure}

\subsubsection{Crack-tip stress distribution}
{The propagated state allows a local test against nonlinear fracture mechanics beyond the preceding field comparisons. We therefore extract Cauchy (true) stress ahead of the crack from elements adjacent to the reference midline and compare it with the asymptotic solution of \citet{long2011finite}.} Stress recovery follows the code's three-dimensional incompressible plane-stress convention, with three-dimensional volume ratio $J^{(3)}=1$ and hence $\sigma_{22}=(\bmP\bmF^\trans)_{22}/J^{(3)}=(\bmP\bmF^\trans)_{22}$. Here $J^{(3)}$ differs from the in-plane Jacobian $J$. We compute this quantity from the full variational stress at each Gauss point. We then average it with reference-volume weights and combine adjacent elements above and below the midline. In the sampling columns used for fitting, $\omega=0$ and $\bmw=\mathbm0$, so total and bulk deformation coincide. For the NH material used here, the asymptotic field ahead of the tip derived by \citet{long2011finite} is
\begin{equation}\label{eq:example4-long-stress}
 \sigma_{22}\sim\frac{\mathcal{J}}{\pi r},\qquad r=X-a>0,
\end{equation}
where $\mathcal{J}$ is the mechanical energy flux defined in the next subsection, $a$ is the effective tip position, and $r$ is distance in the reference configuration.

Figure~\ref{fig:example4-tip}(a) compares $\phi$ and normalized compliance $\omega/(1+\omega)$. In panel (b), a red open circle marks the distinct stress kink $A$ at $X\approx3.754\,\mathrm{mm}$. This is where $\omega$ vanishes, identifying the cohesive zone front in this state. Cohesive effects act to the left of $A$; data to its right are compared with Long's asymptotic field. The fits in panels (b,c) use the data over $X\in[3.753,4.007]\,\mathrm{mm}$, indicated by gray shading and red open symbols in (b). The amplitude is supplied independently by the domain-integral value $\mathcal{J}$ from the next subsection; only the effective tip $a$ and nonsingular background stress $\sigma_0$ are fitted. Vertical dashed lines in (a,b) mark the resulting $a\approx3.703\,\mathrm{mm}$. The log--log plot of $\sigma_{22}-\sigma_0$ against $r$ in (c) shows close agreement with Long's $1/r$ asymptotic field over the selected interval outside the process zone.

\begin{figure}[!htbp]
 \centering
 \includegraphics[width=\linewidth]{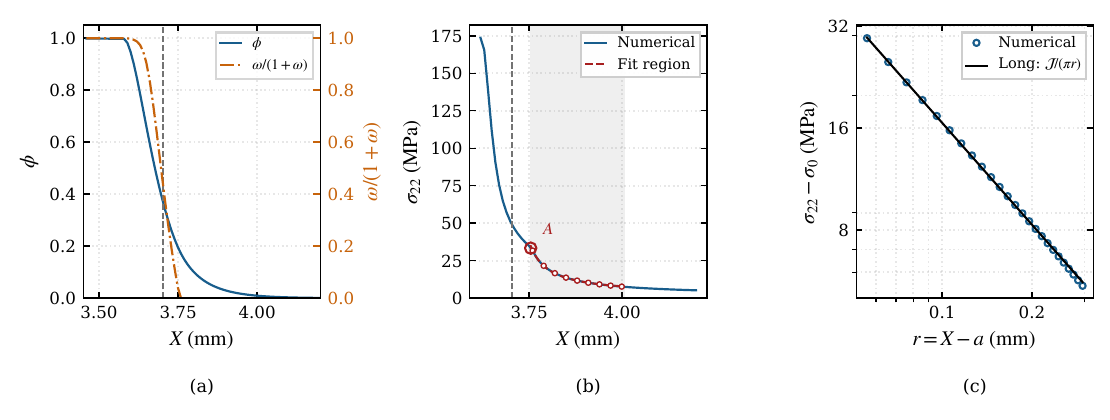}
\caption{Crack-tip comparison: (a) phase-field $\phi$ and normalized compliance $\omega/(1+\omega)$; (b) Cauchy stress $\sigma_{22}$; (c) log--log Cauchy-stress asymptotics. $A$ marks the stress kink at the cohesive zone front; gray shading indicates the fitting interval.}
 \label{fig:example4-tip}
\end{figure}

\subsubsection{Mechanical and fracture-energy fluxes}
{The preceding stress comparison requires an independent measure of its amplitude, while the diffuse damage raises a further question about the energy required for propagation. We address both by evaluating the mechanical and fracture-energy fluxes in the same state. The mechanical $\mathcal{J}$-integral introduced by \citet{rice1968path} measures the energetic driving force for crack growth.} Following \citet{moran1987crack}, we use a weight function $q$ to convert the contour integral to an equivalent domain integral. Per unit reference thickness,
\begin{equation}\label{eq:example4-J-definitions}
 \mathcal{J}=\int_{\mathcal B}
 \bigl(P_{iJ}u_{i,1}-\psi\delta_{1J}\bigr)q_{,J}\,\dd A,
\end{equation}
where $\psi$ is the recoverable energy density~\eqref{eq:split} and $\bmP$ is the full variational stress. The weight $q$ equals one inside the integration box and zero outside, varying continuously across edge strips of width $0.02\,\mathrm{mm}$. Figure~\ref{fig:example4-J}(a) shows four contours enclosing the propagated tip, used to assess the reliability of the integral. Their left edges lie in the nearly fully damaged wake, and their upper and lower integration strips are fixed; only the right edge is varied.

{To compare this mechanical driving force with the energy required for crack formation, we evaluate the fracture-energy flux from the phase-field energy--momentum balance of \citet{karma2005crack,hakim2009laws}.} Take the flux resisting propagation in the positive $X_1$ direction as positive and denote it by $\mathcal{J}_{\mathrm{frac}}$. The fracture-energy density~\eqref{eq:fracture-energy-density} gives the decomposition
\begin{equation}\label{eq:example4-J-fracture}
 \mathcal{J}_{\mathrm{frac}}=\mathcal{J}_{\mathrm{frac}}^{(1)}+\mathcal{J}_{\mathrm{frac}}^{(2)},
\end{equation}
where $\mathcal{J}_{\mathrm{frac}}^{(1)}$ contains the local phase-field and gradient contributions and $\mathcal{J}_{\mathrm{frac}}^{(2)}$ is the compliance-damage/phase-field coupling contribution:
\begin{equation}\label{eq:example4-J-partition}
 \begin{aligned}
 &\mathcal{J}_{\mathrm{frac}}^{(1)}=\int_{\mathcal B}
 \left[\frac{\Gc^{\mathrm{disc}}}{2\ell}
 \bigl(\phi^2+\ell^2\lvert\nabla\phi\rvert^2\bigr)\delta_{1J}
 -\xi_J\phi_{,1}\right]q_{,J}\,\dd A,\\[0.4em]
 &\mathcal{J}_{\mathrm{frac}}^{(2)}=\frac{\Gc^{\mathrm{disc}}}{\ell}
 \int_{\mathcal B}(1-\phi)^2\omega\,q_{,1}\,\dd A,
 \end{aligned}
\end{equation}
where $\boldsymbol{\xi}=\partial_{\nabla\phi}\varphif=\Gc^{\mathrm{disc}}\ell\nabla\phi$ is energetically conjugate to the phase-field gradient. Suppose that quasistatic propagation satisfies the field equations, with no additional dissipation and all fields translating steadily with the tip. The mechanical and fracture fluxes then satisfy $\mathcal{J}=\mathcal{J}_{\mathrm{frac}}$.

Figure~\ref{fig:example4-J}(b) reports results for the four contours. The mechanical flux is $\mathcal{J}/\Gc\approx1.04202$, with relative range approximately $6.41\times10^{-8}$, indicating insensitivity to the chosen right-edge location. The fracture flux is $\mathcal{J}_{\mathrm{frac}}/\Gc\approx1.04096$, and their difference is $(\mathcal{J}-\mathcal{J}_{\mathrm{frac}})/\Gc\approx1.06\times10^{-3}$, showing near balance. Both exceed the input toughness $\Gc$ by approximately $4\%$. The decomposition gives $\mathcal{J}_{\mathrm{frac}}^{(1)}/\Gc\approx1.00181$ and $\mathcal{J}_{\mathrm{frac}}^{(2)}/\Gc\approx0.03916$. The local phase-field and gradient contributions are already very close to $\Gc$. The excess arises mainly from the $(1-\phi)^2\omega$ coupling term.

In the complete-separation limit of the one-dimensional concentrated solution, $\phi^*\to1$ and the coupling term vanishes. Here a diffuse damaged region first forms ahead of the notch before the crack localizes and propagates, and the coupling flux remains nonzero. This is consistent with irreversibility retaining the damage history of nucleation, and reflects coexistence of diffuse damage and concentrated fracture in the propagated state.

For a nearly crack-like V-notch, \citet{bourdin2026strength} found a nucleation stress-intensity factor approximately $6\%$ above the sharp-interface result after removing mesh-induced toughening. Through the antiplane relation $G\propto K^2$, this corresponds to an energy increase of approximately $12\%$, also associated with a nonlinear deformation region formed before nucleation. Despite differences in model and loading, this result is consistent with the present flux decomposition. Diffuse response during nucleation can raise the fracture energy above the interfacial cost of an ideal concentrated crack. Multidimensional $\Gamma$-convergence analysis should help characterize this history effect and its relation to concentration and relaxation limits.

\begin{figure}[!htbp]
 \centering
 \includegraphics[width=\linewidth]{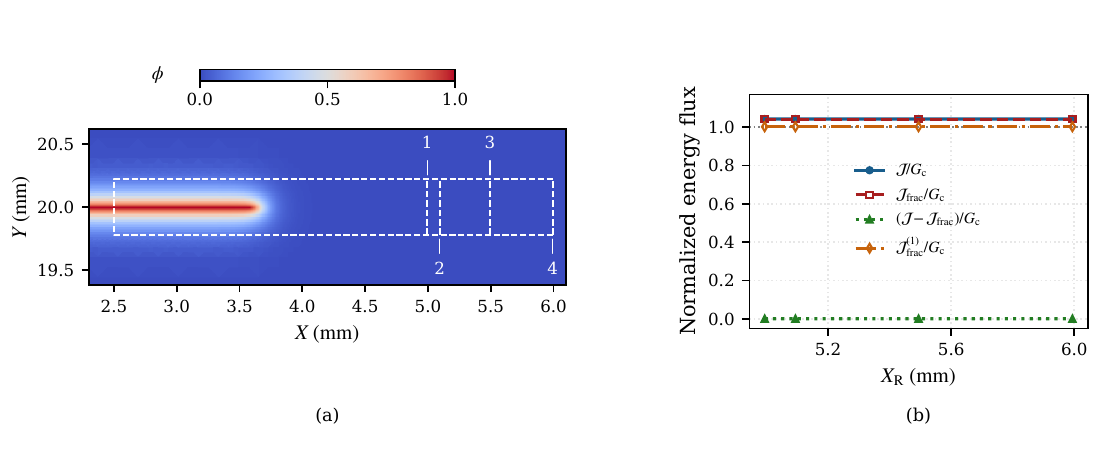}
\caption{Energy fluxes in the propagated state: (a) phase-field $\phi$ and integration contours; (b) normalized mechanical and fracture fluxes and their difference. $\mathcal{J}_{\mathrm{frac}}^{(1)}$ excludes the coupling term.}
 \label{fig:example4-J}
\end{figure}

\FloatBarrier
\subsection{Randomly perforated elastomer sheet: abrupt fracture and complex crack paths}\label{sec:example5}

We reproduce the tensile experiment on a randomly perforated hyperelastic sheet reported by \citet{lee2024size}. Sequential failure of ligaments between holes causes sudden load drops, while interactions among holes lead to pronounced crack deflection and coalescence. This example tests the model's ability to predict abrupt brittle failure and complex crack paths. A CZM can also describe brittle failure when its characteristic cohesive length is much smaller than the structural scale.

The specimen measures $20\times20\,\mathrm{mm}$ in-plane and has thickness $B=0.5\,\mathrm{mm}$. A two-dimensional plane-stress model reconstructs the 12 circular and elliptical holes from the experiment [Fig.~\ref{fig:example5-geometry-mesh}(a)]. The bottom edge is fixed, $u_x=u_y=0$; the top receives vertical displacement $u_y=\bar u$ with $u_x=0$ constrained. Side edges and hole boundaries are traction-free, and natural phase-field conditions are used. No inter-hole cracks or weakened paths are prescribed.

\begin{figure}[!htbp]
 \centering
 \includegraphics[width=\linewidth]{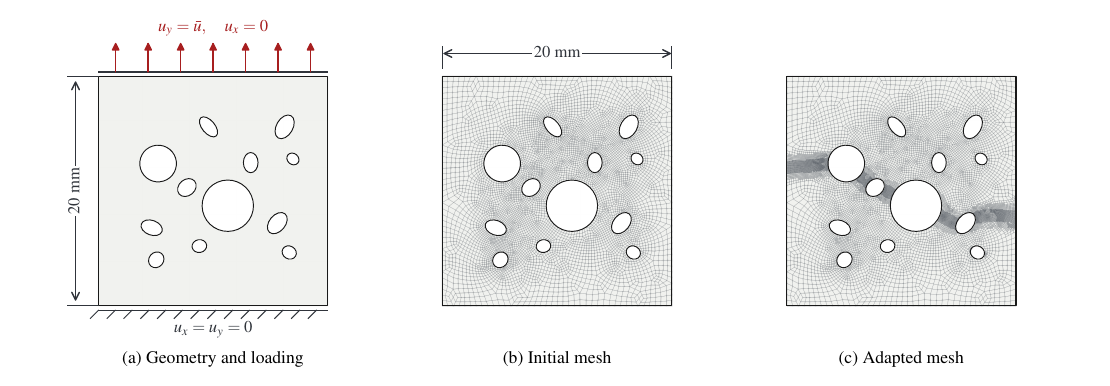}
\caption{Reference configuration of the randomly perforated sheet: (a) geometry and boundary conditions, reconstructed from \citet{lee2024size}; (b) initial mesh; (c) adapted mesh at $\bar u=13\,\mathrm{mm}$.}
 \label{fig:example5-geometry-mesh}
\end{figure}

The bulk uses incompressible plane-stress generalized NH elasticity [Eq.~\eqref{eq:bulk-GNH-energy}] with $\mu=0.22448\,\mathrm{N/mm^2}$, $b=0.04733$, and $p=0.75$. Linear cohesive softening uses $\Gc=0.5\,\mathrm{N/mm}$, $\Tperpc=0.56\,\mathrm{N/mm^2}$, $\ell=0.2\,\mathrm{mm}$, and $\beta=0.1$. Displacement and phase-field use bilinear quadrilaterals with $2\times2$ full integration; internal variables are updated at Gauss points. The initial mesh contains 6764 elements. Adaptive refinement targets diameters $\ell$ and $\ell/2$ in the transition and damage-concentration regions, respectively, giving 13931 final elements [Fig.~\ref{fig:example5-geometry-mesh}(c)]. We use an enhanced dissipation-based arc-length algorithm \citep{rorentrop2024time,rorentrop2026computationally} with the implementation in Section~\ref{sec:numerical-implementation}. This gives stable, converged solutions through abrupt failure.

Figure~\ref{fig:example5-fields} compares the main experimental and simulated failure stages. Macroscopic stretch is $\lambda=1+\bar u/L_0$ with $L_0=20\,\mathrm{mm}$. Dissipation-based arc-length control permits successive crack-growth states at a fixed displacement, so several saved states may share one displacement. Initially the holes elongate while intervening ligaments remain intact [columns (i--iii)]. The central hole then coalesces with neighboring small holes [column (iv)] and connects to the larger upper-left hole, forming an enlarged cavity [column (v)]. Once the cavity reaches the left edge, the right ligament contracts and ruptures [column (vi)], ultimately separating the sheet into upper and lower parts [column (vii)]. Simulation and experiment show similar coalescence, breakthrough at one edge, and ligament failure, but at somewhat different applied stretches. In the simulation, the right ligament contracts and nearly separates during the fixed-displacement stage $\lambda=1.611$. In the experiment, a connecting ligament remains on the right at $\lambda=1.637$. The model captures the observed sequence of coalescence, edge breakthrough, and final ligament failure.
\begin{figure}[!htbp]
 \centering
 \includegraphics[width=\linewidth]{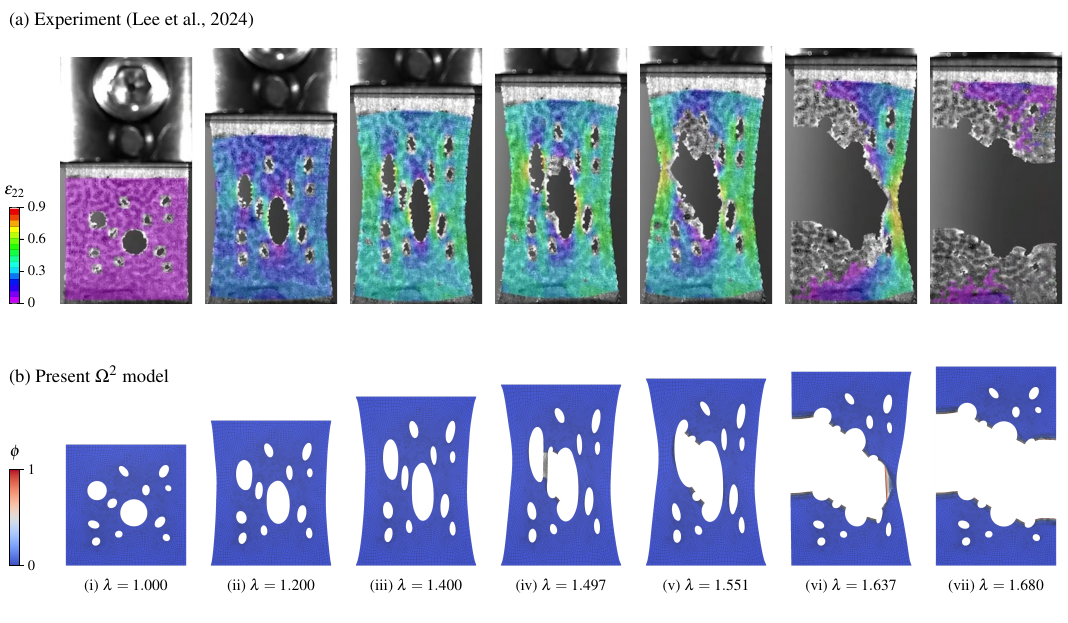}
\caption{Failure stages of the perforated sheet: (a) experimental images from Fig.~9(a) of \citet{lee2024size}; (b) current configuration and damage field of the present model. Color scales show experimental strain $\varepsilon_{22}$ and computed phase-field $\phi$. Shared column labels use experimental stretches; actual simulation sampling is discussed in the text. Regions with $\phi>0.95$ are hidden only for visualization.}
 \label{fig:example5-fields}
\end{figure}

Sequential ligament failure produces multiple sudden reaction drops (Fig.~\ref{fig:example5-response}). Experimental data are taken from Fig.~10 of \citet{lee2024size}. Stretch $\lambda$ and nominal stress $P$ are converted to displacement and force using $\bar u=20(\lambda-1)\,\mathrm{mm}$ and $F_y=P\times(20\times0.5\,\mathrm{mm^2})$. The ascending branches are close. The computed peak is $1.823\,\mathrm{N}$ at $\bar u=9.779\,\mathrm{mm}$; the digitized experimental peak is approximately $1.864\,\mathrm{N}$, a difference of about $2.2\%$.

\begin{figure}[!htbp]
 \centering
 \setlength{\abovecaptionskip}{2pt}
 \includegraphics[trim=0 12bp 0 0,clip,width=0.608333\linewidth]{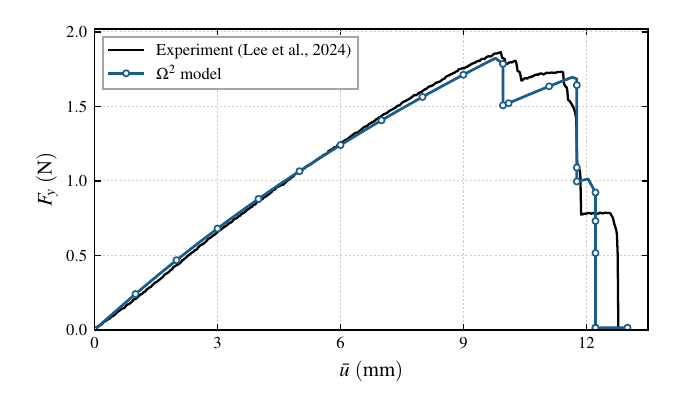}
\caption{Vertical reaction--displacement response of the perforated sheet, compared with the experiment in Fig.~10 of \citet{lee2024size}.}
 \label{fig:example5-response}
\end{figure}

Postpeak coalescence first reduces the reaction. The remaining ligaments carry renewed load as stretching continues, then lose capacity through further rupture. Two main computed drops near $\bar u\approx11.76\,\mathrm{mm}$ and $12.22\,\mathrm{mm}$ correspond to edge breakthrough and final separation between columns (v) and (vi) of Fig.~\ref{fig:example5-fields}. The final drop occurs slightly earlier than in the experiment.
Together, the overall response and crack morphology show that the model captures inter-hole rupture, deflection, coalescence, and load drops caused by abrupt local failure in a complex perforated geometry without prescribed crack paths.

\FloatBarrier
\section{Conclusions and perspectives}\label{sec:conclusions}
This paper has formulated a kinematically and variationally consistent phase-field theory for cohesive fracture at finite deformation, in which continuous bulk deformation and concentrated crack separation can coexist. Its rank-one split retains independent crack-strain and reference-normal fields without prescribing displacement jumps, together with independent compliance damage $\omega$. General objective hyperelastic bulk potentials and admissible mixed-mode cohesive laws enter as independent constitutive inputs. Minimum potential energy links local separation to macroscopic equilibrium, with crack-plane rotation generating the geometric stress required for angular momentum balance.

Within a one-dimensional single-band family, energy comparisons at fixed boundary displacement favor Dirac-type concentration of crack strain and compliance damage. A displacement jump emerges at finite $\ell$ while the phase-field remains continuous, and the energy takes exactly the bulk-volume and cohesive-interface form of a sharp-crack CZM. The prescribed nominal traction--separation law is recovered independently of the bulk potential and $\ell$. Concentration gives separation a sharp interfacial description, where bulk unloading and cohesive separation interact through kinematics, traction balance, and energy transfer.

The near-tip Cauchy stress closely follows the {analytical $1/r$ crack-tip asymptotic solution} of \citet{long2011finite} outside the cohesive zone at finite $\ell$, providing a stringent test of the model. The energetic connection is thus accompanied by quantitative agreement with nonlinear crack-tip mechanics. Three-dimensional calculations and comparison with the perforated-sheet experiment further support the model's description of large deformation, out-of-plane buckling, and complex crack paths.

Diffuse damage generated during nucleation can persist alongside a growing concentrated crack. In the notched hyperelastic sheet, the nearly balanced mechanical and fracture-energy fluxes exceed $\Gc$ by about $4\%$, whereas excluding the $(1-\phi)^2\omega$ coupling contribution brings the fracture flux close to $\Gc$. This contribution records the irreversible history of diffuse damage during notch-induced nucleation and remains necessary for the model's description of bulk-energy relaxation. These findings call for mathematical analysis of multidimensional concentration and relaxation limits, including $\Gamma$-convergence, to clarify how diffuse damage and irreversible loading shape the limiting fracture energy.

Looking ahead, dynamic fracture at finite deformation would bring inertia and nonlinear wave transport into the energy balance, allowing local separation to interact with crack speed and dynamic instability. An extension to finite-strain plasticity would address the complementary question of how bulk plastic dissipation and stress redistribution compete with interfacial separation to govern ductile crack initiation and growth.

\section*{CRediT authorship contribution statement}
Ye Feng: Conceptualization, Methodology, Formal analysis, Software, Validation, Investigation, Data curation, Visualization, Writing -- original draft, Writing -- review \& editing, Funding acquisition.

\section*{Declaration of competing interest}
The author declares no competing financial interests or personal relationships that could have influenced this work.

\section*{Data availability}
The data and source code supporting this study are available from the author upon reasonable request.

\section*{Acknowledgments}
The author gratefully acknowledges financial support from the National Natural Science Foundation of China (Grant Nos.~12672112 and 12302101). 

\appendix

\section{Distributional gradients and measure representations of damage}\label{app:distributional-notation}
{Let $\mathcal B$ be a bounded Lipschitz domain, $\Gamma$ a smooth oriented interface, and $\bmchi$ piecewise $C^1$ with integrable gradient and jump. For smooth compactly supported tests $\mathbm A$ and $\zeta$, define $\langle D\bmchi,\mathbm A\rangle=-\int_{\mathcal B}\bmchi\cdot\mathrm{Div}\,\mathbm A\,\dd V$ and $\langle\delta_\Gamma,\zeta\rangle=\int_\Gamma\zeta\,\dd A$, with $\mathrm{Div}$ acting rowwise on the two-point tensor $\mathbm A$. Then}
\begin{equation}\label{eq:distributional-gradient-measure}
 D\bmchi
 =\bmFreg\,\mathcal L^d\measurerestriction\mathcal B
  +(\jumpu\otimes\bmN)\,\mathcal H^{d-1}\measurerestriction\Gamma,
\end{equation}
{where $\mathcal L^d$ and $\mathcal H^{d-1}$ are Lebesgue volume and Hausdorff surface measures; $\mathcal H^0$ is counting measure. The normal $\bmN$ points from the negative to the positive crack face, and $(\mu\measurerestriction E)(A)\coloneqq\mu(A\cap E)$ denotes measure restriction.}

Since $\mathcal L^d(\Gamma)=\mathcal L^d(\jumpSet)=0$, volume integrals over $\mathcal B$, $\mathcal B\setminus\Gamma$, and $\mathcal B\setminus\jumpSet$ coincide. The fracture integral in Eq.~\eqref{eq:1d-concentrated-energy} remains over $\Gamma$, retaining irreversible damage when the displacement jump closes.

Reference position is continuous across the interface, so $\jump{\bmchi}=\jumpu$. The gradient in Eq.~\eqref{eq:strong-discontinuity-additive} is distributional, with volume density $\bmFreg$; in the model formulated with ordinary functions, $\bmF=\nabla\bmchi$ is the weak gradient of the deformation map.

\subsection{Multiplicative representations of finite-width kinematics}\label{app:classical-multiplicative-kinematics}
{For invertible $\bmFreg$, the finite-width split~\eqref{eq:classical-finite-width} admits the multiplicative representations of \citet{armero1996analysis} and \citet{oliver2003strong}. Extending the jump and normal into the band gives}
\begin{equation}\label{eq:AG-multiplicative-decompositions}
 \bmFh=\bmFreg\bmFRh=\bmFLh\bmFreg,
\end{equation}
where the right and left factors are
\begin{equation}\label{eq:classical-multiplicative-factors}
\bmFRh\coloneqq\bmIref+(\bmFreg^{-1}\jumpu)\otimes\bmN\,\delta_\Gamma^h,
 \qquad
 \bmFLh\coloneqq\bmI+\jumpu\otimes(\bmFreg^{-\trans}\bmN)\,\delta_\Gamma^h,
\end{equation}
where $\bmIref$ and $\bmI$ are the reference and spatial identity tensors; $\bmFRh$ and $\bmFLh$ act on material and spatial vectors, respectively. Let $a_h\coloneqq1+\delta_\Gamma^h\bmN\cdot\bmFreg^{-1}\jumpu$. For $a_h\neq0$, the rank-one determinant and inverse identities give
\begin{equation}\label{eq:classical-rank-one-transport}
 \det\bmFh=(\det\bmFreg)a_h,
 \qquad
 (\bmFh)^{-\trans}\bmN=\frac{\bmFreg^{-\trans}\bmN}{a_h}.
\end{equation}
For orientation-preserving $\bmFh$ and $\bmFreg$, $a_h>0$, so normalization gives Eq.~\eqref{eq:classical-normal-transport}. In a compatible deformation, tangential gradients of the jump are included in the regular part.

\subsection{Damage fields and concentrated states}\label{app:model-concentration}
{Take a continuous $\phi\in H^1(\mathcal B)$ with $0\leq\phi\leq1$ and nonnegative $\omega\in L^1(\mathcal B)$, satisfying the boundary conditions, irreversibility, and finite-energy requirements. Boundary and interface values are interpreted as traces.}

{With $\stackrel{*}{\rightharpoonup}$ denoting weak-star convergence of finite Radon measures, diffuse and concentrated compliance damage are represented by}
\begin{equation}\label{eq:omega-volume-surface-measure}
 \omega^h\,\mathcal L^d\measurerestriction\mathcal B\stackrel{*}{\rightharpoonup}\omegameasure
 =\omegareg\,\mathcal L^d\measurerestriction\mathcal B
  +\omegasurf\,\mathcal H^{d-1}\measurerestriction\Gamma,
\end{equation}
{where $\omegareg\in L^1_+(\mathcal B)$ is the diffuse density and $\omegasurf\in L^1_+(\Gamma;\mathcal H^{d-1})$ the interfacial density, with units of length.}

A purely concentrated crack-strain state $\bmw$ satisfies
\begin{equation}
 (\bmw^h\otimes\bmN^h)\,\mathcal L^d\measurerestriction\mathcal B
 \stackrel{*}{\rightharpoonup}
 (\jumpu\otimes\bmN)\,\mathcal H^{d-1}\measurerestriction\Gamma.
\end{equation}
{Together with convergence of the deformation and bulk gradients, this recovers Eq.~\eqref{eq:distributional-gradient-measure}; any diffuse crack strain contributes to its regular part. The concentration width $h$ is distinct from $\ell$. Energy limits are taken along compatible, orientation-preserving fields satisfying the boundary conditions.}

{For the one-dimensional construction in Section~\ref{sec:energy-comparison}, fix $\Delta\geq0$, $0\leq\phi^*<1$, and total compliance mass $\ell\varpi>0$. In the band centered at $X^*$, choose}
\[
 \w^h=\frac{\Delta}{h}\mathbf1_{\damageband},
 \qquad
 \omega^h=\frac{\ell\varpi}{h}\mathbf1_{\damageband}.
\]
For every continuous test function $\zeta$, $\int_{\mathcal B}\w^h\zeta\,\dd X\to\Delta\zeta(X^*)$; an analogous limit holds for $\omega^h$. Thus
\begin{equation}\label{eq:1d-concentrated-measures}
 \w^h\,\mathcal L^1\measurerestriction\mathcal B
 \overset{*}{\rightharpoonup}\Delta\delta_{X^*},
 \qquad
 \omega^h\,\mathcal L^1\measurerestriction\mathcal B
 \overset{*}{\rightharpoonup}\omegameasure
 =\ell\varpi\delta_{X^*},
\end{equation}
{where $\delta_{X^*}$ is the unit Dirac measure, $\omegareg=0$, and $\omegasurf=\ell\varpi$. The energy limit is Eq.~\eqref{eq:1d-concentrated-energy} for $\Gamma=\{X^*\}$; minimizing over $\varpi$ recovers the condensation~\eqref{eq:1d-compliance-condensation}. For $\Delta=0$, use $\omega=0$.}

For a uniform bulk stretch, the distributional displacement derivative satisfies
\begin{equation}\label{eq:1d-concentrated-gradient-measure}
 Du=(\Fb-1)\,\mathcal L^1\measurerestriction\mathcal B
    +\Delta\delta_{X^*},
 \qquad
 \Fb=1+\frac{\bar u-\Delta}{L}.
\end{equation}
{For piecewise $C^1$ nonuniform bulk deformation, replace the volume term by $u_{,X}\,\mathcal L^1\measurerestriction\mathcal B$. Then $D\chi=\mathcal L^1\measurerestriction\mathcal B+Du$ recovers Eq.~\eqref{eq:strong-discontinuity-additive}. Complete fracture is the limit $\phi^*\to1$ and, for $\Delta>0$, $\varpi\to\infty$.}

\section{Bulk hyperelastic constitutive laws}\label{app:bulk-constitutive}
This appendix lists the bulk strain-energy densities and first Piola--Kirchhoff stresses used here. Energies are per unit reference volume, with $\bmPb=\partial\psib/\partial\bmFb$.

\subsection{Compressible Neo--Hookean law}\label{app:compressible-NH}

The compressible NH strain-energy density is
\begin{equation}\label{eq:bulk-NH-energy}
 \psib(\bmFb)
 =\frac{\mu}{2}(\bmFb:\bmFb-d)
 -\mu\ln\Jb+\frac{\lambda}{2}(\ln\Jb)^2,
 \qquad \Jb=\det\bmFb>0,
\end{equation}
where $\lambda$ and $\mu$ are Lam\'e constants and $d$ is the dimension of $\bmFb$. The bulk stress is
\begin{equation}\label{eq:bulk-NH-stress}
 \bmPb=\mu(\bmFb-\bmFb^{-\trans})
       +\lambda\ln\Jb\,\bmFb^{-\trans}.
\end{equation}
In plane strain, the three-dimensional bulk gradient is $\operatorname{diag}(\bmFb,1)$, and Eq.~\eqref{eq:bulk-NH-energy} reduces to the corresponding $d=2$ expression with unchanged Lam\'e constants.

\subsection{Two-dimensional incompressible plane-stress Neo--Hookean law}\label{app:incompressible-plane-stress-NH}

For a plane-stress sheet with negligible transverse shear, the three-dimensional bulk gradient and incompressibility condition are
\begin{equation}\label{eq:bulk-incompressible-embedding}
 \bmFb^{(3)}=\operatorname{diag}(\bmFb,\lambda_3),
 \qquad\det\bmFb^{(3)}=\Jb\lambda_3=1,
 \qquad\lambda_3=\Jb^{-1},
\end{equation}
{where $\lambda_3$ is the thickness stretch and $\Jb$ the in-plane area ratio. Zero transverse traction determines the pressure; eliminating $\lambda_3$ gives an energy in $\bmFb$.}

Let $\bmC^{(3)}=(\bmFb^{(3)})^\trans\bmFb^{(3)}$. Under reduction~\eqref{eq:bulk-incompressible-embedding}, its first two principal invariants are
\begin{equation}\label{eq:bulk-plane-stress-invariants}
 I_1=\operatorname{tr}\bmC^{(3)}=a+\Jb^{-2},\qquad
 I_2=\tfrac12\bigl[(\operatorname{tr}\bmC^{(3)})^2
                 -\operatorname{tr}((\bmC^{(3)})^2)\bigr]
     =\Jb^2+\Jb^{-2}a,
 \qquad a\coloneqq\bmFb:\bmFb.
\end{equation}
Denote their derivatives with respect to the in-plane bulk gradient by
\begin{align}
 &\bulkInvGradient{1}
 \coloneqq\tfrac12\frac{\partial I_1}{\partial\bmFb}
 =\bmFb-\Jb^{-2}\bmFb^{-\trans},\label{eq:bulk-invariant-gradient-1}\\
 &\bulkInvGradient{2}
 \coloneqq\tfrac12\frac{\partial I_2}{\partial\bmFb}
 =\Jb^{-2}\bmFb+(\Jb^2-\Jb^{-2}a)\bmFb^{-\trans}.
 \label{eq:bulk-invariant-gradient-2}
\end{align}
The incompressible NH plane-stress energy and in-plane stress are
\begin{equation}\label{eq:bulk-incompressible-NH}
 \psib=\frac{\mu}{2}(I_1-3),\qquad\bmPb=\mu\bulkInvGradient{1},
\end{equation}
where $\mu>0$ is the infinitesimal shear modulus of the three-dimensional material in the reference state.
Under uniaxial bulk stretch $\Fb$ with both transverse directions traction-free, incompressibility gives a transverse stretch of $\Fb^{-1/2}$. The one-dimensional energy and axial nominal stress are
\begin{equation}\label{eq:bulk-NH-uniaxial}
 \psib(\Fb)=\frac\mu2\left(\Fb^2+\frac2{\Fb}-3\right),
 \qquad P=\psib'(\Fb)=\mu(\Fb-\Fb^{-2}).
\end{equation}

\subsection{Generalized Neo--Hookean law}\label{app:generalized-NH}

{The generalized Neo--Hookean (GNH) law of \citet{knowles1977finite}, studied in plane-stress crack-tip analysis by \citet{long2011finite} and dynamic phase-field fracture by \citet{tian2025nonlinearity}, has energy density}
\begin{equation}\label{eq:bulk-GNH-energy}
 \psib=\frac{\mu}{2b}\bigl(\xi^p-1\bigr),\qquad
 \xi\coloneqq1+\frac{b}{p}(I_1-3),\qquad \mu>0,\ b>0,\ p>0,
\end{equation}
{where $\mu$ is the reference shear modulus, $b$ and $p$ are dimensionless, and $I_1$ follows Eq.~\eqref{eq:bulk-plane-stress-invariants} under incompressible plane stress. The in-plane stress is}
\begin{equation}\label{eq:bulk-GNH-stress}
 \bmPb=\mu\xi^{p-1}\bulkInvGradient{1}.
\end{equation}
Three-dimensional incompressibility ensures $I_1\geq3$ and thus $\xi\geq1$. For $p=1$, the model reduces to incompressible NH elasticity~\eqref{eq:bulk-incompressible-NH}.

\subsection{Mooney--Rivlin law}\label{app:Mooney-Rivlin}

The Mooney--Rivlin energy is a linear combination of the first two principal invariants of the right Cauchy--Green tensor. In two-dimensional incompressible plane stress,
\begin{equation}\label{eq:bulk-MR-energy}
 \psib=\frac{\mu}{2}\bigl[c_{10}(I_1-3)+c_{01}(I_2-3)\bigr],
 \qquad c_{10}+c_{01}=1,
\end{equation}
where $c_{10}$ and $c_{01}$ are dimensionless. This normalization makes $\mu$ the reference infinitesimal shear modulus. The conventional dimensional coefficients are $C_{10}=\mu c_{10}/2$ and $C_{01}=\mu c_{01}/2$. The in-plane stress is
\begin{equation}\label{eq:bulk-MR-stress}
 \bmPb=\mu(c_{10}\bulkInvGradient{1}+c_{01}\bulkInvGradient{2}).
\end{equation}
The choice $c_{10}=1$, $c_{01}=0$ recovers incompressible NH elasticity.
For traction-free transverse contraction under uniaxial loading, $I_1=\Fb^2+2/\Fb$ and $I_2=2\Fb+\Fb^{-2}$, giving
\begin{equation}\label{eq:bulk-MR-uniaxial}
\psib(\Fb)=\frac\mu2\left[c_{10}(\Fb^2+2/\Fb-3)
                         +c_{01}(2\Fb+\Fb^{-2}-3)\right].
\end{equation}

The axial nominal stress is $P=\mu[c_{10}(\Fb-\Fb^{-2})+c_{01}(1-\Fb^{-3})]$.

\subsection{Saint Venant--Kirchhoff law}\label{app:Saint-Venant-Kirchhoff}

{The Saint Venant--Kirchhoff model is quadratic in the bulk Green--Lagrange strain. With $\bmIref$ the reference identity, the strain and second Piola--Kirchhoff stress are}
\begin{equation}\label{eq:bulk-StVK-strain-stress}
 \bulkquantity{\bmE}\coloneqq\tfrac12(\bmFb^\trans\bmFb-\bmIref),\qquad
 \bulkquantity{\bmS}\coloneqq\lambda\operatorname{tr}(\bulkquantity{\bmE})\bmIref
                         +2\mu\bulkquantity{\bmE}.
\end{equation}
The energy and bulk stress are
\begin{equation}\label{eq:bulk-StVK-energy-stress}
 \psib=\frac{\lambda}{2}[\operatorname{tr}(\bulkquantity{\bmE})]^2
      +\mu\bulkquantity{\bmE}:\bulkquantity{\bmE},\qquad
 \bmPb=\bmFb\bulkquantity{\bmS}.
\end{equation}
Equation~\eqref{eq:bulk-StVK-energy-stress} applies to two- or three-dimensional gradients. With the same Lam\'e constants, its two-dimensional form is the plane-strain restriction of the three-dimensional model.
For the two-dimensional energy, the Lam\'e constants are parameterized by the two-dimensional Poisson ratio $\nu$, with $\lambda=2\mu\nu/(1-\nu)$. Under uniaxial loading with traction-free transverse boundaries, the transverse bulk stretch $\lambda_y$ satisfies $\lambda_y^2=1-\nu(\Fb^2-1)>0$. Eliminating it gives
\begin{equation}\label{eq:bulk-StVK-uniaxial}
 \psib(\Fb)=\frac{\mu(1+\nu)}4(\Fb^2-1)^2,
 \qquad P=\mu(1+\nu)\Fb(\Fb^2-1).
\end{equation}

\section{Geometric Piola stress, direction-conjugate force, and angular momentum balance}\label{app:geometric-Piola-stress}
\subsection{Variation of the constrained energy and geometric stress}
{To retain normal transport in the energy variation, set $\bmm=\bmF^{-\trans}\bmN$ at fixed $\omega>0$. Then}
\begin{equation}\label{eq:normal-full-variation}
 \delta\bmn
 =-\bmItan\bmF^{-\trans}(\delta\bmF)^\trans\bmn
  +\frac{\bmItan\bmF^{-\trans}\delta\bmN}{\lVert\bmm\rVert}.
\end{equation}
The unit reference direction satisfies $\bmN\cdot\delta\bmN=0$. Using the constrained density $\psilambda=\psi+\lambdaperp\wperp$ and generalized force $\bmq$ from the main text, $\delta\bmFb=\delta\bmF-\delta\bmw\otimes\bmN-\bmw\otimes\delta\bmN$ gives
\begin{equation}\label{eq:appendix-first-variation}
 \delta\psilambda={}\bmPb:\delta\bmF
 +(\bmKomega\bmw-\bmPb\bmN+\lambdaperp\bmn)\cdot\delta\bmw
 -(\bmPb^\trans\bmw)\cdot\delta\bmN
 +\bmq\cdot\delta\bmn+\wperp\delta\lambdaperp.
\end{equation}
At fixed $\bmN$, the work associated with changes in the normal is
\[
 \bmq\cdot\delta\bmn
 =-\left[\bmn\otimes(\bmF^{-1}\bmqparallel)\right]:\delta\bmF.
\]
The coefficients of $\delta\bmF$ and $\delta\bmw$ therefore give the geometric stress~\eqref{eq:geometric-Piola-stress} and local traction balance~\eqref{eq:w-equilibrium}, respectively.

\subsection{Direction-conjugate force}
{Collecting the $\delta\bmN$ terms in Eq.~\eqref{eq:appendix-first-variation} gives the direction-conjugate force}
\begin{equation}\label{eq:direction-conjugate-force}
 \bmQ=\partial_{\bmN}\psilambda
 =-\bmPb^\trans\bmw
  +\frac{\bmF^{-1}\bmqparallel}{\lVert\bmF^{-\trans}\bmN\rVert}.
\end{equation}
Stationarity on the unit sphere is $(\bmIref-\bmN\otimes\bmN)\bmQ=\mathbm0$ and involves only the tangential projection of the direction-conjugate force. On the closed branch, $\wperp=0$ gives $\bmqparallel=\lambdaperp\bmwparallel$. Thus contact under compression with tangential sliding can generate geometric stress even for $\kperp=\kparallel$.

\subsection{Angular momentum balance}\label{app:angular-momentum}
{For $\omega>0$ on a differentiable branch, objectivity under any rotation $\bmR\in\mathrm{SO}(d)$ gives}
\[
 \psilambda(\bmR\bmF,\bmR\bmw,\bmN,\omega;\lambdaperp)
 =\psilambda(\bmF,\bmw,\bmN,\omega;\lambdaperp).
\]
Differentiating along $\bmR(s)=\exp(s\mathbm W)$ gives, for every skew-symmetric tensor $\mathbm W$,
\begin{equation}\label{eq:angular-momentum-identity}
 \left[\bmP\bmF^\trans
  +\left(\frac{\partial\psilambda}{\partial\bmw}\right)\otimes\bmw\right]
 :\mathbm W=0.
\end{equation}
{Local equilibrium gives $\partial_{\bmw}\psilambda=\mathbm0$, proving symmetry of $\bmP\bmF^\trans$ and the full Cauchy stress~\eqref{eq:cauchy-symmetry}.}

\section{Principal-plane reduction of three-dimensional crack directions}
\label{app:crack-direction-reduction}
{To prove Proposition~\ref{prop:crack-direction-plane-reduction}, we replace any candidate normal by a principal-plane normal without increasing the energy. Fix $\bmF$ and $\omega>0$ under assumption~\eqref{eq:bulk-plane-reduction-assumption}, and define}
\begin{equation}\label{eq:direction-plane-invariants}
 s\coloneqq\lVert\bmF^{-\trans}\bmN\rVert,
 \qquad r\coloneqq\lVert\bmItan\bmF\bmN\rVert.
\end{equation}
Since $\bmF^{\trans}\bmn=\bmN/s$ gives $\bmn\cdot\bmF\bmN=1/s$, decompose the crack strain $\bmw$ as $\bmw=\wperp\bmn+\bmwparallel$ and write $z\coloneqq\lVert\bmwparallel\rVert$. Then
\begin{equation}\label{eq:direction-plane-bulk-invariants}
 \Jb=J(1-s\wperp),
 \qquad
 I_1=\bmF:\bmF-\frac{2\wperp}{s}+(\wperp)^2+z^2
       -2\bmwparallel\cdot\bmItan\bmF\bmN.
\end{equation}
The admissible domain is thus $0\leq\wperp<1/s$. At fixed $s$, $\wperp$, and $z$, the Cauchy--Schwarz inequality gives
$\bmwparallel\cdot\bmItan\bmF\bmN\leq zr$.
Since $\psibhat$ is nondecreasing in $I_1$, aligning $\bmwparallel$ with $\bmItan\bmF\bmN$ cannot increase the energy. If $r=0$, the tangential direction does not affect the energy. The reduced energy can therefore be written as
\begin{equation}\label{eq:direction-plane-scalar-energy}
 \psired(\bmF,\bmN,\omega)=\mathcal G(s,r),
\end{equation}
where
\begin{equation}\label{eq:direction-plane-scalar-minimization}
 \mathcal G(s,r)=
 \min_{\substack{0\leq\wperp<1/s\\z\geq0}}
 \psibhat\!\left(
 \bmF:\bmF-\frac{2\wperp}{s}+(\wperp)^2+z^2-2zr,
 J(1-s\wperp)\right)
 +\frac{\ell}{2\omega}
 \left[\kperp(\wperp)^2+\kparallel z^2\right].
\end{equation}
For each admissible pair $(\wperp,z)$, energy is nonincreasing in $r$ at fixed $s$; minimization preserves this monotonicity.

We now construct a direction in the plane of the largest and smallest principal stretches with the same $s$ and no smaller $r$. Let
$x_i\coloneqq(\bmN\cdot\bmN_i)^2$ and $a_i\coloneqq\lambda_i^{-2}$, so that $x_i\geq0$, $\sum_i x_i=1$, and
\begin{equation}\label{eq:direction-plane-spectral-invariants}
 s^2=\sum_{i=1}^3a_i x_i,
 \qquad
 r^2=\sum_{i=1}^3\frac{x_i}{a_i}-\frac1{s^2}.
\end{equation}
For $a_1<a_3$, let $\alpha\coloneqq(a_3-a_2)/(a_3-a_1)\in[0,1]$, giving $a_2=\alpha a_1+(1-\alpha)a_3$. Construct
\begin{equation}\label{eq:direction-endpoint-replacement}
 x'_1=x_1+\alpha x_2,
 \qquad x'_2=0,
 \qquad x'_3=x_3+(1-\alpha)x_2.
\end{equation}
The weights in Eq.~\eqref{eq:direction-endpoint-replacement} define a unit direction $\bmN'\in\operatorname{span}\{\bmN_1,\bmN_3\}$ with $s'=s$. Convexity of $a\mapsto1/a$ gives
\begin{equation}
 (r')^2-r^2
 =x_2\left[\frac\alpha{a_1}+\frac{1-\alpha}{a_3}-\frac1{a_2}\right]
 \geq0.
\end{equation}
Monotonicity of $\mathcal G$ then implies
\begin{equation}\label{eq:direction-plane-energy-dominance}
 \psired(\bmF,\bmN',\omega)
 =\mathcal G(s,r')\leq\mathcal G(s,r)
 =\psired(\bmF,\bmN,\omega).
\end{equation}
Applying the replacement to an optimal direction proves the principal-plane reduction. Dependence on squared principal components in Eq.~\eqref{eq:direction-plane-spectral-invariants} gives the interval~\eqref{eq:N-principal-plane}. If all principal stretches coincide, $s$ and $r$ are direction-independent; partially repeated stretches are covered by the same construction.

\section{Variational derivation of damage evolution and initiation}\label{app:damage-inequality}
\subsection{Weak form of irreversibility}
{At fixed $\bmchi$ and $\bmN$, with $\bmw$ condensed, irreversible minimization gives the variational inequality}
\begin{equation}\label{eq:damage-minimum-conditions}
 \firstvariation{\alpha}\energyred[\eta_\alpha-\alpha]\geq0,
 \qquad\alpha\in\{\omega,\phi\}.
\end{equation}
{Here only $\alpha$ varies; comparison fields obey $\eta_\omega\geq\omega_{\mathrm{prev}}$, $\phi_{\mathrm{prev}}\leq\eta_\phi\leq1$, the boundary data, and finite-energy requirements. For the phase-field,}
\begin{equation}\label{eq:phi-weak-inequality}
\firstvariation{\phi}\energyred[\eta-\phi]
 =\frac{\Gc}{\ell}\int_{\mathcal B}
 \left\{\left[\Omega(\phi)\Omega'(\phi)-2(1-\phi)\omega\right](\eta-\phi)
 +\ell^2\nabla\phi\cdot\nabla(\eta-\phi)\right\}\,\dd V\geq0.
\end{equation}
{For admissible $\eta\in H^1(\mathcal B)$ and smooth fields, integration by parts gives}
\begin{equation}\label{eq:phi-weak-boundary-form}
 \int_{\mathcal B}(\variationalderivative{\energyred}{\phi})(\eta-\phi)\,\dd V
 +\Gc\ell\int_{\partial\mathcal B}
 (\nabla\phi\cdot\bmnu)(\eta-\phi)\,\dd A\geq0.
\end{equation}
On a boundary portion with no prescribed phase-field value and $\phi<1$, local admissible variations give $\nabla\phi\cdot\bmnu\geq0$ and $(\nabla\phi\cdot\bmnu)(\phi-\phi_{\mathrm{prev}})=0$. Under smooth continuous evolution, these correspond to the rate conditions \citep{pham2011approximate} in Eq.~\eqref{eq:phi-boundary-KKT}.

The incremental condition for $\omega$ is $\firstvariation{\omega}\energyred[\eta-\omega]=\int_{\mathcal B}(\variationalderivative{\energyred}{\omega})(\eta-\omega)\,\dd V\geq0$, where $\eta\geq\omega_{\mathrm{prev}}$. For sufficiently regular derivative densities and loading paths, these inequalities yield the pointwise KKT conditions in the main text.

\subsection{Derivative of the reduced energy with respect to \texorpdfstring{$\omega$}{omega}}
{At fixed $\bmF$ and $\bmN$, a differentiable branch of minimizers $\bmw(\omega)$ for $\omega>0$ gives}
\begin{equation}\label{eq:reduced-energy-chain-rule}
\frac{\partial\psired}{\partial\omega}
 =\frac{\partial\psi}{\partial\omega}
  +\frac{\partial\psi}{\partial\bmw}\cdot\frac{\partial\bmw}{\partial\omega}
 =\frac{\partial\psic}{\partial\omega}
  -\lambdaperp\bmn\cdot\frac{\partial\bmw}{\partial\omega}
 =\frac{\partial\psic}{\partial\omega}.
\end{equation}
Local equilibrium gives $\partial_{\bmw}\psi=-\lambdaperp\bmn$. The contact term vanishes on both branches: $\lambdaperp=0$ under opening, while closure at fixed $\bmn$ gives $\bmn\cdot\partial_\omega\bmw=0$.

Traction balance and contact complementarity give $\bmn\cdot\bmKomega\bmw=\positivepart{(\Tperp)}$ and $\bmItan\bmKomega\bmw=\bmTparallel$. Hence
\begin{equation}\label{eq:direct-strength-derivative}
 -\frac{\partial\psired}{\partial\omega}
 =\frac{\ell}{2\omega^2}
   \left[\kperp(\wperp)^2+\kparallel\lVert\bmwparallel\rVert^2\right]
 =\frac{\Gc}{\ell}
   \left[\positivepartexpr{\frac{\Tperp}{\Tperpc}}^2
         +\left(\frac{\Tparallel}{\Tparallelc}\right)^2\right].
\end{equation}
{Adding the fracture-energy derivative yields Eq.~\eqref{eq:reduced-strength-relation}; the zero-damage limit follows below.}

\subsection{One-sided energy expansion and direction selection at initiation}\label{app:damage-initiation}
{The one-sided energy expansion determines direction selection at initiation. Fix $\det\bmF>0$ and $\psib\in C^2$ near $\bmF$. Let $\mathbm v_\omega$ minimize Eq.~\eqref{eq:w-reduced-energy}, with $\mathbm v_\omega\to\mathbm0$ as $\omega\to0^+$. For direction selection, assume this convergence is uniform in $\bmN\in\mathbb S^{d-1}$; $\partial_\omega^+$ denotes the right derivative at zero.}

Let $\bmT_0=\bmPb(\bmF)\bmN$, with normal component $T_0^{\normalmode}=\bmn\cdot\bmT_0$ and tangential magnitude $T_0^{\tangentialmode}=\lVert\bmItan\bmT_0\rVert$. The local bulk-energy expansion is
\[
 \psib(\bmF-\mathbm v\otimes\bmN)
 =\psib(\bmF)-\bmT_0\cdot\mathbm v+O(\lVert\mathbm v\rVert^2).
\]
Comparison with the admissible trial $\mathbm v=\mathbm0$, with $k_{\min}=\min(\kperp,\kparallel)>0$, gives $(\ell k_{\min}/(2\omega)-C)\lVert\mathbm v_\omega\rVert^2\leq\lVert\bmT_0\rVert\lVert\mathbm v_\omega\rVert$, where $C$ is independent of sufficiently small $\omega$. Thus $\mathbm v_\omega=O(\omega)$, the Taylor remainder is $O(\omega^2)$, and orientation preservation holds strictly in this neighborhood.

Minimizing the leading quadratic energy $-\bmT_0\cdot\mathbm v+\ell[\kperp(v^{\normalmode})^2+\kparallel\lVert\mathbm v^{\tangentialmode}\rVert^2]/(2\omega)$ over $v^{\normalmode}\geq0$ gives
\begin{equation}\label{eq:intact-w-asymptotics}
 v_\omega^{\normalmode}
 =\frac{\omega}{\ell\kperp}\positivepartexpr{T_0^{\normalmode}}+o(\omega),
 \qquad
 \mathbm v_\omega^{\tangentialmode}
 =\frac{\omega}{\ell\kparallel}\bmItan\bmT_0+o(\omega).
\end{equation}
Substituting into the reduced energy and using $(\Tperpc)^2=2\Gc\kperp$ and $(\Tparallelc)^2=2\Gc\kparallel$ gives
\begin{equation}\label{eq:intact-directional-energy-expansion}
\psired(\bmF,\bmN,\omega)
 =\psib(\bmF)-\frac{\omega}{2\ell}
   \left[\frac{\positivepartexpr{T_0^{\normalmode}}^2}{\kperp}
         +\frac{(T_0^{\tangentialmode})^2}{\kparallel}\right]+o(\omega).
\end{equation}
By definition of the strength function, the bracket equals $2\Gc[1+\reducedstrength(\bmF,\bmN,0,0)]$, so $-\partial_\omega^+\psired=\Gc[1+\reducedstrength(\bmF,\bmN,0,0)]/\ell$. Since $\psired$ is independent of $\phi$, adding $\Gc(1-\phi)^2/\ell$ for any given phase-field value gives Eq.~\eqref{eq:reduced-strength-relation} at $\omega=0$.

The assumed uniform convergence in $\bmN$ gives a uniform $O(\omega^2)$ remainder in Eq.~\eqref{eq:intact-directional-energy-expansion} by comparison with the leading quadratic minimizer. The leading directional term is therefore $-\Gc\omega\reducedstrength/\ell$. If $\bmN_\omega$ minimizes energy and $\bmNinit$ maximizes intact strength, minimality gives $\reducedstrength(\bmF,\bmN_\omega,0,0)\geq\reducedstrength(\bmF,\bmNinit,0,0)-o(1)$. Compactness and continuity then place every accumulation point of $\bmN_\omega$ among the most critical planes; higher-order terms can select among tied directions.

\section{One-dimensional phase-field profiles and damage-band concentration}\label{app:1d-phase-profile}
\subsection{First integral of the phase-field tails}
{On the loading tails outside $\damageband=(X_-,X_+)$, $\omega=0$ and $\dot\phi>0$. Following \citet{feng2025JMPS}, multiplication of Eq.~\eqref{eq:1d-phase-tail-ode} by $2\phi_{,X}$ gives}
\begin{equation}\label{eq:1d-tail-conserved-quantity}
 \partial_X\!\left[\Omega(\phi)^2-\ell^2\phi_{,X}^2\right]=0.
\end{equation}
{The intact-state conditions $\phi,\phi_{,X}\to0$ and $\Omega(0)=0$ set the constant on each tail to zero, giving}
\begin{equation}\label{eq:1d-tail-first-integral}
 \ell\lvert\phi_{,X}\rvert=\Omega(\phi),
 \qquad X\in\mathcal B_\phi\setminus\overline{\damageband}.
\end{equation}

Since $\Omega(\phi)>0$ for $\phi>0$, the gradient has a fixed sign: positive on the left tail and negative on the right. Thus
\begin{equation}\label{eq:1d-tail-signed-gradient}
 \phi_{,X}=
 \begin{cases}
  \Omega(\phi)/\ell, & X<X_-,\\[0.4em]
  -\Omega(\phi)/\ell, & X>X_+,
 \end{cases}
 \qquad X\in\mathcal B_\phi\setminus\overline{\damageband}.
\end{equation}
Separating variables with boundary value $\phi^*$ gives the distance integral~\eqref{eq:1d-phase-tail}. For example, $\Omega(\phi)=\phi$ yields exponential tails and, together with the constant phase-field inside the band,
\begin{equation}\label{eq:1d-exponential-phase-profile}
 \phi(X)=\phi^*\exp\!\left[-\frac{\operatorname{dist}(X,\damageband)}{\ell}\right].
\end{equation}

{To check irreversibility for general $\Omega$, differentiate Eq.~\eqref{eq:1d-phase-tail} at fixed band and $X$:}
\begin{equation}\label{eq:1d-tail-amplitude-monotonicity}
 \frac{\partial\phi(X;\phi^*)}{\partial\phi^*}
 =\frac{\Omega(\phi(X;\phi^*))}{\Omega(\phi^*)}>0,
 \qquad 0<\phi(X;\phi^*)<\phi^*.
\end{equation}
{Thus $\dot\phi^*>0$ increases $\phi$ pointwise. The distance $\ell\int_\eta^{\phi^*}\dd s/\Omega(s)$ of a fixed level $\eta\in(0,\phi^*)$ from the band also increases, so its enclosed width cannot decrease.}

\subsection{Compatibility at the band boundary and concentration}\label{app:1d-band-concentration}
{Following \citet{feng2025JMPS}, we test compatibility where the plateau meets the tails. Take $0<h<L$, $0<\phi^*<1$, and finite $\omega$, with phase-field loading through the band boundary. Uniform stress and compliance-damage loading give $\phi=\phi^*$ and $\phi_{,X}=0$ inside. Bounded local terms in the weak phase-field equation require}
\begin{equation}\label{eq:1d-band-flux-continuity}
 \jump{\Gc\ell\phi_{,X}}(X_-)=\jump{\Gc\ell\phi_{,X}}(X_+)=0.
\end{equation}
Here the jump is the right-hand limit minus the left-hand limit. Joining the tails~\eqref{eq:1d-tail-signed-gradient} to the zero-gradient plateau instead gives
\begin{equation}\label{eq:1d-band-gradient-incompatibility}
 \jump{\phi_{,X}}(X_-)=\jump{\phi_{,X}}(X_+)=-\frac{\Omega(\phi^*)}{\ell}<0,
\end{equation}
contradicting Eq.~\eqref{eq:1d-band-flux-continuity}: the gradient jump produces a Dirac term that bounded local terms cannot balance. Equilibrium on the localized branch therefore requires concentration, $h\to0^+$.

{At fixed $\Delta>0$, Eqs.~\eqref{eq:1d-concentrated-measures} and~\eqref{eq:1d-concentrated-gradient-measure} then give $\jump{u}(X^*)=\Delta$, while the phase-field retains continuous tails governed by $\ell$.}

\section{Common cohesive laws and their energy constructions}\label{app:cohesive-laws}
In the dimensionless variables of Section~\ref{sec:cohesive-law}, the following softening envelopes and characteristic functions $\Omega$ \citep{feng2021endowing} follow by inverting Eq.~\eqref{eq:mixed-mode-parametric-cohesive-law}.

\subsection{Linear softening}\label{app:cohesive-linear}
The linear cohesive law and its characteristic function are
\begin{equation}\label{eq:cohesive-linear-law}
 \normalizedeffectivetraction
 =\positivepartexpr{1-\normalizedeffectiveopening},
 \qquad \Omega(\phi)=\phi.
\end{equation}
Complete separation occurs at $\normalizedeffectiveopening=1$.

\subsection{Exponential softening}\label{app:cohesive-exponential}
The exponential cohesive law and its characteristic function are
\begin{equation}\label{eq:cohesive-exponential-law}
 \normalizedeffectivetraction
 =\exp(-2\normalizedeffectiveopening),
 \qquad \Omega(\phi)=-\frac12\ln(1-\phi).
\end{equation}
Traction tends to zero asymptotically as separation increases, and the total work of separation is $\Gc$.

\subsection{Bilinear softening}\label{app:cohesive-bilinear}
The bilinear law has two descending segments. Let $a\in(0,1)$ be the ratio of traction at the kink to peak traction, and $b\in(0,1)$ the ratio of separation at the kink to complete separation. Energy normalization gives complete separation at $\normalizedeffectiveopening=1/(a+b)$, hence
\begin{equation}\label{eq:cohesive-bilinear-law}
 \normalizedeffectivetraction=
 \begin{cases}
  1-\dfrac{(1-a)(a+b)}{b}\normalizedeffectiveopening,
    &0\leq\normalizedeffectiveopening\leq\dfrac{b}{a+b},\\[0.7em]
  \dfrac{a}{1-b}\left[1-(a+b)\normalizedeffectiveopening\right],
    &\dfrac{b}{a+b}<\normalizedeffectiveopening\leq\dfrac1{a+b},\\[0.7em]
  0,&\normalizedeffectiveopening>\dfrac1{a+b}
 \end{cases}.
\end{equation}
The corresponding characteristic function given by \citet{chen2025high} is
\begin{equation}\label{eq:cohesive-bilinear-Omega}
 \Omega(\phi)=
 \begin{cases}
  \dfrac{b\phi}{(1-a)(a+b)},&0\leq\phi\leq1-a,\\[0.7em]
  \dfrac{a+b-1+(1-b)\phi}{a(a+b)},&1-a<\phi\leq1
 \end{cases}.
\end{equation}
The condition $a+b\leq1$ ensures convexity of the softening envelope and $\Omega^2$. For $a+b=1$, the law reduces to linear softening. A piecewise-smooth representation is used at the kink.

\subsection{PPR polynomial softening}\label{app:cohesive-PPR}
In its extrinsic pure-mode-I form, with no initial elastic branch, the PPR model of \citet{park2009unified} gives a family of power-law softening curves. The envelope and characteristic function are
\begin{equation}\label{eq:cohesive-PPR-law}
 \normalizedeffectivetraction
 =\positivepartexpr{1-\frac{2\normalizedeffectiveopening}{\alpha}}^{\alpha-1},
 \qquad
 \Omega(\phi)=\frac\alpha2
 \left[1-(1-\phi)^{1/(\alpha-1)}\right],
\end{equation}
where $\alpha\geq2$ is the shape parameter; this range ensures convexity of $\Omega^2$. Complete separation occurs at $\normalizedeffectiveopening=\alpha/2$. The cases $\alpha=2$ and $\alpha=3$ give linear and quadratic softening, respectively. The PPR model supplies the pure-mode softening shape; the present model governs mixed-mode coupling.

\subsection{Fixed quadratic phase-field term and the \texorpdfstring{$p$}{p}-model}\label{app:cohesive-p-model}
{To calibrate softening with a fixed quadratic phase-field term, replace $(1-\phi)^2$ in Eq.~\eqref{eq:fracture-energy-density} by $\damagecoupling(\phi)^2$, with $\damagecoupling(0)=1$, $\damagecoupling(1)=0$, and $\damagecoupling'<0$ for $0\leq\phi<1$.} Following \citet{feng2025JMPS}, we obtain the concentrated-state equilibrium relations
\begin{equation}\label{eq:cohesive-alternative-parametrization}
 \normalizedeffectivetraction=\damagecoupling(\phi^*),
 \qquad
 \normalizedeffectiveopening=-\frac{\Omega(\phi^*)}{\damagecoupling'(\phi^*)}.
\end{equation}
One may therefore fix $\damagecoupling=1-\phi$ and construct $\Omega$, or fix $\Omega=\phi$ and construct $\damagecoupling$. For a target envelope $\normalizedeffectivetraction=Q(\normalizedeffectiveopening)$, the latter construction satisfies
\begin{equation}\label{eq:cohesive-fixed-Omega-calibration}
 \phi^2=2\int_{\damagecoupling(\phi)}^1 Q^{-1}(s)\,\dd s,
\end{equation}
where $Q^{-1}$ is the inverse of the strictly monotone softening branch, with endpoint values understood as limits. Both constructions give the same softening envelope and work of separation in the concentrated limit, but their phase-field profiles and finite-mesh responses can differ.

{The $p$-model of \citet{feng2025JMPS} takes $\Omega=\phi$ and $\damagecoupling=(1-\phi)^p$, $p\geq1$. With the same recoverable energy $\psi$, its fracture-energy density is}
\begin{equation}\label{eq:p-model-energy-density}
 \varphif_p(\phi,\nabla\phi,\omega)
 =\frac{\Gc}{2\ell}\left[\phi^2+\ell^2\lvert\nabla\phi\rvert^2\right]
 +\frac{\Gc}{\ell}(1-\phi)^{2p}\omega.
\end{equation}
The total internal-energy density is $\psi+\varphif_p$. Its cohesive law is
\begin{equation}\label{eq:p-model-cohesive-law}
 \normalizedeffectivetraction=(1-\phi^*)^p,
 \qquad
 \normalizedeffectiveopening=\frac{\phi^*}{p}(1-\phi^*)^{1-p}.
\end{equation}
The same law corresponds, in the $\Omega^2$ form used in the main text, to
\begin{equation}\label{eq:p-model-equivalent-Omega}
 \Omega_p(\phi)=\frac1p
 \left[1-(1-\phi)^{1/p}\right](1-\phi)^{(1-p)/p}.
\end{equation}
Each representation uses its own phase-field parameter $\phi^*$. Fixing $\Omega=\phi$ retains a quadratic local phase-field term and exponential tails, simplifying computation. The case $p=1$ gives linear softening; for $p>1$, traction vanishes asymptotically, unlike the finite complete-separation opening of the PPR law.

\section*{Declaration of generative AI and AI-assisted technologies in the manuscript preparation process}
During preparation of the manuscript, the author used ChatGPT and Codex (OpenAI) and Gemini (Google) for English language editing. The author reviewed and revised the text and takes full responsibility for its content.

\bibliographystyle{elsarticle-harv}
\bibliography{document}

@article{barenblatt1962mathematical,



  title={The mathematical theory of equilibrium cracks in brittle fracture},



  author={Barenblatt, G. I.},



  journal={Advances in Applied Mechanics},



  volume={7},



  number={1},



  pages={55--129},



  year={1962}



}

@book{simo2006computational,

  title={Computational Inelasticity},

  author={Simo, J.C. and Hughes, T.J.R.},

  isbn={9780387227634},

  lccn={97026427},

  series={Interdisciplinary Applied Mathematics},

  year={2006},

  publisher={Springer New York}

}

@article{camacho1996computational,



  title={Computational modelling of impact damage in brittle materials},



  author={Camacho, Godofredo T and Ortiz, Michael},



  journal={International Journal of solids and structures},



  volume={33},



  number={20-22},



  pages={2899--2938},



  year={1996},



  publisher={Elsevier}



}

@article{karma2005crack,
  title={Crack path prediction in anisotropic brittle materials},
  author={Hakim, Vincent and Karma, Alain},
  journal={Physical review letters},
  volume={95},
  number={23},
  pages={235501},
  year={2005},
  publisher={APS},
  doi={10.1103/PhysRevLett.95.235501}
}

@article{hakim2009laws,
  title={Laws of crack motion and phase-field models of fracture},
  author={Hakim, Vincent and Karma, Alain},
  journal={Journal of the Mechanics and Physics of Solids},
  volume={57},
  number={2},
  pages={342--368},
  year={2009},
  publisher={Elsevier},
  doi={10.1016/j.jmps.2008.10.012}
}

@article{wu2017unified,



  title={A unified phase-field theory for the mechanics of damage and quasi-brittle failure},



  author={Wu, Jian-Ying},



  journal={Journal of the Mechanics and Physics of Solids},



  volume={103},



  pages={72--99},



  year={2017},



  publisher={Elsevier}



}

@article{pham2011approximate,



author = {Kim Pham and Hanen Amor and Jean-Jacques Marigo and Corrado Maurini},



title ={Gradient Damage Models and Their Use to Approximate Brittle Fracture},



journal = {International Journal of Damage Mechanics},



volume = {20},



number = {4},



pages = {618-652},



year = {2011}



}

@article{bourdin2025variational,
  title={A variational approach to fracture incorporating any convex strength criterion},
  author={Bourdin, Blaise and Marigo, Jean-Jacques and Maurini, Corrado and Zolesi, Camilla},
  journal={arXiv preprint arXiv:2506.22558},
  year={2025}
}

@article{lorentz2011convergence,
  title={Convergence of a gradient damage model toward a cohesive zone model},
  author={Lorentz, Eric and Cuvilliez, Sam and Kazymyrenko, Kyrylo},
  journal={Comptes Rendus M\'ecanique},
  volume={339},
  number={1},
  pages={20--26},
  year={2011},
  doi={10.1016/j.crme.2010.10.010}
}

@article{lorentz2012modelling,



  title={Modelling large crack propagation: from gradient damage to cohesive zone models},
  doi={10.1007/s10704-012-9746-7},



  author={Lorentz, Eric and Cuvilliez, Sam and Kazymyrenko, Kyrylo},



  journal={International journal of fracture},



  volume={178},



  number={1-2},



  pages={85--95},



  year={2012},



  publisher={Springer}



}

@article{FRANCFORT19981319,



title = "Revisiting brittle fracture as an energy minimization problem",



journal = "Journal of the Mechanics and Physics of Solids",



volume = "46",



number = "8",



pages = "1319 - 1342",



year = "1998",



issn = "0022-5096",



author = "G.A. Francfort and J.-J. Marigo"



}

@article{feng_Li2026,
title = {Convergence of phase-field models with emergent discontinuities to {CZM}s in multidimensional space: A proof via cohesive laws and sharp-crack energy limits},
journal = {Journal of the Mechanics and Physics of Solids},
volume = {215},
pages = {106705},
year = {2026},
issn = {0022-5096},
doi = {10.1016/j.jmps.2026.106705},
author = {Ye Feng and Jie Li}
}

@article{feng20253d,
  title={{3D} phase-field cohesive fracture: Unifying energy, driving force, and stress criteria for crack nucleation and propagation direction}, 
  author={Feng, Ye and Hai, Lu},
   journal={Journal of the Mechanics and Physics of Solids},
   volume={196}, 
   pages={106036},
    year={2025},
     publisher={Elsevier} }

@article{bourdin2000numerical,



  title={Numerical experiments in revisited brittle fracture},



  author={Bourdin, Blaise and Francfort, Gilles A and Marigo, Jean-Jacques},



  journal={Journal of the Mechanics and Physics of Solids},



  volume={48},



  number={4},



  pages={797--826},



  year={2000},



  publisher={Elsevier}



}

@book{braides2002gamma,
  author = {Braides, Andrea},
  title = {{$\Gamma$}-Convergence for Beginners},
  year = {2002},
  publisher = {Oxford University Press},
  address = {Oxford},
  doi = {10.1093/acprof:oso/9780198507840.001.0001}
}

@article{rorentrop2026computationally,
  title={A computationally efficient monolithic arc-length method for phase-field models of fracture based on the Efendiev \& Mielke scheme},
  author={R{\"o}rentrop, Felix and Langenfeld, Kai and Mosler, J{\"o}rn},
  journal={Computer Methods in Applied Mechanics and Engineering},
  volume={458},
  pages={119005},
  year={2026},
  doi={10.1016/j.cma.2026.119005},
  publisher={Elsevier}
}

@article{rorentrop2024time,
  author={R{\"o}rentrop, Felix and Boddin, Samira and Knees, Dorothee and Mosler, J{\"o}rn},
  title={A time-adaptive finite element phase-field model suitable for rate-independent fracture mechanics},
  journal={Computer Methods in Applied Mechanics and Engineering},
  volume={431},
  pages={117240},
  year={2024},
  doi={10.1016/j.cma.2024.117240}
}

@article{storvik2021accelerated,
  title={An accelerated staggered scheme for variational phase-field models of brittle fracture},
  author={Storvik, Erlend and Both, Jakub Wiktor and Sargado, Juan Michael and Nordbotten, Jan Martin and Radu, Florin Adrian},
  journal={Computer Methods in Applied Mechanics and Engineering},
  volume={381},
  pages={113822},
  year={2021},
  publisher={Elsevier}
}

@article{lampron2021efficient,
  title={An efficient and robust monolithic approach to phase-field quasi-static brittle fracture using a modified Newton method},
  author={Lampron, Olivier and Therriault, Daniel and Levesque, Martin},
  journal={Computer Methods in Applied Mechanics and Engineering},
  volume={386},
  pages={114091},
  year={2021},
  publisher={Elsevier}
}

@article{farrell2017linear,
  title={Linear and nonlinear solvers for variational phase-field models of brittle fracture},
  author={Farrell, Patrick and Maurini, Corrado},
  journal={International Journal for Numerical Methods in Engineering},
  volume={109},
  number={5},
  pages={648--667},
  year={2017},
  publisher={Wiley Online Library}
}

@article{jiang2026schur,
  title={Schur-complement-enhanced iterative solvers for coupled problems and phase field fracture models},
  author={Jiang, Juhan and Zhao, Yonghui and Lyu, Bing and Shen, Yongxing},
  journal={Computer Methods in Applied Mechanics and Engineering},
  volume={449},
  pages={118509},
  year={2026},
  publisher={Elsevier},
  doi={10.1016/j.cma.2025.118509}
}

@article{Braides1995,
author = {Bouchitté, Guy and Braides, Andrea and Buttazzo, Giuseppe},
journal = {Journal für die reine und angewandte Mathematik},
pages = {1-18},
title = {Relaxation results for some free discontinuity problems},
doi = {10.1515/crll.1995.458.1},
url = {http://eudml.org/doc/153676},
volume = {458},
year = {1995},
}

@article{rodella2026sharp,
title = {Sharp-interface cohesive fracture models with consistent bulk energies: Numerical investigations},
journal = {Journal of the Mechanics and Physics of Solids},
volume = {211},
pages = {106543},
year = {2026},
issn = {0022-5096},
author = {Rodella, Andrea and Marigo, Jean-Jacques and Maurini, Corrado and Vidoli, Stefano},
doi = {10.1016/j.jmps.2026.106543}
}

@book{bertsekas1999nonlinear,
  title={Nonlinear Programming},
  author={Bertsekas, Dimitri P.},
  year={1999},
  edition={2nd},
  publisher={Athena Scientific},
  address={Belmont, MA},
  isbn={978-1-886529-00-5}
}

@article{FREDDI20101154,



title = "Regularized variational theories of fracture: A unified approach",



journal = "Journal of the Mechanics and Physics of Solids",



volume = "58",



number = "8",



pages = "1154 - 1174",



year = "2010",



issn = "0022-5096",



author = "Francesco Freddi and Gianni Royer-Carfagni",

doi = "10.1016/j.jmps.2010.02.010"



}

@article{verhoosel2013phase,



  title={A phase-field model for cohesive fracture},



  author={Verhoosel, Clemens V and de Borst, Ren{\'e}},



  journal={International Journal for Numerical Methods in Engineering},



  volume={96},



  number={1},



  pages={43--62},



  year={2013},



  publisher={Wiley Online Library}



}

@article{CONTI20161033,



title = "Phase field approximation of cohesive fracture models",



journal = "Annales de l'Institut Henri Poincaré C, Analyse non linéaire",



volume = "33",



number = "4",



pages = "1033 - 1067",



year = "2016",



issn = "0294-1449",



author = "S. Conti and M. Focardi and F. Iurlano"



}

@article{feng2021endowing,
  title={Endowing explicit cohesive laws to the phase-field fracture theory},
  author={Feng, Ye and Fan, Jiadi and Li, Jie},
  journal={Journal of the Mechanics and Physics of Solids},
  volume={152},
  pages={104464},
  year={2021},
  publisher={Elsevier}
}

@article{feng2023unified,



  title={A unified regularized variational cohesive fracture theory with directional energy decomposition},



  author={Feng, Ye and Li, Jie},



  journal={International Journal of Engineering Science},



  volume={182},



  pages={103773},



  year={2023},



  publisher={Elsevier}



}

@article{chen2025high,
  title={A high-accuracy phase field model for complex cohesive law},
  author={Chen, Pengcheng and Feng, Ye and Ma, Yu E},
  journal={Engineering Fracture Mechanics},
  volume={315},
  pages={110765},
  year={2025},
  publisher={Elsevier}
}

@article{zhang2005unified,
  title={Unified tensile fracture criterion},
  author={Zhang, ZF and Eckert, J},
  journal={Physical review letters},
  volume={94},
  number={9},
  pages={094301},
  year={2005},
  doi={10.1103/PhysRevLett.94.094301},
  publisher={APS}
}

@article{huber2024length,
  title={Length Scale Insensitive Phase-Field Fracture Methodology for Brittle and Ductile Materials},
  author={Huber, William and Zaeem, Mohsen Asle},
  journal={Theoretical and Applied Fracture Mechanics},
  volume={133},
  pages={104500},
  year={2024},
  doi={10.1016/j.tafmec.2024.104500},
  publisher={Elsevier}
}

@article{griffith1921vi,



  title={The phenomena of rupture and flow in solids},



  author={Griffith, Alan Arnold},



  journal={Philosophical transactions of the royal society of london. Series A, containing papers of a mathematical or physical character},



  volume={221},



  number={582-593},



  pages={163--198},



  year={1921},



  publisher={The royal society London}



}

@article{ren2024physically,
  title={Physically consistent nonlocal macro--meso-scale damage model for quasi-brittle materials: A unified multiscale perspective},
  author={Ren, Yudong and Lu, Guangda and Chen, Jianbing},
  journal={International Journal of Solids and Structures},
  volume={293},
  pages={112738},
  year={2024},
  publisher={Elsevier}
}

@article{schmidt2009eigenfracture,
  author={Schmidt, Bernd and Fraternali, Fernando and Ortiz, Michael},
  title={Eigenfracture: An eigendeformation approach to variational fracture},
  journal={Multiscale Modeling \& Simulation},
  volume={7},
  number={3},
  pages={1237--1266},
  year={2009},
  doi={10.1137/080712568}
}

@article{feng2025JMPS,
title = {Phase-field cohesive fracture models with strong displacement discontinuities},
journal = {Journal of the Mechanics and Physics of Solids},
volume = {208},
pages = {106479},
year = {2026},
doi = {10.1016/j.jmps.2025.106479},
issn = {0022-5096},
author = {Ye Feng and Lu Hai}
}

@article{park2009unified,
  title={A unified potential-based cohesive model of mixed-mode fracture},
  author={Park, Kyoungsoo and Paulino, Glaucio H. and Roesler, Jeffery R.},
  journal={Journal of the Mechanics and Physics of Solids},
  volume={57},
  number={6},
  pages={891--908},
  year={2009},
  doi={10.1016/j.jmps.2008.10.003}
}

@article{vicentini2025variational,
  title={Variational phase-field modeling of cohesive fracture with flexibly tunable strength surface},
  author={Vicentini, Francesco and Heinzmann, Jonas and Carrara, Pietro and De Lorenzis, Laura},
  journal={Journal of the Mechanics and Physics of Solids},
  pages={106424},
  year={2026},
  volume={207},
  doi={10.1016/j.jmps.2025.106424},
  publisher={Elsevier}
}

@article{armero1996analysis,
  title={An analysis of strong discontinuities in multiplicative finite strain plasticity and their relation with the numerical simulation of strain localization in solids},
  author={Armero, F. and Garikipati, K.},
  journal={International Journal of Solids and Structures},
  volume={33},
  number={20--22},
  pages={2863--2885},
  year={1996},
  doi={10.1016/0020-7683(95)00257-X},
  publisher={Elsevier}
}

@article{oliver2003strong,
  title={On the strong discontinuity approach in finite deformation settings},
  author={Oliver, J. and Huespe, A. E. and Pulido, M. D. G. and Samaniego, E.},
  journal={International Journal for Numerical Methods in Engineering},
  volume={56},
  number={7},
  pages={1051--1082},
  year={2003},
  doi={10.1002/nme.607},
  publisher={Wiley}
}

@article{ortiz1999finite,
  title={Finite-deformation irreversible cohesive elements for three-dimensional crack-propagation analysis},
  author={Ortiz, Michael and Pandolfi, Anna},
  journal={International Journal for Numerical Methods in Engineering},
  volume={44},
  number={9},
  pages={1267--1282},
  year={1999},
  doi={10.1002/(SICI)1097-0207(19990330)44:9<1267::AID-NME486>3.0.CO;2-7},
  publisher={Wiley}
}

@article{tian2025nonlinearity,
  author = {Tian, Fucheng and Gong, Jian Ping},
  title = {Nonlinearity tunes crack dynamics in soft materials},
  journal = {Journal of the Mechanics and Physics of Solids},
  volume = {202},
  pages = {106191},
  year = {2025},
  doi = {10.1016/j.jmps.2025.106191}
}

@article{knowles1977finite,
  author = {Knowles, James K.},
  title = {The finite anti-plane shear field near the tip of a crack for a class of incompressible elastic solids},
  journal = {International Journal of Fracture},
  volume = {13},
  number = {5},
  pages = {611--639},
  year = {1977},
  doi = {10.1007/BF00017296}
}

@article{long2011finite,
  author = {Long, Rong and Krishnan, Venkat R. and Hui, Chung-Yuen},
  title = {Finite strain analysis of crack tip fields in incompressible hyperelastic solids loaded in plane stress},
  journal = {Journal of the Mechanics and Physics of Solids},
  volume = {59},
  number = {3},
  pages = {672--695},
  year = {2011},
  doi = {10.1016/j.jmps.2010.12.005}
}

@article{lee2024size,
  author = {Lee, Jaehee and Lee, Jeongun and Yun, Seounghee and Kim, Sanha and Lee, Howon and Chester, Shawn A. and Cho, Hansohl},
  title = {Size-dependent fracture in elastomers: Experiments and continuum modeling},
  journal = {Physical Review Materials},
  volume = {8},
  pages = {115602},
  year = {2024},
  doi = {10.1103/PhysRevMaterials.8.115602}
}

@book{brent1973algorithms,
  author = {Brent, Richard P.},
  title = {Algorithms for Minimization without Derivatives},
  publisher = {Prentice-Hall},
  address = {Englewood Cliffs, New Jersey},
  year = {1973}
}

@article{feng2026multilevel,
  author = {Feng, Ye and Ren, Yu-Dong and Lu, Guang-Da},
  title = {Multi-level adaptive mesh refinement with non-local energy correction for sharp-diffusive phase-field fracture models},
  journal = {Computer Methods in Applied Mechanics and Engineering},
  volume = {462},
  pages = {119311},
  year = {2026},
  doi = {10.1016/j.cma.2026.119311}
}

@article{conti2024vectorial,
  author = {Conti, Sergio and Focardi, Matteo and Iurlano, Flaviana},
  title = {Phase-Field Approximation of a Vectorial, Geometrically Nonlinear Cohesive Fracture Energy},
  journal = {Archive for Rational Mechanics and Analysis},
  volume = {248},
  pages = {21},
  year = {2024},
  doi = {10.1007/s00205-024-01962-4}
}

@article{tian2022mixed,
  author = {Tian, Fucheng and Zeng, Jun and Zhang, Mengnan and Li, Liangbin},
  title = {Mixed displacement--pressure-phase field framework for finite strain fracture of nearly incompressible hyperelastic materials},
  journal = {Computer Methods in Applied Mechanics and Engineering},
  volume = {394},
  pages = {114933},
  year = {2022},
  doi = {10.1016/j.cma.2022.114933}
}

@article{tian2025large,
  author = {Tian, Fucheng and Ji, Jiale and Chen, Shuyu and Li, Liangbin},
  title = {Large deformation fracture in soft materials: Phase-field modeling and applications},
  journal = {Acta Polymerica Sinica},
  volume = {56},
  number = {2},
  pages = {179--199},
  year = {2025},
  doi = {10.11777/j.issn1000-3304.2024.24224}
}

@article{brighenti2011buckling,
  author = {Brighenti, Roberto and Carpinteri, Andrea},
  title = {Buckling and fracture behaviour of cracked thin plates under shear loading},
  journal = {Materials \& Design},
  volume = {32},
  number = {3},
  pages = {1347--1355},
  year = {2011},
  doi = {10.1016/j.matdes.2010.09.018}
}

@article{buehler2003hyperelasticity,
  author = {Buehler, Markus J. and Abraham, Farid F. and Gao, Huajian},
  title = {Hyperelasticity governs dynamic fracture at a critical length scale},
  journal = {Nature},
  volume = {426},
  pages = {141--146},
  year = {2003},
  doi = {10.1038/nature02096}
}

@article{wang2023tensile,
  author = {Wang, Meng and Shi, Songlin and Fineberg, Jay},
  title = {Tensile cracks can shatter classical speed limits},
  journal = {Science},
  volume = {381},
  number = {6656},
  pages = {415--419},
  year = {2023},
  doi = {10.1126/science.adg7693}
}

@article{simo1993analysis,
  author = {Simo, J. C. and Oliver, J. and Armero, F.},
  title = {An analysis of strong discontinuities induced by strain-softening in rate-independent inelastic solids},
  journal = {Computational Mechanics},
  volume = {12},
  number = {5},
  pages = {277--296},
  year = {1993},
  doi = {10.1007/BF00372173}
}

@article{alessi2026gamma,
  author = {Alessi, Roberto and Colasanto, Francesco and Focardi, Matteo},
  title = {Phase-Field Modeling of Cohesive Fracture. Part {I}: {$\Gamma$}-Convergence Results},
  journal = {SIAM Journal on Mathematical Analysis},
  volume = {58},
  number = {1},
  pages = {513--546},
  year = {2026},
  doi = {10.1137/25M1759781}
}

@article{bourdin2026strength,
  author = {Bourdin, Blaise and Maurini, Corrado},
  title = {Strength-degradation phase-field regularization of cohesive fracture: the antiplane case},
  journal = {arXiv preprint arXiv:2607.29157},
  year = {2026},
  doi = {10.48550/arXiv.2607.29157}
}

@article{zhao2025staggered,
  author = {Zhao, Yonghui and Jiang, Juhan and Lyu, Bing and Shen, Yongxing},
  title = {Staggered algorithms for coupled problems: Convergence analysis and application to phase field modeling of thermal cracking},
  journal = {Computer Methods in Applied Mechanics and Engineering},
  volume = {446},
  pages = {118238},
  year = {2025},
  doi = {10.1016/j.cma.2025.118238}
}

@article{cheng2025global,
  author = {Cheng, FengYu and Yu, Hao and Wang, Quan and Huang, HanWei and Xu, WenLong and Wu, HengAn},
  title = {Global-local adaptive meshing method for phase-field fracture modeling},
  journal = {Computer Methods in Applied Mechanics and Engineering},
  volume = {438},
  pages = {117846},
  year = {2025},
  doi = {10.1016/j.cma.2025.117846}
}

@article{cheng2026selfcorrecting,
  author = {Cheng, FengYu and Yu, Hao and Lyu, ChengSi and Li, Bo and Wang, Quan and Poh, Leong Hien and Wu, HengAn},
  title = {A Self-Correcting Adaptive Remeshing Strategy for Three-Dimensional Phase-Field Fracture Simulation},
  journal = {International Journal for Numerical Methods in Engineering},
  volume = {127},
  number = {16},
  pages = {e70402},
  year = {2026},
  doi = {10.1002/nme.70402}
}

@article{verhoosel2009dissipation,
  author = {Verhoosel, C. V. and Remmers, J. J. C. and Guti{\'e}rrez, M. A.},
  title = {A dissipation-based arc-length method for robust simulation of brittle and ductile failure},
  journal = {International Journal for Numerical Methods in Engineering},
  year = {2009},
  volume = {77},
  number = {9},
  pages = {1290--1321},
  doi = {10.1002/nme.2447}
}

@article{bharali2022robust,
  author = {Bharali, Ritukesh and Goswami, Somdatta and Anitescu, Cosmin and Rabczuk, Timon},
  title = {A robust monolithic solver for phase-field fracture integrated with fracture energy based arc-length method and under-relaxation},
  journal = {Computer Methods in Applied Mechanics and Engineering},
  year = {2022},
  volume = {394},
  pages = {114927},
  doi = {10.1016/j.cma.2022.114927}
}

@article{rice1968path,
  author = {Rice, James R.},
  title = {A Path Independent Integral and the Approximate Analysis of Strain Concentration by Notches and Cracks},
  journal = {Journal of Applied Mechanics},
  year = {1968},
  volume = {35},
  number = {2},
  pages = {379--386},
  doi = {10.1115/1.3601206}
}

@article{moran1987crack,
  author = {Moran, Brian and Shih, C. F.},
  title = {Crack tip and associated domain integrals from momentum and energy balance},
  journal = {Engineering Fracture Mechanics},
  year = {1987},
  volume = {27},
  number = {6},
  pages = {615--642},
  doi = {10.1016/0013-7944(87)90155-X}
}

@article{tang2025tough,
  author = {Tang, Jingyuan and Chen, Xi and Liu, Fengkai and Zeng, Liangsong and Suo, Zhigang and Tang, Jingda},
  title = {Why are soft collagenous tissues so tough?},
  journal = {Science Advances},
  year = {2025},
  volume = {11},
  number = {25},
  pages = {eadw0808},
  doi = {10.1126/sciadv.adw0808}
}

@article{chen2017instability,
  author = {Chen, Chih-Hung and Bouchbinder, Eran and Karma, Alain},
  title = {Instability in dynamic fracture and the failure of the classical theory of cracks},
  journal = {Nature Physics},
  volume = {13},
  pages = {1186--1190},
  year = {2017},
  doi = {10.1038/nphys4237}
}

@article{riccobelli2024elastic,
  author = {Riccobelli, D. and Ciarletta, P. and Vitale, G. and Maurini, C. and Truskinovsky, L.},
  title = {Elastic Instability behind Brittle Fracture},
  journal = {Physical Review Letters},
  volume = {132},
  pages = {248202},
  year = {2024},
  doi = {10.1103/PhysRevLett.132.248202}
}

@article{bouchbinder2014dynamics,
  author = {Bouchbinder, Eran and Goldman, Tamar and Fineberg, Jay},
  title = {The dynamics of rapid fracture: instabilities, nonlinearities and length scales},
  journal = {Reports on Progress in Physics},
  volume = {77},
  number = {4},
  pages = {046501},
  year = {2014},
  doi = {10.1088/0034-4885/77/4/046501}
}

@book{dalmaso1993introduction,
  author = {Dal Maso, Gianni},
  title = {An Introduction to {$\Gamma$}-Convergence},
  series = {Progress in Nonlinear Differential Equations and Their Applications},
  volume = {8},
  year = {1993},
  publisher = {Birkh{\"a}user},
  address = {Boston},
  doi = {10.1007/978-1-4612-0327-8}
}
\end{document}